\documentclass[11pt]{article}

\usepackage[T1]{fontenc}
\usepackage[utf8]{inputenc}
\usepackage{lmodern}
\usepackage{textcomp}
\ifdefined\DeclareUnicodeCharacter
  \DeclareUnicodeCharacter{00A0}{\ }  
  \DeclareUnicodeCharacter{2010}{-}    
  \DeclareUnicodeCharacter{2011}{-}    
  \DeclareUnicodeCharacter{2012}{--}   
  \DeclareUnicodeCharacter{2013}{--}   
  \DeclareUnicodeCharacter{2014}{---}  
  \DeclareUnicodeCharacter{2018}{`}    
  \DeclareUnicodeCharacter{2019}{'}    
  \DeclareUnicodeCharacter{201C}{\textquotedblleft} 
  \DeclareUnicodeCharacter{201D}{\textquotedblright} 
  \DeclareUnicodeCharacter{2212}{-}    
\fi

\usepackage[margin=1in]{geometry}
\usepackage{amsmath, amssymb, amsthm, mathtools}
\usepackage{bbm}
\usepackage{enumitem}
\usepackage[numbers,sort&compress]{natbib}
\usepackage[dvipsnames]{xcolor}
\usepackage{hyperref}
\usepackage{xurl}
\usepackage[nameinlink,capitalise,noabbrev]{cleveref}
\usepackage{subfiles}
\usepackage{algorithm}
\usepackage{algpseudocode}
\usepackage{microtype}
\usepackage{graphicx}
\usepackage{subcaption}
\usepackage{tikz}
\usetikzlibrary{arrows.meta,calc,decorations.pathreplacing,patterns}


\hypersetup{
  colorlinks=true,
  linkcolor=NavyBlue,
  citecolor=NavyBlue,
  urlcolor=NavyBlue
}

\title{Autodeleveraging: Impossibilities and Optimization}
\author{Tarun Chitra\\ 
        Gauntlet \\
        \texttt{\small tarun@gauntlet.xyz} 
       }
\date{\today}
\newcommand{\ones}{\mathbf 1}
\newcommand{\reals}{{\mbox{\bf R}}}

\newcommand{\naturals}{{\mbox{\bf N}}}

\newcommand{\Expect}{\mathop{\bf E{}}}

\newcommand{\Prob}{\mathop{\bf Prob}}

\newcommand{\argmin}{\mathop{\rm argmin}}

\newcommand{\cf}{{\it cf.}}
\newcommand{\eg}{{\it e.g.}}
\newcommand{\ie}{{\it i.e.,}}

\newcommand{\BEAS}{\begin{eqnarray*}}
\newcommand{\EEAS}{\end{eqnarray*}}
\newcommand{\BEA}{\begin{eqnarray}}
\newcommand{\EEA}{\end{eqnarray}}
\newcommand{\BEQ}{\begin{equation}}
\newcommand{\EEQ}{\end{equation}}
\newcommand{\BIT}{\begin{itemize}}
\newcommand{\EIT}{\end{itemize}}



\newcommand{\iparagraph}[1]{\paragraph{\normalfont\itshape #1}}

\newcommand{\filecite}[1]{}

\newcommand{\asympp}{\asymp_p}

\theoremstyle{plain}
\newtheorem{theorem}{Theorem}[section]
\newtheorem{lemma}[theorem]{Lemma}
\newtheorem{corollary}[theorem]{Corollary}
\newtheorem{proposition}[theorem]{Proposition}
\theoremstyle{definition}
\newtheorem{definition}[theorem]{Definition}
\newtheorem{assumption}[theorem]{Assumption}
\newtheorem{example}[theorem]{Example}
\newtheorem{remark}[theorem]{Remark}

\begin{document}
\maketitle

\begin{abstract}
Autodeleveraging (ADL) is a last‑resort loss socialization mechanism for perpetual futures venues.
It is triggered when solvency-preserving liquidations fail.
Despite the dominance of perpetual futures in the crypto derivatives market, with over \$60 trillion of volume in 2024, there has been no formal study of ADL.
In this paper, we provide the first rigorous model of ADL.
We prove that ADL mechanisms face a fundamental \emph{trilemma}: no policy can simultaneously satisfy exchange \emph{solvency}, \emph{revenue}, and \emph{fairness} to traders.
This impossibility theorem implies that as participation scales, a novel form of \emph{moral hazard} grows asymptotically, rendering `zero-loss' socialization impossible.
On the positive side, we show that three classes of ADL mechanisms can optimally navigate this trilemma to provide fairness, robustness to price shocks, and maximal exchange revenue.
We analyze these mechanisms on the Hyperliquid dataset from October 10, 2025, when ADL was used repeatedly to close \$2.1 billion of positions in 12 minutes.
By comparing production ADL to transparent benchmark allocations, we find that Hyperliquid's production algorithm overshot the minimum trader profit haircut required to cover the shortfall.
Our methodology suggests the excess profits lost by profitable traders is between \$45.0M and \$51.7M.
In terms of the positions liquidated, this corresponds to roughly \$653.6M of positions being closed.
This comparison also suggests that Binance overutilized ADL far more than Hyperliquid.
Our results show both theoretically and empirically that optimized ADL mechanisms can dramatically reduce losses of trader profitability while maintaining exchange solvency.
\end{abstract}

\section{Introduction}
Perpetual futures (or simply, perpetuals) are by far the most popular crypto derivatives contract.
These contracts allow for duration-independent hedging of cryptocurrencies, such as Bitcoin, Ethereum or Solana.
That is, unlike traditional future contracts, perpetuals are not subject to a fixed expiration date.
Instead, perpetuals are similar to \emph{contracts for difference}, where market participants repeatedly send each other payments depending on how the spot price trades relative to the futures price.
For instance, if the spot price is higher than the futures price, then traders with short futures positions pay the traders with long futures until the prices are equalized.
We note that while Robert Shiller's 1993 proposal for perpetual futures provides the genesis of this idea, the BitMEX exchange's 2016 launch of the XBTUSD perpetual swap became the first live implementation~\cite{shiller1993measuring,Soska2021BitMEX,Hayes2025AdaptOrDie}.

Perpetual futures have a higher notional trading volume than spot cryptocurrency volume.
This is akin to index futures for US stocks, where there is far more volume traded in the futures contracts than the underlying spot assets, with the S\&P 500 futures index future trading roughly ten times more volume than the spot asset~\cite{CMEGroup2024EqualWeight}.
In 2024, centralized crypto exchanges such as Binance and Bybit processed nearly \$58.5 trillion (notional) in perpetuals trades~\citep{CoinGecko2025Perps}.
On the other hand, centralized exchanges processed nearly \$17.4 trillion of spot cryptocurrency trades~\citep{CoinGecko2025Q1}, which is a 3.3x ratio of perpetuals to spot volume.

The main reason for the elevated usage of perpetual futures, much like their index futures cousins, is due to the high level of leverage they offer.
In particular, the capital requirements for trading perpetual futures are much lower than spot trading.
Moreover, the continuous nature of perpetual futures allows for easier margin management than duration-based futures.
This is illustrated via the large gap between the maximum leverage offered by perpetuals exchanges (up to 125x~\cite{Binance2025CollateralLeverageUpdate}) versus the index futures on the Chicago Mercantile Exchange (around 10-15x leverage for index futures~\cite{CMEGroup2025SP500Margins}).

\paragraph{Centralized vs. Decentralized Exchanges.}
Historically, the majority of perpetuals trading volume has been concentrated on centralized exchanges (CEXes).
These venues, which include Binance, Bybit, and Coinbase, require users to deposit collateral that is custodied by the exchange.
The exchange is responsible for ensuring that users' positions are solvent and that their collateral is not used to cover losses from other traders.
In practice, this involves traders having to trust the exchange to not use their collateral for other purposes.
In 2022, centralized exchange FTX was found to have utilized customer collateral for other purposes, leading to multi‑billion‑dollar losses~\citep{FTXDebtors2023Report}.

In response, there has been an increasing trend in the usage of decentralized exchanges (DEXes) for trading perpetual futures.
These exchanges provide far more transparency and auditability into the mechanics of the exchange.
In particular, users can always verify what their collateral is being used for and what positions it is collateralizing.
Moreover, these exchanges are permissionless and pseudonymous, meaning anyone can trade on these exchanges without explicitly revealing their identity.

The transparent nature of decentralized exchanges does have, however, a downside.
Most positions in decentralized exchanges need to be fully overcollateralized, which generally means traders can take less leverage than they would on a centralized exchange.
Generally speaking, this is because the exchange doesn't know the identity of a trader and can only enforce global collateral invariants (\eg~assets > liabilities) by forcing users to post more collateral than if the exchange knew their identity.
Early decentralized perpetuals exchanges, such as Perpetual Protocol, were unable to offer much more than 5x leverage.

However, in recent years, hybrid models of decentralized exchanges where withdrawals and deposits are permissionless, but certain exchange operations are centralized, have become popular.
The most popular exchange of this form is Hyperliquid, which has helped increase the market share of perpetuals volume on DEXs from roughly 1\% in 2024 to over 8\% in 2025.
These hybrid models allow exchanges to use cryptographic commitments to enforce collateral rights while allowing oracles, markets, and solvency to be controlled by a smaller set of participants.
This allows the exchange to avoid issues faced by fully open decentralized exchanges, where weak collateral can be added and cannot be sold or liquidated successfully during a market crash to ensure solvency.

More concretely, there have been recent incidents on perpetual DEXs that highlight how novel market structures can be exploited or stressed in practice.
These include but are not limited to: Hyperliquid’s XPL pre‑market price spike and the JellyJelly attack~\cite{CryptoTimesXPL2025,ODailyXPL2025,OakResearchJELLY}.
These episodes led to emergency safeguards and community post‑mortems, underscoring the need for robust oracle design, listing policies, and circuit breakers to mitigate manipulation.
Such emergency responses would have been nearly impossible in the fully decentralized setting, demonstrating the value of hybrid models of decentralization in perpetual futures.

\paragraph{Liquidations.}
Despite the empirical prominence of perpetuals markets, there has been little formal study of the stability and robustness of these markets relative to the large body of research into spot cryptocurrency trading.
For instance, the main price stability mechanism used for perpetuals is the \emph{funding rate}, a continuous payment stream between the long and short positions.
There has been some study of how to replicate funding rates via other financial instruments~\cite{AngerisChitra2023PerpsSIAM, he2022fundamentals, ackerer2024perpetual}, but little to no empirical study of how well such approximations work.
Even more surprising is the lack of formal study of how robust perpetual exchanges are under high leverage conditions, despite their usage being driven by high leverage traders.

A crucial component to ensuring the stability of perpetual futures markets with high leverage is~\emph{liquidation}.
This is the process of closing a user's position that is worth less than the cash collateral they have posted to the exchange.
Perpetuals exchanges rely on liquidation mechanisms to ensure that they stay solvent --- that is, the assets they hold are greater than the liabilities they owe to traders.
To the author's knowledge, there has been no formal study of liquidations in perpetuals exchanges, despite there being a large body of work studying liquidations in overcollateralized cryptocurrency lending (\eg~\cite{perez2021liquidations,kao2020analysis,sun2022liquidity,qin2021empirical}).

Liquidations are executed by taking the collateral of a trader and selling it in a manner that minimizes losses that the exchange realizes.
In centralized exchanges, liquidations are usually executed by the exchange itself in its own spot markets.
For instance, suppose that a user uses \$1,000 worth of collateral to open a long position with 10x leverage.
This implies that if the price falls from the initial price when the contract was opened, $p_0$, by roughly 10\%, then the user's position will have a value of -\$1,000, so that user's net position is \$0.
If the price decreases beyond this point, then the user's position will have a negative value, implying that the user owes the exchange (and hence, other traders) capital.
We say that a trader's position is \emph{insolvent} if such a negative net position is realized.

Exchanges utilize liquidations to remove trader positions before they are insolvent.
For instance, in our example, an exchange might deem a trader's position to be liquidatable if the price falls by 9\%, such that the trader's position is closed before insolvency.
Provided that the exchange can execute this position (inclusive of transaction fees and market impact costs) at a price in between $0.91 p_0$ and $0.9 p_0$, the position can be closed profitably.
This profit is usually distributed to other traders in the exchange, added to an insurance fund for future losses, or realized as a profit for the exchange.

Centralized exchanges usually execute liquidations in their own spot markets.
This is because they can effectively guarantee atomicity of liquidations, ensuring that the liquidation transaction does not get front run and have a worse execution price than expected.
In decentralized exchanges, liquidiations are usually executed by third-party actors known as liquidators.
Liquidators can be viewed as traders who have the ability to warehouse the inventory of a collateral position and exit profitably.
From an economic perspective, liquidations need to be profitable for either the exchange or the liquidator in order to help with exchange solvency.

\paragraph{Autodeleveraging.}
In severe market dislocations, it is sometimes possible for liquidations to be so deeply unprofitable that they are unable to be executed by any party.
When this happens, the exchange can reach a state of insolvency, where the assets held are less than the liabilities owed to its traders.
This excess in liabilities is known as the \emph{shortfall}.
In such scenarios, a natural strategy to reduce the risk of total insolvency is to socialize losses by haircutting profitable (or \emph{winning}) traders.
In the worst case, a winning user's position is completely closed to zero, leaving them with a potentially large opportunity cost.
This process is known as \emph{autodeleveraging} (ADL) and is a last-resort measure used by exchanges.

Autodeleveraging is unique in that it algorithmically socializes losses on winners based on a public rule or mechanism that an exchange posts.
This means that all of the position closures and implicit liquidations of insolvent positions are performed in an irreversible and atomic manner.
The closest analogue to autodeleveraging in traditional finance is a central counterparty (CCP) such as a derivatives clearing house. 
In traditional finance, CCPs employ default waterfalls and loss mutualization (including variation‑margin gains haircutting (VMGH)) that share the same core idea of allocating residual losses across non‑defaulting members~\cite{DuffieZhu2011,Pirrong2011,CPMI_IOSCO_2014,ContGhamamiSITG,KubitzaWinnersLosers,Turing2019MagicRelighting}.
CCPs generally have manual, non-algorithmic socialization and disbursement of insurance funds during such events with disputes reconciled via the legal system.
We note, furthermore, that only clearing members bare the losses in CCPs, which leads to a different principal-agent dynamic than perpetuals markets as smaller traders need not bare the final losses. 

Autodeleveraging, to the best of the author's knowledge, was first introduced by the crypto exchange Huobi in 2015~\cite{HuobiADL}.
Prior to the implementation of ADL, exchanges would manually socialize losses, with the exchange serving as the sole CCP.
By 2016, the BitMEX exchange implemented ADL via a formula that ranked trader positions for ADL by their profit times their leverage~\cite{BitMEXADL}.
This formula persists to today on virtually every centralized exchange (including Binance, who implemented the formula in 2019~\cite{BinanceADL}) and on the decentralized exchange Hyperliquid~\cite{HyperliquidDocsLiquidations}.
We note that there are some decentralized exchanges such as Drift~\cite{DriftADLCode} and Paradex~\cite{ParadexADLBlog} that use a different ADL mechanism (the \emph{pro-rata} mechanism defined in~\S\ref{subsec:adl}).
However, well over 95\% of perpetuals trading volume is on exchanges that utilize the BitMEX ranking model.

In practice, ADL is reserved for extreme tail events where routine liquidation mechanisms are insufficient to maintain venue solvency.
There are two primary scenarios that can trigger ADL activation.
First, a sudden and large price movement can push positions far beyond their liquidation thresholds before the exchange's liquidation system can execute the necessary trades.
Second, a temporary but severe loss of executable liquidity in the order book can prevent the exchange from closing positions at prices near their liquidation triggers, even when the price movement itself is not extreme.

To illustrate how such scenarios can lead to insolvency, consider a concrete example.
Suppose that a trader posts \$1{,}000 of collateral to take a 10x leveraged long position with a notional value of \$10{,}000, where the position becomes liquidatable when the price declines by 9\% from its initial value $p_0$.
Under normal market conditions, the exchange would execute a liquidation when the price reaches $0.91\,p_0$, closing the position and recovering the collateral.
However, if an abrupt 12\% price drop occurs between price updates or during a period of low liquidity, the position may be closed at a significantly worse execution price, such as $0.88\,p_0$ rather than the intended $0.91\,p_0$.
In this scenario, the account's equity becomes negative, representing approximately 2\% of its initial collateral value.
When many traders simultaneously experience such execution slippage and the exchange's insurance fund reserves cannot absorb the aggregate resulting losses, the venue triggers ADL.
Under the exchange's predefined ADL rules, the system algorithmically reduces or closes profitable opposing positions until the shortfall is covered and solvency is restored.

Inherently, ADL creates \emph{moral hazard} --- that is, losing traders have an incentive to take on more risk in anticipation of being socialized by winning traders.
The goal of an exchange is to balance the risk of insolvency against haircutting winning trader profits.
Moral hazard is measured by this trade-off --- if it is possible to reduce the risk of insolvency to zero without haircutting winning traders, then the system has no moral hazard.
On the other hand, if it is not possible to reduce the risk of insolvency by any appreciable amount without severely haircutting winning trades, the system has excessive moral hazard.
The job of an exchange operator is to choose an ADL mechanism that tries to induce a low level of moral hazard under most tail event scenarios.
This connects to principal–agent models of moral hazard and linear/robust contracting \cite{Holmstrom1982,Carroll2015RobustLinearContracts,DuttingEtAl2023MultiAgentContracts}, with key differences: many agents act simultaneously, externalities operate through a common solvency constraint, and tail‑risk mitigation substitutes for standard effort observability.

Despite the rarity of ADL, it is certainly not hypothetical: there have been dozens of instances of ADL activations on major venues since 2018~\citep{BinanceADL,BitMEXADL,AevoADL,CoinglassADL}.
The largest episode of ADL usage was on October 10-11, 2025 (UTC), when outsized liquidations exhausted insurance funds and multiple venues simultaneously invoked ADL~\cite{CoinDesk2025LargestLiquidations}.
However, other famous incidents such as the March 2020 ``Black Thursday'' and May 19, 2021 deleveraging event were other prominent days involving ADL usage.
Most of these usages of ADL involved controversies, in that highly profitable but lower risk (\eg~lower leverage) traders found their positions closed out by ADL mechanisms. 
These traders have often threatened to sue exchanges for unfair treatment, with rumors and public speculation surfacing during Oct. 10–11, 2025 (subsequently denied by Wintermute)~\citep{CointelegraphWintermute2025,TheBlockWintermute2025}\footnote{See also prior complaints during May 2020 on Binance futures when profitable shorts were closed via ADL~\citep{CointelegraphBinanceADL2020}}.

These incidents demonstrate a natural trade-off that perpetuals exchanges have to make: they can either aggressively socialize losses to their winners and (potentially) lose the future revenue of these users while preserving solvency or hold the losses due to insolvency on their balance sheet.
In the former case, the exchange effectively chooses to prioritize short-term solvency and potential liquidation profits over long-term value and fees generated by users.
The latter case effectively creates a long-term risk for the exchange's balance sheet, especially if it does not have sufficient funds to cover insolvencies.
This strongly suggests that one can view the choice of how an exchange implements ADL as a trade-off between solvency, moral hazard, and long-term revenue for an exchange.

We note that open interest after October 10, 2025 suggests that Hyperliquid has lost nearly 50\% of open interest, whereas competitors such as Lighter and even Binance have recovered to pre-event levels.\footnote{Source: DeFiLlama perpetuals dashboard, \url{https://defillama.com/perps}, accessed Nov.\ 30, 2025.}
Open interest is generally more expensive to manipulate than trading volume and can serve as a coarse measurement of expected future revenue for an exchange. 
Numerous parties have argued that this open interest loss for Hyperliquid is due to their overly aggressive ADL policy.\footnote{See, \eg~, public commentary from @fiddybps1, @0xReaz, @Eugene\_Bulltime, and @0xLouisT on X (Nov.\ 2025): \url{https://x.com/fiddybps1/status/1978750722901762321}, \url{https://x.com/0xReaz/status/1986486213101166599}, \url{https://x.com/Eugene_Bulltime/status/1994380900067582182}, \url{https://x.com/0xLouisT/status/1990815109787058654}.}

\paragraph{This Paper.}
In this paper, we provide (to the author's best knowledge) the first formal mathematical model of autodeleveraging.
The main model of the paper is venue‑agnostic: it abstracts away pricing (\eg~order book versus automated market maker microstructure) and focuses on how positions are opened and closed.
As a warning to the reader: this paper is more verbose than necessary.
The author aims to provide a self-contained, pedagogical introduction to perpetual futures and liquidations, as the current literature and open source documentary does not cover liquidations and ADL in any manner comprehensive enough for formal study.
Throughout the paper, there will be numerous simple numerical examples to illustrate the basic concepts and mechanics of perpetual futures.

The first section of the paper (\S\ref{sec:background}) provides definitions and examples of the basic objects --- \emph{positions} --- and basic operations that one can perform on them.
In the process of defining the set of operations on positions, we will naturally need to define the economics of positions using funding rates and arbitrage.
These definitions will hopefully provide answers to the questions of why one opens a position, how much one has to spend on keeping a position open, and how positions on different exchanges achieve price synchronization.
From these economic properties, we will naturally define solvency for both individual traders and of the exchange itself.

Liquidations and ADL then naturally follow as operations that exist to try to use economic means to enforce solvency constraints with high probability.
We will define the set of ADL mechanisms in a sufficiently broad manner to allow for venue-specific ADL, this allowing us to compare Binance, Hyperliquid, Drift, and other exchanges.
Our decomposition of the ADL strategy space into~\emph{severity} (\ie~dollar amount impacted by ADL) and \emph{haircuts} (\ie~per-trader percentage liquidated) allows us to separate exchange solvency and trader profitability into two separate optimization problems.
This will be key to simplifying our results and constructing practical algorithms.

Given a perpetuals exchange, defined in the combinatorial manner of~\S\ref{sec:background}, we next focus on notions of risk that are endemic to perpetuals exchanges in~\S\ref{sec:risk-prelim}.
We construct a set of risk measurements that, while specific to perpetuals exchanges, are connected to classical risk metrics such as Value-at-Risk.
Subsequently, we describe four facets of perpetuals exchange risk that we will analyze: extreme value analysis, Schur ordering, fairness, and moral hazard.

Extreme value analysis is the classical actuarial estimation of tail probabilities and rare events.
We provide a simple primer to make the text self-contained and our notation consistent.
The Schur ordering provides a way of partially ordering risk measurements, even if they are not fully comparable.
This allows us to compare ADL mechanisms across venues, despite the fact that the outcomes might not have explicitly comparable risks in every situation.
Algorithmic fairness measures (axiomatic and optimization-based) are subsequently introduced to concretely answer the question of whether ADL is more fair to certain traders than others.
Finally, we describe the general moral hazard problem and specialize it to the perpetuals exchange setting.
In this setting, it corresponds to whether the exchange operator has an incentive to take on more risk than necessary because they can socialize worst-case losses on their best traders.

With these risk tools in tow, we are now prepared to describe the main results of the paper.
We have five classes of results in this paper:
\begin{enumerate}
  \item \emph{Negative}: Moral hazard is strictly increasing in the size of the exchange.
  \item \emph{Fairness}: The Pro-Rata ADL mechanism is the most fair mechanism
  \item \emph{Robustness}: Given a risk model $g$ for historical user behavior, there is a unique ADL mechanism maximally robust to price shocks 
  \item \emph{ADL as Stackelberg Game}: Multiple price shocks causing ADL to be used repeatedly are a Stackelberg game that can be studied rigorously. 
  \item \emph{Empirical}: Using Hyperliquid's data from October 10, 2025, where ADL was used over 40 times in a 10-minute period, we evaluate the performance of different ADL mechanisms
\end{enumerate}

\iparagraph{Negative Results.}
Our negative results focus on quantifying limits to how much ADL can actually reduce insolvency without inuring traders with disproportionate losses.
We explicitly ask the question of how the loss of profit due to ADL to the trader with the highest PNL scales as the number of positions $n$ grows.
We quantify this by looking at ratios of the maximum winner's profit to the total shortfall (over all traders) and to the worst trader's shortfall.
These ratios can be viewed as analogues of the Value-at-Risk and Expected Shortfall metrics in traditional risk modeling, as we show in Appendix~\ref{app:proofs}.
The main difference is that while Value-at-Risk and Expected Shortfall measure the total expected loss (over all traders) in tail events, our metrics only measure the expected loss to the most winning trader.
Moreover, Appendix~\ref{app:classical-risk-measures} shows that classical VaR/ES rankings echo the same conclusion: queue-based ADL always delivers the lowest top-winner survival even under these traditional risk lenses.

When these ratios, as constructed, are bounded below by a constant, then the winner can expect to always retain at least a constant fraction of their earnings after an ADL event.
However, if these ratios decay to zero, then the winner has no guarantee of retaining any of their profits after an ADL event.
We show that under some mild technical assumptions, these ratios decay as $O(\frac{b_n}{n})$, where $b_n$ is the \emph{extreme value scale} (EV scale) of profits (see~\S\ref{sec:risk-prelim}).
For most well-behaved probability distributions over the set of profits, $b_n = o(n)$, so the ratios decay to zero.
This result establishes a fundamental limit of these markets: moral hazard scales with participation.
It suggests that as crypto exchanges grow, ADL must become either more frequent or more severe for top traders; they cannot ``grow out'' of the problem.

\iparagraph{Fairness Results.}
As the negative results show that no ADL mechanism can scale to preserve winner's profit, a natural follow-up question is, what is the most fair way of socializing profitable users?
We consider two notions of fairness in~\S\ref{sec:fairness}: axiomatic and optimization-based.
They are both common in the algorithmic fairness literature within machine learning~\cite{KleinbergEtAl2018FairnessImpossibility,BarocasHardtNarayanan2019FairnessBook,Mehrabi2021SurveyFairness,DworkEtAl2012FairnessAwareness,Agarwal2018ReductionsFairness,HardtPriceSrebro2016EqualityOfOpportunity}.
We demonstrate that the pro-rata ADL rule (used by Drift~\cite{DriftADLCode} and Paradex~\cite{ParadexADLBlog}) is the unique fair ADL rule in \emph{both} the axiomatic and optimization senses~\citep{johnson2023concave}.

Axiomatic fairness involves defining rules or invariants that a particular ADL mechanism must satisfy to be deemed fair.
The rules we analyze for fairness are monotonicity, scale invariance and Sybil resistance.
Monotonicity (stable ordering) requires that larger pre-ADL haircutable endowment remains larger post-ADL.
Scale invariance states that if everyone scales positions and deficits by an equal factor, relative ADL losses do not change.
Sybil resistance (split invariance) states that splitting one economic position across accounts does not change its aggregate ADL seizure.

We show that under mild smoothness conditions, the pro-rata ADL rule is the unique rule to satisfy these axioms.
In stark contrast, queue-based mechanisms can violate monotonicity: the greedy waterfall allocation can fully wipe out the top-ranked winner while leaving a nearly identical winner untouched (Appendix~\ref{app:capped-pro-rata}, Proposition~\ref{prop:queue-monotonicity-failure}).
Queue Sybil resistance depends on how accounts are ranked (Appendix~\ref{app:capped-pro-rata}, Proposition~\ref{prop:queue-sybil-score-preserving}): absolute scores (\eg~total PNL) are Sybil vulnerable, while ratio scores such as Hyperliquid's documented index are Sybil resistant when proportional splitting does not create score ties~\cite{HyperliquidDocsADL,Doug2025SybilTweet}.
In our model and data, the queue-based mechanisms used by Binance and Hyperliquid produce extremely low Profit-to-Solvency Ratios (PTSR), imposing disproportionately large losses on the top winner.
This finding contradicts the common intuition that queues are orderly; instead, they concentrate losses on the most successful traders.

Optimization-based fairness involves showing that for a large class of individual trader utility functions, a particular mechanism maximizes social welfare (\ie~the total profit).
This is more of a system-wide fairness notion, as it shows that the system as a whole is fair, even though there might be individuals who were punished more by an ADL mechanism.
Unlike the axiomatic rules, which have to hold for every trader exactly, the optimization-based fairness rules can be thought of as `fuzzier' and allowing for more slack.
Despite this, we still show that the pro-rata ADL rule is the unique rule to maximize social welfare for a large class of concave utility functions.

\iparagraph{Robustness Results.}
A natural question, beyond fairness, is whether an ADL rule is \emph{robust} to subsequent price shocks.
The main question here is whether the choices of severities (\eg~total amount to socialize) and haircuts (\eg~per-account socialization) increases the chance of insolvency should another price shock occur.
In~\S\ref{sec:glpr}, we consider a new threat model where an ADL mechanism is graded based on how well it reduces to total shortfall from two price shocks.
The model involves a price shock occurring that triggers ADL, an ADL policy being executed, and then another price shock occurring.

We show that a natural way to model total shortfall is via a \emph{risk model} $g$ that maps a user's amount of leverage used to a risk score.
A higher risk score implies that a trader's leverage level places a higher burden on the rest of the exchange's solvency, especially in the repeated price shock scenario.
We first show that the risk model is effectively a way to quantify the total shortfall.
Given a stochastic process for price shocks, one can construct an optimal risk model (using convex duality) $g^*$ that maximizes the total solvency gained per dollar of winner position that is haircut by an ADL policy.
With the optimal risk model in hand, we construct a linear-time algorithm to choose haircuts for each trader that are weighted by this model.
We then show that the resulting ADL policy derived from these weights has lowest expected total shortfall and is the optimal ADL policy (in the Schur ordering sense).

\iparagraph{ADL as a Stackelberg Game.}
A direct extension of the robustness threat model is to consider a sequential ADL problem, where for multiple rounds, a price shock is received, the exchange chooses an ADL policy, and then losses are realized.
This model might at first glance seem to only have theoretical value since ADL is itself a rare event.
However, on October 10, 2025 there were multiple ADL events that occurred over a short time window on both Hyperliquid and Binance.
While the account-level data for Binance's usage of ADL is not public, we do have this data for Hyperliquid.
Starting at 10:21:00 UTC, there were at least 40 recorded ADL events on Hyperliquid in a 10-minute period.
This strongly suggests that there is clustering or cascading of ADL events, making the multiple round model necessary to study.

Unlike in the single round models (used for fairness and robustness), the exchange is allowed to adjust their ADL policy after each ADL shock.
One can view this as a sequential game and we argue that it is a \emph{Stackelberg game}.
Stackelberg games have a well-known follower-leader decompsoition, where one user (the leader) plays a strategy first.
Then the follower conditions their strategy on the leader's action, before playing their own move.
In the ADL setting, the exchange is the leader and the traders are the followers.

We first show that viewing ADL as a dynamic, multiple round Stackelberg game leads to sharply different equilibria.
We prove a crucial ``principal-agent'' tension: any strategy an exchange executes that minimizes the time to recover solvency (\ie~number of rounds) necessarily maximizes the loss of exchange long-term revenue.
This trade-off is the economic heart of the paper: the exchange wants to be solvent in few rounds, but doing so destroys the best customers (\eg~the whales).
Intuitively, this occurs because any strategy that quickly recovers solvency autodeleverages the largest winners more aggressively.
These winners, however, are also often the most profitable users and generate high long-term fees for the exchange.
By liquidating them aggressively, you're removing the capital that they would use to later create positions and generate fees.
We note that threats to sue exchanges for ADL losses in profit after October 10, 2025 provide indirect empirical evidence to support this theoretical result~\cite{CointelegraphWintermute2025,TheBlockWintermute2025}.

This implies that any optimal ADL policy over multiple rounds also has to include incentive compatibility constraints that balance future exchange revenue, current exchange solvency, and current trader profits.
We construct an explicit algorithm for optimizing a revenue and solvency based objective and show that a simple mirror descent algorithm achieves vanishing regret for this objective.
We conclude with a simple example demonstrating that no static policy that dominate the dynamic policy, especially when exchange long-term revenue is a consideration.

\iparagraph{Empirical Results: Overshoot Diagnostics.}
We conclude by analyzing the October 10, 2025 ADL episode on Hyperliquid under a public-data, partial-observability reconstruction (\S\ref{sec:numerics}).
This event involved roughly \$2.1B of liquidations across a 12-minute cascade.
Using loser-side liquidation equity fields (for deficits) and a two-pass no-ADL counterfactual (for realized wealth removed from winners), we estimate an aggregate loser deficit of approximately \$100.1M and a production overshoot vs needed of approximately \$45.0M in PNL dollars at \(\Delta=0\), with a short-horizon band of approximately \$45.0M--\$51.7M.
Under the same data, a stylized wealth-space queue diagnostic corresponds to approximately \$653.6M in equity dollars.
The ratio between these figures reflects winner overcollateralization (equity/PNL $\approx 6.66\times$ trimmed mean across winners).
\footnote{Hyperliquid's on-chain fill log for October 10, 2025 records 34{,}983 individual ADL executions across 19{,}337 distinct wallets and 162 tickers.
We analyze the per-fill data but aggregate those executions into global time waves (gap 5s) when reporting per-wave totals.}
Finally, we compare production to transparent benchmark allocations that target the per-wave needed budget \(B_t^{\mathrm{needed}}\); these benchmarks can match \(B_t^{\mathrm{needed}}\) closely (subject to capacity caps and contract discreteness), making production overshoot vs needed empirically salient.

\paragraph{Corrections.}
In a previous version of this paper, we defined ADL in terms of equity haircuts.
After extensive public discussion and critique of how production ADL is executed in \emph{contract space} (and how that maps to wealth-space impacts), it became clear that haircuts should apply to PNL (profits) rather than total equity (i.e., principal/collateral is protected and senior to the haircutable PNL endowment)~\cite{Thogiti2025ADL653M}.
This version reflects these updates: we rewrite the ADL setup with an explicit haircut numéraire (protected cash plus a haircutable endowment, specialized to positive PNL), clarify the observation model used for empirical reconstruction, and show that the theoretical orderings are effectively unchanged while empirical conclusions remain directionally the same with revised magnitudes.
This formulation matches the theoretical haircutable endowment definition in~\S\ref{sec:background}.

\paragraph{Notation.}
\begin{itemize}
\item For a vector $x \in \reals^n$, we define $(x)_+ = \max(x, 0)$, where this is done coordinate-wise and $(x)_- = \min(x, 0)$ (also coordinate-wise).
Using our convention, we have $(x)_+ = -(-x)_-$
\item We denote the set of integers $\{1, \ldots, n\}$ by $[n]$.
\item For a sequence $x_1, x_2, \ldots$, we denote the subsequence from $s$ to $t$ as $x_{s:t} = (x_s, x_{s+1}, \ldots, x_t)$.
\item We use $\ones_E$ or $\ones\{E\}$ to denote the indicator function of an event $E$, which is $1$ if $E$ is true and $0$ otherwise.
\item We use the notation $f(n) \asymp g(n)$ to denote that $f(n) = \Theta(g(n))$, \ie~there exist constants $c, C > 0$ such that $c |g(n)| \le |f(n)| \le C |g(n)|$ for sufficiently large $n$.
\item We write $X_n \asympp Y_n$ for random variables $X_n, Y_n$ if there exist constants $c, C > 0$ such that $c |Y_n| \le |X_n| \le C |Y_n|$ with high probability as $n \to \infty$.
\item For vectors $x, y \in \reals^n$, we write $x \prec y$ if $x$ is majorized by $y$ (see \S\ref{sec:schur-convex-ordering} for details).
\end{itemize}

\section{Background}\label{sec:background}\label{sec:model}
We will first introduce a mathematical model for solvency risks in a perpetuals exchange.
Using this model, we define the key components for measuring the solvency of perpetuals exchanges: trader equity, exchange deficits, and autodeleveraging (ADL) policies.
These components will allow us to formalize the notion of moral hazard in perpetuals exchanges and sets the notation up for our main results in~\S\ref{sec:negative},~\S\ref{sec:fairness},~\S\ref{sec:glpr}, and~\S\ref{sec:stackelberg-rewrite}.

\subsection{Perpetuals Exchanges.}
A \emph{perpetuals futures exchange} (or perpetuals exchange) is a derivatives trading venue that allows traders to use leverage to speculate on the prices of cryptocurrency assets.
The main asset traded is a perpetual future, which is a continuous, expiryless futures contract.
Traders speculate on the contract in either the long (buy) or short (sell) direction.
Traders are allowed to make leveraged bets up to a maximum leverage cap $\ell^{\max}$.
As is often common with levered positions, the trader has to post collateral, known as \emph{margin}, to keep their positions \emph{solvent}.
We define a precise notion of solvency in this section, but for now think of solvency conditions as meaning that a trader's assets are greater than their liabilities.

The exchange incentivizes futures prices to approximately match spot prices by having long and short traders pay each other a continuous rate, known as the \emph{funding rate}.
The funding rate is determined by the spot price of the asset and the relative imbalance of long and short futures positions.
Arbitrageurs keep the price of the futures contract in line with the spot price by opening futures positions that push the contract price towards the spot price.
For simplicity, we assume that trades occur in discrete time $t \in \naturals$, but note that our results can easily be extended to continuous time $t \in \reals_+$.

\paragraph{Spot Price Oracle.}
Perpetual futures exchanges require a spot price oracle that provides a spot price of the asset from a spot trading venue.
Price oracles are assumed to be expensive or difficult to manipulate, which is dependent on the liquidity of the spot market.
Most perpetuals exchanges utilize oracles that aggregate price reporting over multiple venues to increase the overall liquidity represented in the price and hence, the manipulation cost.
We assume that there is a non-manipulable spot price $\hat{p}_t$ for an asset at all times $t \in \reals_+$.
Furthermore, we denote the futures price quoted by the exchange as $p_t$ and similarly assume that it exists for all $t \in \reals_+$.

\paragraph{Positions.}
Given a spot price oracle and futures price, a perpetuals exchange is able to open and close \emph{positions} created by its traders.
A position consists of collateral deposited by the trader, a leverage amount, and a direction (\eg~buy or sell).
A trader's position is updated as the spot price changes, realizing a gain or a loss for the trader.
The accumulation of gains and losses over the lifetime of the position is the trader's PNL (profit-and-loss).
We define a perpetuals exchange with $n$ traders as the set 
\begin{equation}\label{eq:perp-exch}
\mathcal{P}_n = \{(q_i, c_i, t_i, b_i) \in \reals_+^3 \times \{-1, 1\}  : i \in [n]\}
\end{equation}
where $q_i \geq 0$ is the notional quantity of futures held by the trader, $c_i \geq 0$ is the trader's collateral (cash) position (also known as the \emph{initial margin}), $t_i \in \naturals$ is the timestamp that the position was created, and $b_i$ is their side ($-1$ is sell, $1$ is buy).
We denote each position $\mathfrak{p}_i = (q_i, c_i, t_i, b_i)$ and assume that without the loss of generality that there is a fixed, strictly increasing ordering of these orders such that we can refer to $\mathfrak{p}_i$ unambiguously.

For each position $\mathfrak{p}$, we define the \emph{notional exposure} at time $t$ is the gross dollar size
\[
 n_{i,t}\ =\ p_t\,q_i\,.
\]
The notional exposure is the gross dollar size of the trader's position and is used to define leverage.

\iparagraph{Example.}
In this section, we fix a canonical running example that we will reuse across definitions to illustrate each exchange function.
Our simple test perpetuals exchange will have five positions, $\mathcal{P}_5$:
\begin{align*}
    \mathcal{P}_5 = \Big\{& \mathfrak{p}_A = (q_A,c_A,t_A,b_A) = \big(1,\,2,\,0,\,+1\big),\\
    & \mathfrak{p}_B = (q_B,c_B,t_B,b_B)=\big(1,\,\tfrac{2}{3},\,0,\,+1\big),\\
    & \mathfrak{p}_C = (q_C,c_C,t_C,b_C)=\big(4,\,\tfrac{8}{3},\,0,\,-1\big),\\
    & \mathfrak{p}_D = (q_D,c_D,t_D,b_D)=\big(1,\,\tfrac{2}{19},\,0,\,+1\big),\\
    & \mathfrak{p}_E = (q_E,c_E,t_E,b_E)=\big(1,\,\tfrac{10}{99},\,0,\,-1\big)\Big\}.
\end{align*}
Moreover, if at $t=0$ we have $p_0=1$ for $\mathcal{P}_5$, then the notional exposure are:
\[
 n_{A,0}=1\cdot 1=1,\quad n_{B,0}=1\cdot 1=1,\quad n_{C,0}=1\cdot 4=4,\quad n_{D,0}=1\cdot 1=1,\quad n_{E,0}=1\cdot 1=1.
\]
Note this toy example is net short by 2 contracts: $\sum_i b_i q_i = -2$. 
For AMM-style venues, the pool holds the offsetting inventory $q_V=2$ with $b_V=+1$, so across traders plus venue $\sum_i b_i q_i + b_V q_V = 0$. 
For order-book venues, interpret $\mathcal{P}_5$ as a subset of accounts.
The complementary matched positions are omitted but implied.\footnote{We note that in centralized exchanges, the exchange itself creates the complementary positions whereas in decentralized exchanges with order books (like Hyperliquid and Lighter), a liquidity provider like HLP or LLP owns the position; see Hyperliquid's HLP documentation~\cite{HyperliquidHLPVaults} and the Lighter whitepaper~\cite{LighterWhitepaper}.} 
All zero-sum statements (\eg,~funding rates, which will be defined shortly) are taken over traders plus the venue inventory.

\paragraph{Position Creation.}
When an exchange opens a new position, it performs two tasks:
\begin{enumerate}
\item Selecting a price $p_{t}$ for the position and constructing $\mathfrak{p} = (q, c, t, b)$.
\item Construction of an equal but opposite position, $\overline{\mathfrak{p}} = (q, c, t, -b)$ is created
\end{enumerate}
The latter condition is needed because the perpetuals exchange is meant to be a neutral matching layer for traders.
This means that the exchange should not have a net exposure to the spot asset, as it increases the risk of default for traders.
Specifically, the position $\overline{\mathfrak{p}}$ is created to ensure that the total notional exposure of the exchange to the spot asset is zero.\footnote{In centralized venues such as Binance Futures, the matching engine and clearing layer are designed to operate a matched book over user positions, as reflected in their open interest endpoints for USD\$- and COIN-margined futures~\cite{BinanceFuturesOpenInterest}. In decentralized order-book venues like Hyperliquid, the protocol core similarly avoids taking directional exposure. Offsetting inventory is held in separate HLP vaults on behalf of liquidity providers~\cite{HyperliquidHLPVaults}.}
We define 
\[
\overline{\mathcal{P}}_n = \mathcal{P}_n \cup \{ \overline{\mathfrak{p}} : \mathfrak{p} \in \mathcal{P}_n\}
\]
as the full set of positions held by the exchange (including matching, complementary positions).

While all venues construct $\overline{\mathfrak{p}}$ in the same manner, the price $p_t$ is chosen in a venue-specific manner.
Each venue has a \emph{liquidity model} for determining how prices $p_t$ are chosen.
The two most popular liquidity models are order books and automated market makers (AMMs).

Order book venues maintain a limit order book of bids and asks for perpetual futures contracts at discrete price levels, and have been extensively studied in the market microstructure literature~\cite{Kyle1985,AlmgrenChriss2001,Gatheral2010,BouchaudImpact2010}.
Traders may submit limit orders, which rest in the book until they are matched, or market orders, which execute immediately against resting liquidity.
When a trader opens a new position as a taker, the execution price $p_t$ is the volume-weighted average of the resting orders they consume, \ie~the price they would obtain by submitting a market order of size $q$ into the book at time $t$.
The largest order book venues in the perpetual futures market include centralized exchanges such as Binance and Bybit and decentralized exchanges such as Hyperliquid and Lighter.\footnote{From a risk-warehousing perspective, Hyperliquid's HLP vaults cause the protocol to function as a hybrid between a pure order-book venue and an inventory-taking pool: user orders are matched via an order book, but directional exposure is warehoused collectively in HLP vaults on behalf of liquidity providers~\cite{HyperliquidHLPVaults,OakResearchJELLY,CoinDesk2025LargestLiquidations}.}

Automated market maker (AMM) venues maintain a pool of collateral and inventory for the perpetual, and quote prices according to a deterministic pricing rule or bonding curve that depends on the pool state, typically implemented via constant-function market makers~\cite{AngerisChitra2020ImprovedOracles,Angeris2023GeometryCFMM}.
Traders interact with the pool via swap-like trades: submitting an order of size $q$ moves the pool along its pricing curve, and the execution price $p_t$ is the average price paid over this path.
When a trader opens a new position against the AMM, the pool itself holds the offsetting inventory $\overline{\mathfrak{p}}$, so that the trader's exposure is matched by an equal and opposite position in the pool, as in inventory-taking perpetual venues such as GMX, Drift, and Perpetual Protocol v2~\cite{GMXDocs,DriftADLCode,PerpV2Docs}.

In the remainder of this paper, we will ignore the liquidity model of a perpetuals exchange.
This is because the default conditions that we will study for autodeleveraging only depend on the positions $\mathcal{P}_n$ and not the process by which they were created.

\iparagraph{Example.}
Consider an order-book venue at time $t$ with three resting sell (ask) limit orders
\[
(p^1, q^1) = (1.0, 1),\quad (p^2, q^2) = (1.1, 3),\quad (p^3, q^3) = (1.2, 10).
\]
A trader submits a market buy order for quantity $q=5$. The exchange matches this order against the resting asks in price–time priority: it first consumes $q_1=1$ at $p_1=1.0$, then $q_2=3$ at $p_2=1.1$, and finally $1$ unit from the third order at $p_3=1.2$.
The total notional paid is $1\cdot 1.0 + 3\cdot 1.1 + 1\cdot 1.2 = 5.5$, so the execution price is the volume-weighted average $p_t = \frac{5.5}{5} = 1.1$.
Ignoring fees, the taker acquires a new long position $\mathfrak{p} = (q,c,t,b) = (5,c,t,+1)$ with notional exposure $q\,p_t = 5.5$.
Aggregating across the three maker accounts, the exchange simultaneously creates an equal and opposite short position $\overline{\mathfrak{p}}$ of size $5$, so that total long and short contract quantities (over all accounts) remain matched.

\paragraph{Open Interest.}
Note that the notion of leverage is encoded into the quantity $q_i$ of a position, as a user who can buy $p_{t_i} q$ units of spot assets with $c_i$ units of collteral can open a position with $q_i \leq \ell_{\max} q$.
This implies that the correct notion of the notional trades outstanding, known as \emph{interest}, on a pereptuals exchange is simply the sum of price-weighted quantities.
We divide the exchange $\mathcal{P}_n$ into long and short positions:
\begin{align*}
\mathcal{L}(\mathcal{P}_n) &= \{ (q_i, c_i, t_i, b_i) \in \mathcal{P}_n : b_i = 1\} \\
\mathcal{S}(\mathcal{P}_n) &= \{ (q_i, c_i, t_i, b_i) \in \mathcal{P}_n : b_i = -1\}
\end{align*}
We define the total set of long and short positions as:
\begin{align*}
  \overline{\mathcal{L}}(\mathcal{P}_n) &= \mathcal{L}(\mathcal{P}_n) \cup\{ \overline{\mathfrak{p}} : \mathfrak{p} \in \mathcal{L}(\mathcal{P}_n)\} \\
  \overline{\mathcal{S}}(\mathcal{P}_n) &= \mathcal{S}(\mathcal{P}_n) \cup\{ \overline{\mathfrak{p}} : \mathfrak{p} \in \mathcal{S}(\mathcal{P}_n)\}
\end{align*}
This is the set of all long and short positions, including matching positions.

Futures exchanges define the \emph{open interest} of the exchange as the total outstanding notional value of the positions held by the exchange.
Given the futures price $p_t$, the long open interest $L(\mathcal{P}_n, p_t)$ and short open interest $S(\mathcal{P}_n, p_t)$ are defined as 
\begin{align*}
L(\mathcal{P}_n, p_t) &= \sum_{(q_i, c_i, t_i, b_i) \in \mathcal{L}(\mathcal{P}_n)} q_i p_t \\
S(\mathcal{P}_n, p_t) &= \sum_{(q_i, c_i, t_i, b_i) \in \mathcal{S}(\mathcal{P}_n)} q_i p_t 
\end{align*}
Given the long and short open interest, we define the exchange's open interest simply as the sum, $OI(\mathcal{P}_n, p_t) = L(\mathcal{P}_n, p_t) + S(\mathcal{P}_n, p_t)$
\footnote{We note that exchange's will report open interest over $\overline{\mathcal{P}}_n$, which would be twice the open interest reported here. Since we assume that $\mathcal{P}_n$ is the set of unmatched trader positions, one can view our definition of open interest (which differs from Binance's definition by a factor of two) as `trader open interest' vs. `exchange open interest'.}
To lighten notation, we will write $L_t = L(\mathcal{P}_n, p_t), S_t = S(\mathcal{P}_n, p_t), OI_t = OI(\mathcal{P}_n, p_t)$ when the exchange $\mathcal{P}_n$ is already defined.

\iparagraph{Example.}
For $\mathcal P_5$ at $t=0$ with $p_0=1$, longs are $\{A,B,D\}$ and shorts are $\{C,E\}$, so
\[
 L_0=\sum_{b_i=+1} q_i p_0=(1+1+1)\cdot 1=3,\quad S_0=\sum_{b_i=-1} q_i p_0=(4+1)\cdot 1=5,\quad OI_0=8.
\]
At $t=1$ with $p_1=1.4$ (same quantities), $L_1=3\cdot 1.4=4.2$, $S_1=5\cdot 1.4=7.0$, $OI_1=11.2$ .

\paragraph{Leverage and Margin.}
The \emph{leverage} of a position is the ratio of the notional exposure to the collateral value.
If the leverage is greater than 1, it means that the user has effectively borrowed money from the exchange to have an exposure higher than their cash balance.
When the leverage is less than 1, the user's position is \emph{overcollateralized}.

As the price of the asset changes, the leverage of the position naturally changes due to its dependence on notional exposure.
When a position is opened, there is an \emph{initial leverage} ratio that the position is created with.
We define the initial leverage of a position as 
\[
\ell_i = \ell(\mathfrak{p}_i) = \frac{n_{i,t_i}}{c} = \frac{p_{t_i} q_i}{c_i}
\]
The larger the initial leverage, the most exposure the trader has to the asset and implicitly, the more risk the exchange is taking on lending to the user.

For risk management purposes, exchanges define leverage limits in terms of initial and maintenance~\emph{margins}.
Margins are bounds on the ratio of notional exposure to collateral that the exchange enforces as an invariant.
The \emph{initial margin ratio}, $m_I \in (0, 1)$ is defined such that for any valid position, we have:
\begin{equation}\label{eq:initial-margin-condition}
m_I p_{t_i} q_i \leq c_i
\end{equation}
This invariant implies that the maximum leverage that a position can have, $\ell^{\max}$ is $\ell^{\max} = \frac{1}{m_I}$.

In order to continually satisfy~\eqref{eq:initial-margin-condition} as spot and future prices change, a trader often needs to `top up' or add more collateral their position.
The amount of collateral needed is related to the exchange's \emph{maintenance margin ratio} $m_{\mu}$, which will be defined in the sequel.
As the user adjusts their collateral position dynamically, we let $c_{i,t}$ be the total quantity of collateral (margin) placed by the user at time $t$.
This also means that the state of position is dynamic, $\mathfrak{p}_{i,t} = (q_i, c_{i,t}, t_i, b_i)$.
We define the \emph{leverage at time $t$} as $\ell_{i,t} = \ones_{t \geq t_i} \frac{p_t q_i}{c_{i,t}}$



\iparagraph{Example.}
These five examples of $\mathcal{P}_5$ represent one highly overcollateralized and under leveraged position ($\mathfrak{p}_A$), two somewhat leveraged positions with medium leverage ($\mathfrak{p}_B, \mathfrak{p}_C$) and two highly leveraged positions ($\mathfrak{p}_D, \mathfrak{p}_E$).
We will initialize the example with $p_0=1$ and take a maximum leverage $\ell^{\max}=10$ (\eg,~$m_I=0.10$).
This gives the following opening leverages:
\begin{align*}
\ell_{A,0}&=\tfrac{1\cdot 1}{2}=0.5, &
\ell_{B,0}&=\tfrac{1\cdot 1}{2/3}=1.5, &
\ell_{C,0}&=\tfrac{1\cdot 4}{8/3}=1.5,\\
\ell_{D,0}&=\tfrac{1\cdot 1}{2/19}=9.5, &
\ell_{E,0}&=\tfrac{1\cdot 1}{10/99}=9.9.
\end{align*}
If the mark moves to $p_1=1.4$ and cash is unchanged ($c_{i,1}=c_i$), then time-$1$ leverages are
\begin{align*}
\ell_{A,1}&=\tfrac{1.4\cdot 1}{2}=0.7, &
\ell_{B,1}&=\tfrac{1.4\cdot 1}{2/3}=2.1, &
\ell_{C,1}&=\tfrac{1.4\cdot 4}{8/3}=2.1,\\
\ell_{D,1}&=\tfrac{1.4\cdot 1}{2/19}=13.3, &
\ell_{E,1}&=\tfrac{1.4\cdot 1}{10/99}=13.86.
\end{align*}
In practice, if $\ell_{i,1}>\ell^{\max}$ (\eg,~$\mathfrak{p}_D$ and $\mathfrak{p}_E$), then the venue requires collateral additions from the trader to restore the constraint $\ell_{i,1} \leq \ell^{max}$.

\paragraph{Funding Rate.}
A perpetual future is only useful as a hedging instrument if the futures price $p_t$ and spot price $\hat{p}_t$ are sufficiently ``close.''
To incentivize this, the majority of perpetuals futures markets use a continuous payment stream between the long and short positions called a \emph{funding rate}.
If the price of the future is lower than spot, $p_t < \hat{p}_t$, traders holding short futures positions pay a payment to traders holding long futures positions.
This encourages traders who want to earn these payments to open long positions, which will decrease the over all rate.
Conversely, if $p_t > \hat{p}_t$, traders holding long positions pay those holding short positions.

Most funding rates are simple linear functions of the relative imbalance between the long and short positions.
A common funding rate~\cite{Chitra2025PDLP} formula takes the form:
\begin{equation}\label{eq:single-time-funding}
\gamma(\mathcal{P}_n, p_t, \hat{p}_t) = \kappa \left(\frac{L(\mathcal{P}_n, p_t)}{S(\mathcal{P}_n, p_t)} - \frac{p_t}{\hat{p}_t} \right)
\end{equation}
for a constant $\kappa > 0$.
For brevity, write $\gamma_t = \gamma(\mathcal{P}_n, p_t, \hat{p}_t)$.
Note that we have defined $\gamma_t$ such that if $\gamma_t > 0$, then short traders pay long traders $\gamma_t$\% of their holdings.
This convention is opposite to what many exchanges use (\ie~$\gamma_t < 0$ represents a payment from the short side to the long side).
However, our convention leads to needing to check fewer negative signs in our formulas, as we will illustrate with our examples.

Given funding rates $\gamma_t$ for $t \geq 0$, a position $\mathfrak{p} = (q, c, t, b)$ linearly accumulates funding positions between time $t$ and the current time $T$, if the position is solvent.
That is, if the position is solvent at time $T$, the user receives a payments of the form
\[
\Gamma((q, c, t, b), t, T) = \sum_{s=t+1}^T (bq) \gamma_s p_s 
\]
If $\Gamma > 0$, then over the interval $[t, T]$, the user receives a payment from the exchange whereas if $\Gamma < 0$, the user pays the exchange.
When the user pays the exchange, their profits are distributed by the exchange to the other side of the market.
For brevity, we use the notation $\Gamma_{i,t,T} = \Gamma(\mathfrak{p}_i, t, T)$.

We also note that the funding rate, by construction, is zero-sum when one considers entire set of exchange positions $\overline{\mathcal{P}}_n$.
To see, this note the following:
\begin{equation}\label{eq:funding-zero-sum}
\sum_{\mathfrak{p} \in \overline{\mathcal{P}}_n} \gamma_t p_t\,(b q) = \sum_{\mathfrak{p} \in \overline{\mathcal{L}}(\mathcal{P}_n)} \gamma_t p_t\,(q) - \sum_{\mathfrak{p} \in \overline{\mathcal{S}}(\mathcal{P}_n)} \gamma_t p_t\,(q) = 0
\end{equation}
where the last equality follows from the fact that $\mathfrak{p} \in \overline{\mathcal{L}}(\mathcal{P}_n) \iff \overline{\mathfrak{p}} \in \overline{\mathcal{S}}(\mathcal{P}_n)$.

\iparagraph{Example.}
To make the linear funding rule concrete, take $\kappa=1$ and suppose the oracle and mark follow
$(\hat p_0,p_0)=(1,1)$, $(\hat p_1,p_1)=(1.5,1.4)$, $(\hat p_2,p_2)=(1.25,1.3)$, with quantities unchanged (so $\tfrac{L}{S}=\tfrac{3}{5}$).
Then with $\kappa=1$, $\gamma_1=\tfrac{L}{S}-\tfrac{p_1}{\hat p_1}=0.6-\tfrac{1.4}{1.5}\approx -0.3333$ and $\gamma_2=0.6-\tfrac{1.3}{1.25}=-0.44$ (longs pay shorts at both steps).
Per-step funding cash to each position is $\gamma_t p_t\,(b_i q_i)$ (positive means received, negative means paid):
\[
\begin{array}{c|ccccc}
 t & \mathfrak p_A & \mathfrak p_B & \mathfrak p_C & \mathfrak p_D & \mathfrak p_E \\
 \hline
 0 & 0.0000 & 0.0000 & 0.0000 & 0.0000 & 0.0000 \\
 1 & -0.4667 & -0.4667 & +1.8667 & -0.4667 & +0.4667 \\
 2 & -0.5720 & -0.5720 & +2.2880 & -0.5720 & +0.5720 \\
\end{array}
\]
Values follow Eq.~\eqref{eq:single-time-funding} using $(\hat p_t,p_t)$ above.

\paragraph{Profit and Loss.}
Given funding rates, one can write the explicit profit and/or loss that a trader faces during the lifetime of their position.
We define the profit-and-loss (PNL) of duration $T > 0$, $\mathsf{PNL}_{s:T} : (\reals_+^4 \times \{-1, 1\}) \times \reals_+^T \times \reals_+^T \rightarrow \reals$ as:
\begin{equation}\label{eq:PNL}
\mathsf{PNL}_{s:T}((q, c, \ell, t, b), p_{1:T}, \hat{p}_{1:T}) = \ones_{s\leq t < T} \left( b q (p_{\hat{T}} - p_t) + \Gamma((q, c, t, b), t, \hat{T}) \right) 
\end{equation}
where $\hat{T} = \min(T, \tau)$ is the last time the position is solvent.
We will use the shorthand $\mathsf{PNL}_T = \mathsf{PNL}_{0:T}$.
We will formally define the quantity $\tau$ shortly when we discuss liquidations, but for now think of it as a martingale stopping time for position solvency.
This condition states that the total position of the user is their collateral plus the net change from funding costs.
A natural consequence of the zero-sum nature of funding is that the total PNL is also zero-sum:
\[
\sum_{\mathfrak{p} \in \overline{\mathcal{P}}_n} \mathsf{PNL}_T(\mathfrak{p}, p_{1:T}, \hat{p}_{1:T}) = \sum_{\mathfrak{p} \in \overline{\mathcal{P}}_n} b q (p_{\hat{T}} - p_0) + \Gamma(\mathfrak{p}, t, \hat{T}) = 0
\]
The first term sums to zero, because there will be offsetting positions with $b = -1, b = 1$ and the second term is zero via~\eqref{eq:funding-zero-sum}.

\iparagraph{Example.}
Under the additive funding convention in Eq.~\eqref{eq:PNL} with the prices, open interest, and funding above, PNL over $[0,1]$ and $[0,2]$ is:
\begin{align*}
\mathsf{PNL}_{0:1}(\mathfrak p_A) &= -0.0667,&\quad \mathsf{PNL}_{0:2}(\mathfrak p_A) &= -0.7387,\\
\mathsf{PNL}_{0:1}(\mathfrak p_B) &= -0.0667,&\quad \mathsf{PNL}_{0:2}(\mathfrak p_B) &= -0.7387,\\
\mathsf{PNL}_{0:1}(\mathfrak p_C) &= +0.2667,&\quad \mathsf{PNL}_{0:2}(\mathfrak p_C) &= +2.9547,\\
\mathsf{PNL}_{0:1}(\mathfrak p_D) &= -0.0667,&\quad \mathsf{PNL}_{0:2}(\mathfrak p_D) &= -0.7387,\\
\mathsf{PNL}_{0:1}(\mathfrak p_E) &= +0.0667,&\quad \mathsf{PNL}_{0:2}(\mathfrak p_E) &= +0.7387.\end{align*}

\paragraph{Equity.}
For each position $\mathfrak{p}_{i,t}$, one can view the assets of the position as $c_{i, t}$ (\eg~the margin posted by the user) and the liabilities of the users as $-\mathsf{PNL}_T(\mathfrak{p}_{i,t}, p_{1:T}, \hat{p}_{1:T})$.
The liabilities are the negative of PNL because if the user has a loss, then they have positive liabilities to the exchange (\eg~they owe the exchange money).
On the other hand, if the user has a gain, then the exchange owes the user money (and is a negative liability).
The \emph{equity} of a position $\mathfrak{p}_{i, t}$, $e(\mathfrak{p}_{i, t})$ is simply the assets of the position minus the liabilities:
\begin{equation}\label{eq:equity}
e(\mathfrak{p}_{i, t}, p_{1,T}, \hat{p}_{1,T}) = c_{i, t} + \mathsf{PNL}_T(\mathfrak{p}_{i, t}, p_{1:T}, \hat{p}_{1:T})
\end{equation}
For notational convenience, we will write the terminal equity shorthand
\[
e_T(\mathfrak{p}_{i,t}) = e(\mathfrak{p}_{i,t},\ p_{1:T},\ \hat{p}_{1:T}).
\]
A position is said to be \emph{totally insolvent} if $e_T(\mathfrak{p}_{i, t}) < 0$.
Traders can update their collateral $c_{i, t}$ or close positions to avoid fully insolvency.
We define the random variable $\tau(\mathfrak{p}_{i, t}, p_{1:t})$ as the first time that a position is insolvent, \ie~
\[
    \tau(\mathfrak{p}_{i, t}, p_{1:t}) = \min\{ s \leq t : e(\mathfrak{p}_{i, s}, p_{1:s}) \leq \mu p_s q_i \}
\]
For any time index $t$, we will refer to the \emph{winner} and \emph{loser} index sets as
\[
 \mathcal{W}_t\ =\ \{\,i:\ e(\mathfrak p_{i,t})>0\,\},\qquad
 \mathcal{L}_t\ =\ \{\,i:\ e(\mathfrak p_{i,t})<0\,\}.
\]
Write their sizes as $k_t=|\mathcal{W}_t|$ and $m_t=|\mathcal{L}_t|$. We will use $\mathcal{W}_T,\mathcal{L}_T$ for the terminal sets at horizon $T$.


\iparagraph{Example.}
To illustrate the maintenance test, let $\mu=0.10$ and reuse the $[0,1]$ price and funding path above. Equities at $t=1$ (using $c_{i,1}=c_i$) are
\begin{align*}
e(\mathfrak p_A)&=2-0.0667\approx 1.9333, & e(\mathfrak p_B)&=\tfrac{2}{3}-0.0667\approx 0.6000, & e(\mathfrak p_C)&=\tfrac{8}{3}+0.2667\approx 2.9334,\\
e(\mathfrak p_D)&=\tfrac{2}{19}-0.0667\approx 0.0386, & e(\mathfrak p_E)&=\tfrac{10}{99}+0.0667\approx 0.1677.
\end{align*}
At $t=2$ (using the same path and $c_{i,2}=c_i$), equities are
\begin{align*}
e(\mathfrak p_A)&\approx 1.2613, & e(\mathfrak p_B)&\approx -0.0720, & e(\mathfrak p_C)&\approx 5.6214,\\
e(\mathfrak p_D)&\approx -0.6334, & e(\mathfrak p_E)&\approx 0.8397.
\end{align*}
Note that under this price move, there are now negative equity positions that are liquidatable.

\paragraph{Maintenance Margin.}
As spot and future prices change, an exchange needs to enforce dynamic collateral requirements to avoid insolvencies.
The \emph{maintenance margin ratio}, $m_{\mu}$, represents the ratio of collateral to notional exposure that a trader must maintain under adverse price moves.
With $m_{\mu}\in(0,1)$ and equity from eq.~\eqref{eq:equity}, a position is maintenance-insolvent (and liquidatable) when
\begin{equation}\label{eq:maintenance-margin}
e(\mathfrak p_{i,t})\ \le\ m_{\mu}\,p_t\,|q_i|.
\end{equation}
Note that maintenance margin depends on the equity of a position (\eg~includes the position's PNL) whereas initial margin only depends on the initial cash position.
If a user violates~\eqref{eq:maintenance-margin} due to a price move, they can add collateral within a specified time interval to keep their position open.
Otherwise, the venue liquidates the position.

\iparagraph{Example.}
Take $m_{\mu}=0.10$ and reuse the $[0,1]$ path above with $p_1=1.4$ and the computed equities
\[
e(\mathfrak p_A)\approx 1.9333,\quad e(\mathfrak p_B)\approx 0.6000,\quad e(\mathfrak p_C)\approx 2.9334,\quad e(\mathfrak p_D)\approx 0.0386,\quad e(\mathfrak p_E)\approx 0.1677.
\]
The maintenance thresholds are $m_{\mu} p_1|q|=(0.14,\ 0.14,\ 0.56,\ 0.14,\ 0.14)$ for $(A,B,C,D,E)$. Here $e(\mathfrak p_D)\le m_{\mu} p_1|q|$ while others remain above the threshold.
Position $D$ is maintenance-insolvent and would need to add collateral or be liquidated.

\paragraph{Leverage Mass.}
The final quantity we will define aims to measure how much leverage the winning (positive equity) and losing (negative equity) sides have on the exchange.
This will serve as a coarse measure of an exchange's risk exposure and will be used in~\S\ref{sec:negative}.
We define a natural aggregate leverage quantity for both the winning and losing sides that normalizes leverage by equity, allowing us to compare traders with very different margin ratios on a common ``true asset composition'' scale.
For $i\in\mathcal{W}_t$ define the \emph{winner effective leverage} $\lambda^+_{i,t}=\frac{n_{i,t}}{e(\mathfrak p_{i,t})}$, and for $i\in\mathcal{L}_t$ the \emph{loser effective leverage} $\lambda^-_{i,t}=\frac{n_{i,t}}{|e(\mathfrak p_{i,t})|}$. The \emph{winner} and \emph{loser leverage masses} are
\begin{equation}
\label{eq:leverage-mass-def}
 \ell^+_t\ =\ \sum_{i\in\mathcal{W}_t} \lambda^+_{i,t},\qquad
 \ell^-_t\ =\ \sum_{i\in\mathcal{L}_t} \lambda^-_{i,t},\qquad
 \bar \ell^{\pm}_t\ =\ \frac{\ell^{\pm}_t}{n}\ (\text{per–trader}).
\end{equation}
Note that unlike the leverage $\ell_{i, t}$, we divide by the equity (which can be larger or smaller than leverage, depending on the trader's PNL).

\iparagraph{Example.}
Consider $t=1$ with $p_1=1.40$. Then $n_{i,1}= p_1 q_i = 1.4$ for $i\in\{A,B,D,E\}$ and $5.6$ for $C$.
Using the equities at $t=1$ above, all five are winners and
\begin{align*}
\lambda^+_{A,1}&= \tfrac{1.4}{1.9333} \approx 0.7246, &
\lambda^+_{B,1}&= \tfrac{1.4}{0.6000} = 2.3333, &
\lambda^+_{C,1}&= \tfrac{5.6}{2.9334} \approx 1.9090,\\
\lambda^+_{D,1}&= \tfrac{1.4}{0.0386} \approx 36.269, &
\lambda^+_{E,1}&= \tfrac{1.4}{0.1677} \approx 8.353.
\end{align*}
Hence $\ell^+_1\approx 49.59$; the large $\lambda^+_{D,1}$ reflects $D$'s very small equity at $t=1$.
Recall $D$'s opening leverage was only $\ell_{D,0}\approx 9.5$, so effective leverage can greatly exceed raw leverage when equity has been eroded by losses.
At $t=2$ with $p_2=1.30$, $n_{i,2}=p_2|q_i|$ equals $1.3$ for $i\in\{A,B,D,E\}$ and $5.2$ for $C$.
From the equity example above, the winner/loser sets are $\mathcal W_2=\{A,C,E\}$ with $e_A\approx1.2613$, $e_C\approx5.6214$, $e_E\approx0.8397$ and $\mathcal L_2=\{B,D\}$ with $e_B\approx-0.0720$, $e_D\approx-0.6334$.
Winner effective leverages are
\begin{align*}
\lambda^+_{A,2}&= \tfrac{1.3}{1.2613} \approx 1.031, &
\lambda^+_{C,2}&= \tfrac{5.2}{5.6214} \approx 0.925, &
\lambda^+_{E,2}&= \tfrac{1.3}{0.8397} \approx 1.548,
\end{align*}
so $\ell^+_2\approx 3.504$.
Loser effective leverages are
\begin{align*}
\lambda^-_{B,2}&= \tfrac{1.3}{|{-}0.0720|} \approx 18.05, &
\lambda^-_{D,2}&= \tfrac{1.3}{|{-}0.6334|} \approx 2.052,
\end{align*}
so $\ell^-_2\approx 20.102$.
The ratio of leverage masses $\ell^-_2 / \ell^+_2 \approx 5.7$ shows that, at this horizon, most effective leverage sits on the losing side of the book.
Here $B$ and $D$ have the same notional $n_{i,2}$ but very different effective leverages, illustrating how $\lambda^{\pm}_{i,t}$ can diverge from the raw leverage $\ell_{i,t}$ when equities differ.

\subsection{Liquidations}
When a user's position is insolvent,~\ie~\eqref{eq:maintenance-margin} holds, the exchange must \emph{liquidate} the position --- that is, remove some or all of the position from $\mathcal{P}_n$.
Liquidations are used to ensure that the exchange itself remains solvent, which ensures that traders can withdraw their funds and profits as expected.
The process of liquidating a position is probabilistic, however, and can fail in a number of ways.

We will illustrate the high-level process of liquidation using bankruptcy and liquidation prices.
Then we will describe how exchanges can lose solvency when liquidations fail, leading to the usage of autodeleveraging mechanisms.
Our description of liquidations is minimal, avoiding formal mathematical description unless needed, as we are concerned with liquidation failure versus liquidation mechanics.
We refer the interested reader to~\cite{Soska2021BitMEX,perez2021liquidations,qin2021empirical,AngerisChitra2023PerpsSIAM} for detailed accounts of liquidation mechanisms and modeling on centralized and decentralized venues.

\subsubsection{Liquidation Prices}
Before defining how liquidations work, we need to define criteria for when a position is eligible for liquidation.
These criteria depend on when certain prices are reached where a trader position has low or negative equity.
We focus on defining three prices: bankruptcy price (zero equity), liquidation price (low equity), and execution price (actual realized price of a liquidation).
In practice, liquidation mechanisms often use more price-based variables to decide on execution.
However, all mechanisms define these three prices, allowing us to abstract liquidations to mechanisms involving these prices.

\paragraph{Bankruptcy Price.}
The \emph{bankruptcy price} of a position $\mathfrak{p}_{i, t}$, $p^{bk}(\mathfrak{p}_{i, t})$, is the highest price at which the position is totally insolvent.
The threshold condition for a position to be totally insolvent is $e_T(\mathfrak{p}_{i, t}) = 0$:
\[
    0 = e_T(\mathfrak{p}_{i,t}) = c_{i, t} + \mathsf{PNL}_T(\mathfrak{p}_{i, t}, p_{1:T}, \hat{p}_{1:T}) = c_{i,t} + b_i q_i (p^{bk}(\mathfrak{p}_{i, t}) - p_t) + \Gamma(\mathfrak{p}_{i, t}, t, \hat{T})
\]
Rearranging this gives:
\begin{equation}\label{eq:bankruptcy-price}
p^{bk}(\mathfrak{p}_{i, t}) = \max\left(p_t - \frac{c_{i,t} + \Gamma(\mathfrak{p}_{i, t}, t, \hat{T})}{b_i q_i}, 0 \right)
\end{equation}
One can view the bankruptcy price as the worst-case liquidation price, in the sense that if the position is liquidated at any $p < p^{bk}(\mathfrak{p}_{i, t})$, then other traders of the exchange will realize a loss.
Such a loss, which is either bourne by the exchange or by other traders, is known as \emph{bad debt} (see~\S\ref{subsec:exchange-solvency}).

The bankruptcy price also constrains the maximum leverage that a position can have.
Suppose that the funding is zero, \eg~$\Gamma(\mathfrak{p}_{i, t}, t, \hat{T}) = 0$.
The bankruptcy price then simplifies to:
\begin{equation}\label{eq:bankruptcy-price-no-funding}
p^{bk}(\mathfrak{p}_{i, t}) = \max\left(p_t - \frac{c_{i,t}}{b_i q_i}, 0\right) = \max\left(p_t - \frac{b_i}{\ell_{i,t}}\, p_t, 0\right) = p_t \max\left(1 - \frac{b_i}{\ell_{i,t}}, 0\right)
\end{equation}
This formula represents the common maxim that a perpetuals position with $\ell_i \geq 1$ times leverage will be liquidated when the price moves by $\frac{1}{\ell_i}$\% in the direction opposite to $b_i$; see, for example,~\cite{Binance2025CollateralLeverageUpdate,he2022fundamentals}.
For instance, a maximal 10x leveraged position will be liquidated when the price moves by 10\% from the initial price.
Detailed numerical examples of bankruptcy prices are provided in Appendix~\ref{app:liquidation-mechanics}.

\paragraph{Liquidation Price.}
To avoid bad debt, an exchange defines a \emph{liquidation price} $\hat{p}^{liq}(\mathfrak{p}_{i, t}) \ge p^{bk}(\mathfrak{p}_{i, t})$, when the position is liquidatable at time $t$.
Generally, the liquidation price is a spot price derived from the oracle versus a futures price quoted by the exchange.
This makes market manipulation --- creating and executing non-economically rational orders to force a liquidation --- costly to perform.
When the maintenance margin condition~\eqref{eq:maintenance-margin} is satisfied, the exchange attempts to partially liquidate the position.
The $\mu$-liquidation price, $p^{liq}(\mathfrak{p}_{i,t}, \mu)$ is defined as the maximal price where~\eqref{eq:maintenance-margin} holds. To avoid ambiguity about the base price used in $p^{bk}$, we write $\hat p^{liq}$ directly in terms of the entry price $p_{t_i}$ (independent of bankruptcy):
\[
\hat p^{liq}(\mathfrak p_{i,t},\mu) =
\begin{cases}
\displaystyle \max\!\left( \frac{\,p_{t_i} - (c_{i,t}+\Gamma(\mathfrak p_{i,t},t,\hat T))/q_i\,}{1-\mu},\ 0\right), & b_i=+1,\\[10pt]
\displaystyle \frac{\,p_{t_i} + (c_{i,t}+\Gamma(\mathfrak p_{i,t},t,\hat T))/|q_i|\,}{1+\mu}, & b_i=-1.
\end{cases}
\]
Detailed numerical examples of liquidation prices are provided in Appendix~\ref{app:liquidation-mechanics}.

\paragraph{Execution Price.}
When $\hat p_t < p^{liq}(\mathfrak p_{i,t},\mu)$ a position is liquidatable. The venue then sells (if $b_i=+1$) or buys (if $b_i=-1$) a slice of size $\Delta q$, realizing an execution price $p^{exec}(\mathfrak p_{i,t},\mu,\Delta q)$.
Whether a liquidation creates bad debt depends on the location of $p^{exec}$ relative to the bankruptcy price $p^{bk}$:
\begin{itemize}[leftmargin=12pt]
  \item \emph{Long} ($b_i=+1$): no shortfall if $p^{exec}\ge p^{bk}$; otherwise the realized shortfall is $D_t=(p^{bk}-p^{exec})\,\Delta q$.
  \item \emph{Short} ($b_i=-1$): no shortfall if $p^{exec}\le p^{bk}$; otherwise the realized shortfall is $D_t=(p^{exec}-p^{bk})\,\Delta q$.
\end{itemize}
Shortfalls are first absorbed by the insurance fund (\S\ref{subsec:exchange-solvency}) up to its balance.
Any residual shortfall not covered by the insurance fund is socialized via ADL (\S\ref{subsec:adl}).
The choice of liquidation size $\Delta q$ and impact parameter $\alpha$ together determines whether thin‑equity positions can typically be closed without realizing bad debt.
Detailed numerical examples of execution prices and shortfalls are provided in Appendix~\ref{app:liquidation-mechanics}.

\subsubsection{Liquidation Mechanics}\label{subsec:liquidation-mechanics}
Given liquidation prices, a natural question is how liquidations are mechanically executed.
At a high level, the exchange seizes a low equity position's cash and collateral and sells it to the market.
Such a sale can realize a profit or a loss for the exchange and is best thought of as a trading strategy itself.
We will provide a high-level description of liquidation mechanisms as trading strategies, with the caveat that there are many idiosyncrasies in practical implementations (see, \eg~\cite{BinanceADL,HyperliquidDocsLiquidations,BitMEXADL} for venue documentation and differences between auctions, order‑book liquidations, and RFQ-style closures).

\paragraph{Liquidation Costs.}
Most exchanges charge penalties to users who are liquidated as a means of disincentivizing liquidations and moral hazard.
These fees come in three flavors: fixed charges, insurance fund fees, and liquidation incentives.
The fixed fees correspond to reimbursement for gas and/or operational costs that an exchange realizes for performing a liquidation.
The insurance fund fees are proportional to the liquidation price and allow for the exchange to cover bad debt.
Finally, the liquidation incentives are used to encourage third-party actors known as liquidators to hold the risk of buying the position of size $\Delta q$ from the exchange and exiting it profitably.

In decentralized exchanges, liquidators are usually any market participant with enough capital who can connect to the exchange.
On the other hand, in centralized exchanges, liquidators are usually whitelisted parties approved by the exchange to perform liquidations
We denote the set of liquidation costs for a liquidation of size $\Delta q$ at time $t$ as $\tau_t(\Delta q) \in \reals_+$, which represents the cash cost paid by the user upon liquidation.
Detailed numerical examples of liquidation costs are provided in Appendix~\ref{app:liquidation-mechanics}.

\paragraph{Liquidation Strategies.}
The policy by which an exchange chooses the liquidation quantity $\Delta q$ is known as the \emph{liquidation strategy}, $L : \mathcal{P}_n \times \reals^T_+ \times \reals^T_+ \rightarrow \reals$.
The strategy $L(\mathfrak{p}_{i,t}, p_{1:T}, \hat{p}_{1:T})$ outputs a quantity $\Delta q$ to liquidate.
Generally speaking, the strategy has some model of the environment (\eg~the price impact function of the spot and futures exchanges) and utilizes that to pick the quantity.
Most exchanges utilize a simple a greedy liquidation strategy that chooses the minimal $\Delta q$ to ensure that the equity does not satisfy~\eqref{eq:maintenance-margin}.
Such strategies choose $\Delta q$ such that an equality of the form
\begin{equation}\label{eq:liquidation-strategy}
    e_T(\mathfrak{p}_{i,t}) + b\,\Delta q\,(p^{exec}_t - p_t) - \tau_t(\Delta q) = \mu p_t (q - \Delta q)
\end{equation}
approximately holds.
For notational simplicity, we assume that $L(\mathfrak{p}_{i,t}, p_{1:T}, \hat{p}_{1:T}) = 0$ if the position $\mathfrak{p}_{i,t}$ is not liquidatable.

As an explicit example, suppose that we have linear price impact, \ie~$p^{exec}=p_t\mp \tfrac{\alpha}{2}\,\Delta q$.
We receive a quadratic equation in $\Delta q$:
\[
\tfrac{\alpha}{2}\,\Delta q^2\; -\; \mu p_t\,\Delta q\; +\; (\mu p_t q - e_T + \tau)\;=\;0,
\]
The greedy liquidation strategy simply takes the smallest feasible root of this quadratic and uses it as the liquidation quantity $\Delta q$.
We refer the reader to Appendix~\ref{app:liquidation-mechanics} for detailed numerical examples of this strategy and related liquidation procedures.

\paragraph{Bad Debt.}
A liquidation for position $\mathfrak{p}_{i,t}$ generates \emph{bad debt} if the post-liquidation equity (including execution price and fees) is non-positive.
Formally, we define an adjusted terminal equity $\tilde{e}_T(\mathfrak{p}_{i,t}, p_{1:T}, \hat{p}_{1:T}, \Delta q)$ that incorporates liquidation costs:
\begin{equation}\label{eq:adjusted-equity}
\tilde e_T\;=\; e_T\; +\; b\,\Delta q\,\big(p^{exec}-p_t\big)\; -\; \tau_t(\Delta q)
\end{equation}
Intuitively, relative to marking the entire position at $p_t$, realizing a slice $\Delta q$ at $p^{exec}$ changes equity by the slippage term $b\,\Delta q\,(p^{exec}-p_t)$, and fees reduce equity via $\tau_t(\Delta q)$.
We say that a liquidation creates bad debt if $\tilde{e}_T(\mathfrak{p}_{i,t}, p_{1:T}, \hat{p}_{1:T}, \Delta q) < 0$, \ie~the liquidation leaves a residual liability for the venue.

Note that in practice, a liquidation strategy might retry or reattempt to liquidate the position repeatedly.
We elide formulating the details of such a liquidation strategy here for simplicity, but note that our model can be easily extended to account for this.
We define the \emph{total bad debt} or \emph{shortfall} of an exchange, $D_t$ given a liquidation strategy $L$ is,
\begin{equation}\label{eq:total-bad-debt}
    D_t(L) = \sum_{\mathfrak{p} \in \mathcal{P}_n} \max(0, -\tilde{e}(\mathfrak{p}_{i,t}, p_{1:T}, \hat{p}_{1:T}, \Delta q_i)) = \sum_{\mathfrak{p} \in \mathcal{P}_n} -\tilde{e}(\mathfrak{p}_{i,t}, p_{1:T}, \hat{p}_{1:T}, \Delta q_i)_-
\end{equation}
where $\Delta q_i = L(\mathfrak{p}_{i,t}, p_{1:T}, \hat{p}_{1:T})$.
This represents the shortfall that the exchange must cover to be solvent (\S\ref{subsec:exchange-solvency}).
Detailed numerical examples of bad debt are provided in Appendix~\ref{app:liquidation-mechanics}.

\begin{figure}[!ht]
  \centering
  \begin{subfigure}{0.9\textwidth}
    \centering
    \begin{tikzpicture}[x=1cm,y=3.4cm]
      \def\barw{0.35}
      \draw[->,thick] (0.5,-0.32) -- (0.5,0.52) node[left,font=\scriptsize] {$e_{T,i}$};
      \draw[->,thick] (0.5,0) -- (6.6,0) node[below right,font=\scriptsize] {equity ranking ($k \rightarrow 1$)};
      \draw[dashed,thick,NavyBlue] (0.5,0.10) -- (6.5,0.10);
      \node[font=\scriptsize,anchor=south west,text=NavyBlue] at (0.6,0.11) {$10\%$ liquidation threshold};
      \draw[densely dotted] (0.5,0) -- (6.5,0);
      \node[font=\scriptsize,anchor=east] at (0.48,-0.01) {bankruptcy};
      \foreach \i/\lab in {1/k,2/{k-1},3/{k-2},4/{\cdots},5/2,6/1}{
      \path[fill=Goldenrod!25,draw=Goldenrod!80!black,densely dashed] (5.5-\barw,0) rectangle (5.5+\barw,0.20);
        \draw (\i+0.5,0) -- (\i+0.5,-0.015);
        \node[below,font=\scriptsize] at (\i+0.5,-0.02) {$\lab$};
      }
      \foreach \i/\val in {1/-0.25,2/-0.12,3/0.05,4/0.18,5/0.30,6/0.45}{
        \pgfmathsetmacro{\barheight}{\val}
        \pgfmathsetmacro{\xL}{0.5+\i-\barw}
        \pgfmathsetmacro{\xR}{0.5+\i+\barw}
        \ifdim \barheight pt>0pt
          \path[fill=ForestGreen!65,draw=ForestGreen!80!black] (\xL,0) rectangle (\xR,\barheight);
        \else
          \path[fill=BrickRed!70,draw=BrickRed!80!black] (\xL,0) rectangle (\xR,\barheight);
        \fi
      }
      \path[fill=Orange!70,draw=OrangeRed!80] (3.5-\barw,0) rectangle (3.5+\barw,0.05);
      \node[font=\scriptsize,text=OrangeRed,anchor=south] at (3.5,-0.28) {liquidatable};
      \node[font=\scriptsize,anchor=south] at (3.4,0.47) {$p_t = 0.88\,p_0$};
      \node[font=\scriptsize,anchor=south] at (5.35,0.47) {$\hat{p}_t = 0.80\,p_0$};
      \node[font=\scriptsize,anchor=west] at (0.6,-0.32) {$e_{T,k}<0$};
    \end{tikzpicture}
    \caption{Baseline 10\% liquidation buffer: mark price $p_t=0.88p_0$ is 10\% above bankruptcy $\hat{p}_t=0.80p_0$.}
  \end{subfigure}

  \vspace{0.75em}

  \begin{subfigure}{0.9\textwidth}
    \centering
    \begin{tikzpicture}[x=1cm,y=3.4cm]
      \def\barw{0.35}
      \draw[->,thick] (0.5,-0.32) -- (0.5,0.52) node[left,font=\scriptsize] {$e_{T,i}$};
      \draw[->,thick] (0.5,0) -- (6.6,0) node[below right,font=\scriptsize] {equity ranking ($k \rightarrow 1$)};
      \draw[dashed,thick,NavyBlue] (0.5,0.10) -- (6.5,0.10);
      \node[font=\scriptsize,anchor=south west,text=NavyBlue] at (0.6,0.15) {$10\%$ liquidation threshold};
      \draw[densely dotted] (0.5,0) -- (6.5,0);
      \node[font=\scriptsize,anchor=east] at (0.48,-0.01) {bankruptcy};
      \foreach \i/\lab in {1/k,2/{k-1},3/{k-2},4/{\cdots},5/2,6/1}{
      \path[fill=Goldenrod!25,draw=Goldenrod!80!black,densely dashed] (5.5-\barw,0) rectangle (5.5+\barw,0.16);
        \draw (\i+0.5,0) -- (\i+0.5,-0.015);
        \node[below,font=\scriptsize] at (\i+0.5,-0.02) {$\lab$};
      }
      \foreach \i/\val in {1/0.08,2/0.12,3/0.18,4/0.24,5/0.32,6/0.42}{
        \pgfmathsetmacro{\barheight}{\val}
        \pgfmathsetmacro{\xL}{0.5+\i-\barw}
        \pgfmathsetmacro{\xR}{0.5+\i+\barw}
        \ifdim \barheight pt>0pt
          \path[fill=ForestGreen!65,draw=ForestGreen!80!black] (\xL,0) rectangle (\xR,\barheight);
        \else
          \path[fill=BrickRed!70,draw=BrickRed!80!black] (\xL,0) rectangle (\xR,\barheight);
        \fi
      }
      \draw[->,thick,NavyBlue] (5.2,0.015) -- (5.2,0.11);
      \node[font=\scriptsize,anchor=south] at (3.2,0.47) {$p_{t+\Delta} = 0.94\,p_0$};
      \node[font=\scriptsize,anchor=south] at (5.45,0.47) {$\hat{p}_{t+\Delta} = 0.93\,p_0$};
      \node[font=\scriptsize,anchor=west,text=ForestGreen!70!black] at (0.6,0.45) {all positions solvent};
    \end{tikzpicture}
    \caption{Price recovery: increasing the mark price by $\Delta p_t=+0.06p_0$ and the bankruptcy level by $\Delta \hat{p}_t=+0.13p_0$ moves them to $p_{t+\Delta}=0.94p_0$ and $\hat{p}_{t+\Delta}=0.93p_0$, leaving every account with positive equity and a residual mark--bankruptcy buffer of about $1\%$.}
  \end{subfigure}

  \vspace{0.75em}

  \begin{subfigure}{0.9\textwidth}
    \centering
    \begin{tikzpicture}[x=1cm,y=3.4cm]
      \def\barw{0.35}
      \draw[->,thick] (0.5,-0.38) -- (0.5,0.52) node[left,font=\scriptsize] {$e_{T,i}$};
      \draw[->,thick] (0.5,0) -- (6.6,0) node[below right,font=\scriptsize] {ranking ($k \rightarrow 1$)};
      \draw[densely dotted] (0.5,0) -- (6.5,0);
      \node[font=\scriptsize,anchor=east] at (0.48,-0.01) {bankruptcy};
      \foreach \i/\lab in {1/k,2/{k-1},3/{k-2},4/{\cdots},5/2,6/1}{
        \draw (\i+0.5,0) -- (\i+0.5,-0.015);
        \node[below,font=\scriptsize] at (\i+0.5,-0.02) {$\lab$};
      }
      \foreach \i/\val in {1/-0.32,2/-0.24,3/-0.15,4/0.04,5/0.22,6/0.36}{
        \pgfmathsetmacro{\barheight}{\val}
        \pgfmathsetmacro{\xL}{0.5+\i-\barw}
        \pgfmathsetmacro{\xR}{0.5+\i+\barw}
        \ifdim \barheight pt>0pt
          \path[fill=ForestGreen!65,draw=ForestGreen!80!black] (\xL,0) rectangle (\xR,\barheight);
        \else
          \path[fill=BrickRed!70,draw=BrickRed!80!black] (\xL,0) rectangle (\xR,\barheight);
        \fi
      }
      \draw[->,very thick,BrickRed!90] (2,-0.05) -- (2,-0.3);
      \node[font=\scriptsize,anchor=south] at (3.4,0.47) {$p_t = 0.88\,p_0$};
      \node[font=\scriptsize,anchor=south] at (5.35,0.47) {$\hat{p}_t = 0.80\,p_0$};
      \node[font=\scriptsize,anchor=west] at (1.95,-0.34) {negative equity (insolvent)};
    \end{tikzpicture}
    \caption{Example of a negative-equity account inside the sorted stack.}
  \end{subfigure}

  \caption{Sorted equity profiles for stylized liquidation examples. Negative positions (red) appear on the left, positive positions (green) on the right. Dashed lines highlight the bankruptcy level and liquidation triggers.}
  \label{fig:liquidation-examples}
\end{figure}
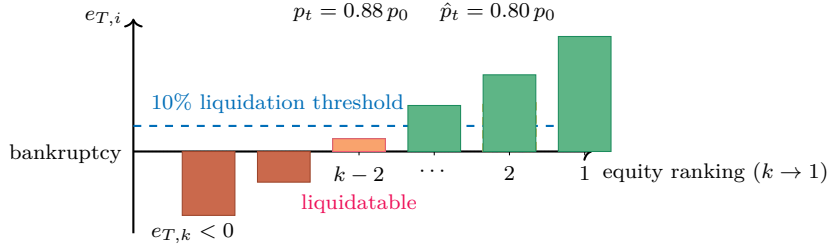
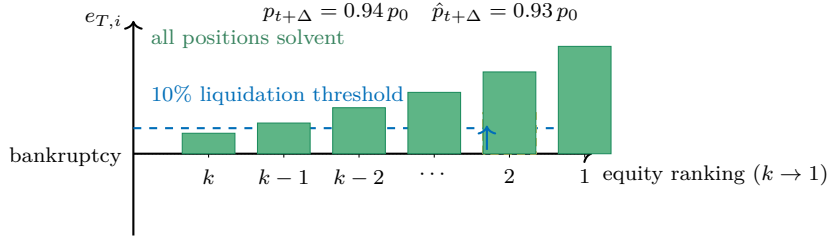
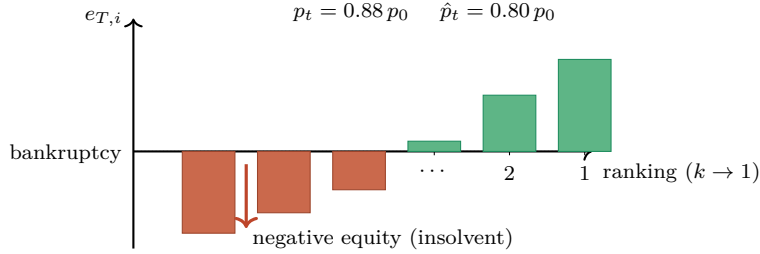

\paragraph{Anatomy of a Liquidation.}
Given the bankruptcy, liquidation, and execution prices, we can now describe the high-level algorithm that liquidations follow (see Appendix~\ref{app:liquidation-mechanics}).
We note that many live liquidation systems will have much more complex liquidation algorithms.
These complexities deal with the coordination costs of coordinating many parties (\eg~oracle provider, liquidators, spot order book liquidity) and precise models that exchanges use for their liquidation strategy.
However, we effectively lump all of these complexities into the definition of the liquidation strategy.
A detailed algorithmic description of the liquidation loop and a step-by-step example are provided in Appendix~\ref{app:liquidation-mechanics}.

\subsection{Exchange Solvency}\label{subsec:exchange-solvency}
In the previous section, we focused on liquidations as a means to remove individual trader solvencies.
However, there is also a `macroscopic' notion of solvency at the exchange level.
The goal of the exchange is to make sure that traders can realize their full profits and losses given the cash collateral held at the exchange.
In rare cases, an exchange might not be able to successfully execute a liquidation strategy, leading to the exchange being unable to pay unrealized but earned profits to some users.

In this section, we define solvency for exchanges and how exchanges use \emph{insurance funds} to try to ensure solvency.
An exchange's insurance fund is the first line of defense against insolvency.
However, if the insurance fund has insufficient balance, the exchange will need to utilize a autodeleveraging mechanism to socialize losses.
In this section, we will define solvency for exchanges and how insurance funds are constructed.

\paragraph{Exchange Solvency.}
The \emph{solvency} of an exchange is simply the total equity of all positions:
\begin{equation}\label{eq:solvency-condition}
\mathsf{Solv}_T(\mathcal{P}_n) = \sum_{\mathfrak{p} \in \mathcal{P}_n} e_T(\mathfrak{p}) = \sum_{\mathfrak{p} \in \mathcal{P}_n} c_{\mathfrak{p}} + \mathsf{PNL}_T(\mathfrak{p}, p_{1:T}, \hat{p}_{1:T})
\end{equation}
An exchange is said to be \emph{insolvent} if the following condition holds $\mathsf{Solv}_T(\mathcal{P}_n) \leq 0$.
The insolvency condition corresponds to when the total cash held by the exchange on behalf of users is less than the total accrued profits.
This mean that if this insolvency condition is hit, the exchange will not be able to payout all profitable traders when they withdraw from the exchange.
Note that this is a global condition on $\mathcal{P}_n$ as opposed to a per user constraint.

Most exchanges have a solvency threshold, $\delta > 0$, such that if it takes a global loss of $\delta$ units of equity to become truly insolvent, then the exchange is deemed approximately insolvent.
This notion of approximate insolvency is useful in practice, as it allows an exchange to have a risk parameter for tuning how aggressively it enters the ADL regime.
Formally, we define an exchange to be $\delta$-insolvent if the following condition holds $\mathsf{Solv}_{T}(\mathcal{P}_n) \leq \delta$.

\paragraph{Comparison to overcollateralized lending.}
In overcollateralized lending, where a user uses one crypto asset as collateral to borrow another crypto asset, there is a similar notion of a user's equity position.
However, in overcollateralized lending, the user's equity position is not global, but rather is local to the user.
The goal of a lending protocol is to ensure that the equity $e_{lend}(p) \geq 0$ for \emph{every} user's position $p$.
Ensuring per account solvency condition inherently forces overcollateralized lending protocols to offer far lower leverage than perpetuals exchanges.
From a mathematical standpoint, perpetual solvency is easier to satisfy, allowing for higher leverage --- there is only one global constraint that has to hold versus $O(n)$ local constraints for $n$ users.
However, in perpetuals exchanges, users who are solvent can face losses from users who are insolvent which cannot happen in isolated overcollateralized lending markets like Morpho. 

\paragraph{Insurance Funds.}
Many exchanges maintain an insurance fund to cover bad debt that arises from failed liquidations.
When a positive shortfall $D_t > 0$ is realized, the insurance fund is used to reduce $D_t$ to as close to zero as possible.
The insurance fund is typically funded using a portion of the exchange's revenue from transaction, liquidation, and funding fees.

Insurance reserves in crypto funds were first popularized alongside perpetual futures in the mid‑2010s (notably on BitMEX)~\citep{Soska2021BitMEX}.
A majority of perpetuals venues (centralized and on‑chain) maintain such reserves (\eg~BitMEX~\citep{BitMEXADL}, Binance~\citep{BinanceFuturesInsuranceFund}, Drift~\citep{DriftADLCode}, Bybit~\citep{BybitInsuranceFund}).
To illustrate the size of these funds, we note that mid-2025 estimates place OKX and Bybit insurance funds at roughly \$300M and \$150M, respectively~\citep{OKXBybitInsuranceComparison}.

We will first provide a simple model for how the insurance fund evolves over time.
Let $\mathsf{IF}_t\ge 0$ denote the insurance\-fund balance at time $t$.
Furthermore, let $\mathcal{T}_t$ be the multiset of liquidations at time $t$ with sizes $\{\Delta q_j\}_{j\in\mathcal{T}_t}$ and execution costs $\tau_t(\Delta q_j)\in\reals_+$ (\cf~eq.~\eqref{eq:adjusted-equity}).
Let $V_t\ge 0$ denote the gross traded volume (absolute quantity) in the contract at time $t$.
We introduce three parameters to model insurance fund growth: $\alpha, \beta, \eta \in [0, 1]$.
The $\alpha, \beta, \eta$ parameters controls how much of liquidation, funding, and trading fee revenue is added $\mathsf{IF}_t$, respectively.

Given these parameters, the time evolution of the insurance fund is given by:
\begin{equation}
\mathsf{IF}_{t+1}
= \mathsf{IF}_t
+ \alpha\sum_{j\in\mathcal{T}_t} \tau_t(\Delta q_j)
+ \eta\,p_t V_t
+ \beta\sum_{i=1}^n |\gamma_t|\,p_t\,q_i
- \min\{\mathsf{IF}_t,\, D_t\}
\label{eq:if-process}
\end{equation}
Note that the positive contributions represent the revenue sharing terms whereas the negative contribution represents the cover of a shortfall using the insurance fund.

The main quantity of interest, given $\mathsf{IF}_t$, is the \emph{residual shortfall}.
This represents the amount of bad debt that cannot be covered by the insurance fund and needs to be socialized via ADL (\S\ref{subsec:adl}).
Formally, we define the residual shortfall as:
\begin{equation}
R_t=D_t-\min\{\mathsf{IF}_t,\,D_t\}=(D_t-\mathsf{IF}_t)_+
\label{eq:adl-residual}
\end{equation}
We note that when $R_t > 0$, an ADL mechanism is triggered.
We term times $t$ such that $R_t > 0$ as~\emph{breach times}.
The frequency of breach events is referred to as~\emph{breach frequency} or \emph{breach rate}.

\iparagraph{Example.}
Continuing from the previous liquidation example, the bad debt realized at $t=4$ by $\mathfrak{p}_E$ is
\[
 D_4\;=\;(p^{exec}_E-p^{bk}(\mathfrak p_{E,4}))\,\Delta q_E\;\approx\;0.399.
\]
Let $\mathsf{IF}_4$ denote the pre-coverage fund balance. The coverage is $\min\{\mathsf{IF}_4, D_4\}$ and the ADL residual is $R_4=(D_4-\mathsf{IF}_4)_+$.
We assume that the fund evolves via eq.~\eqref{eq:if-process} and model liquidation fees as affine in notional size:
\[
\tau_t(\Delta q) = \tau^{fix} + \phi^{mark}\,p_t\,\Delta q + \phi^{exec}\,p^{exec}_t\,\Delta q,
\]
where $\tau^{fix}$ is a fixed fee and $\phi^{mark}, \phi^{exec}$ are proportional rates on mark and execution price, respectively.
Using the Hyperliquid fee schedule $(\tau^{fix},\phi^{mark},\phi^{exec})=(0,20\,\mathrm{bps},10\,\mathrm{bps})$ (see \S\ref{subsec:liquidation-mechanics}), the fees are
\[
 \tau_4(\Delta q_E)
 =0 + 20\cdot 10^{-4}\cdot p_4\,\Delta q_E + 10\cdot 10^{-4}\cdot p^{exec}_E\,\Delta q_E
 \approx0.00365.
\]
Assuming $\gamma_4=0$ and neglecting trading volume $V_4$ for simplicity:
\begin{itemize}[leftmargin=12pt]
  \item \emph{Sufficiently sized fund ($R_4=0$).} If $\mathsf{IF}_4\ge D_4$, the debt is fully covered. No ADL occurs. The balance updates to
  \[
  \mathsf{IF}_5=\mathsf{IF}_4+\alpha\,\tau_4(\Delta q_E)-D_4.
  \]
  \item \emph{Insufficiently sized fund ($R_4>0$).} If $\mathsf{IF}_4 < D_4$, the fund is depleted to zero and pays $\mathsf{IF}_4$. The residual $R_4 = D_4-\mathsf{IF}_4$ is socialized via ADL (\S\ref{subsec:adl}). The new balance is simply the inflows:
  \[
  \mathsf{IF}_5=\alpha\,\tau_4(\Delta q_E).
  \]
\end{itemize}
Increasing $\alpha$, $\beta$, or $\eta$ builds $\mathsf{IF}_t$ faster, reducing the probability of $R_t > 0$.

\paragraph{Optimal Insurance Fund Size.}
A natural question to ask is, what is the optimal size for an insurance fund?
While this question is nuanced in practice, we provide a simple mean-field model that provides some intuition.
In particular, we determine the optimal fund size $K^*$ by minimizing a total cost function that balances the opportunity cost of capital against the expected franchise damage from ADL events.
Let $r > 0$ denote the opportunity cost of holding capital (\eg~the risk-free rate) and $\kappa > 0$ denote the reputation cost per unit of socialized loss (reflecting lost future volume).
The exchange minimizes
\[
\min_{K \ge 0} \; r K + \kappa \, \Expect[(D_T - K)_+].
\]
This is a classic newsvendor-type problem~\citep{ArrowHarrisMarschak1951}.
Classical result show that the optimal size is the Value-at-Risk of the deficit distribution at quantile $1 - r/\kappa$:
\[
K^* \;=\; \text{VaR}_{1 - r/\kappa}(D_T) \;=\; \inf \{ x : \Prob(D_T \le x) \ge 1 - r/\kappa \}.
\]
For completeness, we provide a derivation in Appendix~\ref{app:optimal-capital}.
Intuitively, as the reputation cost $\kappa$ increases relative to the cost of capital $r$, the exchange holds a larger buffer to push ADL events further into the tail.

\subsection{Autodeleveraging}\label{subsec:adl}
Any mechanism for reducing or closing profitable user positions to reduce an exchange's insolvency is termed \emph{autodeleveraging} (ADL).
Mathematically, ADL mechanisms attempt to liquidating positions in order to reduce the residual $R_t$.
This inherently means that users with positive equity (\eg~are profitable) will be forced to close their positions and not realize their full profit.

\paragraph{Contract space vs wealth space.}
Production ADL is executed in \emph{contract space}: the engine selects positions by a ranking score and forces contract-level reductions.
The theoretical analysis in this paper is written in \emph{wealth space} (equity haircuts and haircutable endowment).
To compare production ADL to model ADL, we need an explicit mapping from contract closures to wealth-space impacts.
This mapping is non-trivial: a close in contracts is not the same statement as haircutting 100\% of equity, because the economic effect depends on prices and on how much of the account's equity is attributable to the seized contracts.
Production ADL induces wealth-space haircuts only through the observation-model mapping (see~\S\ref{sec:numerics}).
It is critical to distinguish between:
\begin{itemize}
  \item \emph{Wealth-space queue} (theory object): a greedy allocation rule that ranks positions by a score and allocates haircuts in endowment space until the budget is met.
  \item \emph{Contract-space queue} (production): Hyperliquid's implementation ranks positions and closes contracts; this induces wealth-space haircuts only through the observation-model mapping (two-pass replay comparing ADL-on vs ADL-off equity outcomes).
  \item \emph{Wealth-pro-rata} (theory): each winner loses the same fraction of their haircutable endowment, defined in~\eqref{eq:pro-rata}.
  \item \emph{Contracts-pro-rata} (production): some exchanges allocate losses proportional to contract size; this is a different mechanism with different fairness properties.
\end{itemize}
Without this separation, readers conflate wealth-space theory results with contract-space implementation and misread empirical overshoot as a production ledger claim.
These are different objects: the theory queue is an abstraction that shares the same ranking idea as production, but does not reproduce contract-level execution details.

\paragraph{Equity decomposition and haircut numéraire.}
We decompose each position's equity into protected cash/principal and profit-and-loss:
\begin{equation}\label{eq:equity-decomp}
e_{i,t} = c_{i,t} + \mathrm{PNL}_{i,t},
\end{equation}
where $c_{i,t} \ge 0$ is the protected cash/principal (returned on forced close; not haircutable in production) and $\mathrm{PNL}_{i,t}$ is the profit-and-loss.
We define the \emph{haircutable endowment} $w_{i,t} \ge 0$ as the only quantity ADL may seize from position $i$ at time $t$.
In the corrected measurement model used throughout this paper, we specialize to \emph{PNL-only} haircuts:
\begin{equation}\label{eq:endowment-PNL}
w_{i,t} = (\mathrm{PNL}_{i,t})_+.
\end{equation}
This reflects the production reality that principal/collateral is returned when a position is force-closed; counting it as ``haircut capacity'' would be a numéraire error.
(For theoretical comparison, one may also consider an equity-haircut abstraction $w_{i,t} = (e_{i,t})_+$, but this is explicitly labeled as a theoretical simplification, not production.)

\paragraph{Deficit and capacity definitions.}
We define four key quantities.
The \emph{total shortfall} (deficit) and \emph{maximum shortfall} are computed from loser-side negative equity:
\begin{align*}
D_T(\mathcal{P}_n) &= \sum_{\mathfrak{p} \in \mathcal{P}_n} \big(-e_T(\mathfrak{p})\big)_+ \geq 0,  &&& \Delta_T(\mathcal{P}_n) &= \max_{\mathfrak{p} \in \mathcal{P}_n} \big(-e_T(\mathfrak{p})_-\big) \geq 0.
\end{align*}
The \emph{total haircut capacity} and \emph{maximum haircut capacity} are computed from the haircutable endowment:
\begin{align*}
U_T(\mathcal{P}_n) &= \sum_{\mathfrak{p} \in \mathcal{P}_n} w_T(\mathfrak{p}) \geq 0, &&& \upsilon_T(\mathcal{P}_n) &= \max_{\mathfrak{p} \in \mathcal{P}_n} w_T(\mathfrak{p}) \geq 0.
\end{align*}
For notational simplicity, we will use Roman letters for total quantities and Greek letters for maximum or extreme value quantities.
The shortfall represents the total amount of negative equity (loser-side), whereas the capacity represents the total amount of haircutable endowment available from winners.
By definition, $\mathsf{Solv}_T(\mathcal{P}_n) = \sum_{\mathfrak{p}} e_T(\mathfrak{p}) = W_T(\mathcal{P}_n) - D_T(\mathcal{P}_n)$, where $W_T(\mathcal{P}_n) = \sum_{\mathfrak{p}} (e_T(\mathfrak{p}))_+$ is the total winner equity.
However, under PNL-only haircuts, solvency ($\mathsf{Solv}_T > 0$) does \emph{not} imply the deficit is coverable by ADL.
The correct one-round feasibility condition is $\theta_T D_T \le U_T$ (not $\theta_T D_T \le W_T$).
There is an intermediate regime where $D_T < W_T$ but $D_T > U_T$: the exchange is solvent in gross equity but cannot clear the deficit without touching principal.

Instead, ADL mechanisms attempt to liquidate a fraction of the shortfall, then allow the market to react for some time (\eg~for users to close positions or post more margin and spot prices to update), before attempting a future deleveraging (if needed).
This means that the ADL process should be viewed probabilistically, as a sequence of decisions that are dynamically adjusted based on market conditions.
We will demonstrate that from this perspective, ADL can be formulated as a standard reinforcement learning problem. 

\paragraph{ADL Policies.}
The fundamental object of reinforcement learning is the policy, which maps states of the world to actions.
For ADL, the positions $\mathcal{P}_n$ serve as the state space, whereas the actions are the fraction of shortfall to socialize and which positions are socialized.
While the majority of exchanges (\eg~Binance and Hyperliquid) use queue-based policies that greedily rank positions by PNL and leverage, we will define a formalism over a larger class of potential ADL policies.
This larger class will enable us to find more efficient and fair ADL policies than queue-based models and includes pro-rata ADL policies used by smaller exchanges such as Drift and Paradex.

An \emph{ADL policy} $\pi(\mathcal{P}_n)$ maps a perpetuals exchange to a fraction of shortfall to socialize, $\theta \in [0, 1]$, and a set of seized amounts $x_{\pi,i} \in [0, w_{i,T}]$ (or equivalently, haircut fractions $h_{\pi,i} \in [0,1]$ such that $x_{\pi,i} = h_{\pi,i} w_{i,T}$).
We term the fraction $\theta$ the \emph{severity} of the ADL policy.
For notational convenience we define $\theta_{\pi}, h_{\pi}$ as the severity and haircuts induced by an ADL policy $\pi$.
We say an ADL policy $\pi$ is \emph{valid} if the following constraints always hold:
\begin{equation}\label{eq:budget-balance}
\sum_{\mathfrak{p} \in \mathcal{P}_n} h_{\pi,i} w_T(\mathfrak{p})
= \sum_{\mathfrak{p} \in \mathcal{P}_n} x_{\pi,i}
= \theta_{\pi} D_T(\mathcal{P}_n).
\end{equation}
\begin{equation}\label{eq:feasibility}
\theta_{\pi} D_T(\mathcal{P}_n) \leq U_T(\mathcal{P}_n) = \sum_{\mathfrak{p} \in \mathcal{P}_n} w_T(\mathfrak{p}).
\end{equation}
The budget balance constraint~\eqref{eq:budget-balance} ensures that when an ADL policy is executed, exactly $\theta$\% of the deficit is socialized from the winners' haircutable endowment.
The feasibility constraint~\eqref{eq:feasibility} ensures that $\theta_{\pi}$ is chosen such that the required seizure does not exceed the total haircut capacity $U_T$.
Note that under PNL-only haircuts ($w = (\mathrm{PNL})_+$), feasibility is stricter than solvency: even if $\mathsf{Solv}_T = W_T - D_T > 0$, we may have $D_T > U_T$, meaning the deficit cannot be covered without touching principal.

We define the post-policy deficit $D^{\pi}_T$, capacity $U^{\pi}_T$, and max-winner endowment $\upsilon^{\pi}_T$ as
\begin{align*}
D^{\pi}_T &= \theta_{\pi} D_T(\mathcal{P}_n), &&& \Delta^{\pi}_T &= \theta_{\pi} \Delta_T(\mathcal{P}_n), \\
U^{\pi}_T &= \sum_{\mathfrak{p} \in \mathcal{P}_n} (w_T(\mathfrak{p}) - x_{\pi,i})_+ = \sum_{\mathfrak{p} \in \mathcal{P}_n} (1-h_{\pi,i}) w_T(\mathfrak{p}), &&& \upsilon^{\pi}_T &= \max_{\mathfrak{p} \in \mathcal{P}_n} (w_T(\mathfrak{p}) - x_{\pi,i})_+.
\end{align*}
Our goal is to construct policies that minimize an objective function that balances individual user profits with global exchange solvency.

Given an ADL policy outputting $(\theta_{\pi}, h_{\pi})$ (or equivalently, seized amounts $x_{\pi,i}$), the post-ADL equity for winners is:
\begin{equation}\label{eq:post-adl-equity}
e'_{i,T} = c_{i,T} + \mathrm{PNL}_{i,T} - x_{\pi,i} = c_{i,T} + \mathrm{PNL}_{i,T} - h_{\pi,i} w_{i,T}.
\end{equation}
For PNL-only haircuts ($w_{i,T} = (\mathrm{PNL}_{i,T})_+$), this becomes:
\begin{equation}\label{eq:post-adl-equity-PNL}
e'_{i,T} = c_{i,T} + \mathrm{PNL}_{i,T} - h_{\pi,i} (\mathrm{PNL}_{i,T})_+.
\end{equation}
For instance, if we have equities $e_T = (10, 5, 1, -3, -12)$ with cash components $c_T = (2, 1, 0.5, 0, 0)$ and PNL components $\mathrm{PNL}_T = (8, 4, 0.5, -3, -12)$, then under PNL-only haircuts with $w_T = (8, 4, 0.5, 0, 0)$ and policy outputs $h = (0.5, 1, 0, 0, 0)$ (seizing $x = (4, 4, 0, 0, 0)$), the post-ADL equities are $e'_T = (2+8-4, 1+4-4, 0.5+0.5-0, -3, -12) = (6, 1, 1, -3, -12)$.

\paragraph{Examples of ADL Policies.}
We will first provide two canonical examples: \emph{queueing} (or leverage ranking) and \emph{pro-rata}.
These examples represent the solvency policies of virtually all live perpetuals exchanges as of November 2025.
For both of these policies, the choice of $\theta_{\pi}$ is independent of the choice of $h_{\pi}$. 

\iparagraph{Queue.}
The policy that the largest centralized exchange, Binance, and the largest decentralized exchange, Hyperliquid, use is the Queue policy.
The implementation at these exchanges uses the PNL-leverage ranking.
However, the Queue model is somewhat more generic, as we define it below.

At a high level, the Queue strategy is a greedy algorithm that ranks positions using a \emph{score function}, $s_T$.
Once positions are ordered by score, positions are closed in order of score from highest to lowest, until the deficit is fully covered.
We will denote this strategy as $\pi_Q$ and define it via the algorithm that implements it.

Given a haircut budget $B = \theta_{\pi_Q} D_T(\mathcal{P}_n)$, the algorithm to choose $h_i$ (or equivalently, seized amounts $x_i = h_i w_{i,T}$) works as follows:
\begin{enumerate}
\item For each position $\mathfrak{p} = (q, c, \ell, t, b) \in \mathcal{P}_n$, define the score\footnote{We note that technically, many exchanges define the leverage used in the score differently (\ie~instead of the score being linear in $\ell$, some exchanges multiply by $\frac{\ell p q}{pq+c})$. This does not change our main results and mainly adds technical complications to the example.} $s_T(\mathfrak{p}, p^{ref})$. For the PNL-leverage ranking, we have $s_T(\mathfrak{p}, p^{ref})= \ell\, \frac{\hat p_{T}}{p^{ref}}$.
\item Let $\sigma \in \mathsf{Perm}(\mathcal{P}_n)$ be any permutation of the positions ranked by decreasing score $s_T(\mathfrak{p}, p^{ref})$
\item For $i \in [n]$, define $A_i = (B - \sum_{j=1}^{i} w_T(\mathfrak{p}_{\sigma(j)}))_+$, where $w_T(\mathfrak{p}_{\sigma(j)})$ is the haircutable endowment of the $j$th largest position under the ranking $\sigma$
\item Define $h \in [0,1]^n$ (or equivalently, $x \in [0, w_T]^n$) as
\begin{equation}\label{eq:pl-ranking}
h_{\pi_Q,\sigma(i)} = \begin{cases}
    1 & \text{if } A_{i-1} - A_{i} = w_T(\mathfrak{p}_{\sigma(i)}),\; A_i > 0 \\
\frac{A_{i-1}}{w_T(\mathfrak{p}_\sigma(i))} & \text{if } A_i = 0, A_{i-1} > 0 \\
0 & \text{otherwise}
\end{cases}
\end{equation}
\end{enumerate}
Note this definition of $h_{\pi_Q}$ ensures that $\sum_{i=1}^n h_{\pi_Q,i} w_T(\mathfrak{p}_i) = B$, satisfying~\eqref{eq:budget-balance}.
On the other hand, since we only subtract from positions with positive endowment (\eg~$j$ such that $w_T(\mathfrak{p}_{\sigma(j)}) > 0$), we satisfy~\eqref{eq:feasibility} by construction.

There is one remaining question to address: how should we interpret the score $s_T(\mathfrak{p}, p^{ref})$?
Most exchanges justify this form for the score by arguing that it can represent the risk a single position holds, so that the exchange ranks positions to ADL by risk-level.
Winning, positive equity positions that are autodeleveraged first will tend to be higher risk positions (\ie~used more leverage for their winnings).
The price terms in the numerator, $p_{\hat{T}}$, represents the last mark price of the position whereas the parameter $p^{ref}$ represents the initial or opening price of the user.

Binance chooses the scoring parameter $p^{ref} = p^{bk}$ of the position,~\ie $s^{\text{Binance}}_T(\mathfrak{p}) = \ell \frac{p_{\hat{T}}(\mathfrak{p})}{p^{bk}(\mathfrak{p})}$.
Hyperliquid, on the other hand, chooses the parameter $p^{ref}$ to be the entry price, $p^{ref} = p_{t_i}$.
In other words, Binance's choice effectively ranks a position based on the maximum possible PNL that a user could realize whereas Hyperliquid ranks based on the current return on equity.

\iparagraph{Pro-Rata Rules.}
This haircut strategy is even simpler than the Queue and has some benefits in terms of fairness, as we will show in~\S\ref{subsec:fairness}.
The Pro-Rata strategy, simply put, haircuts each winning user with positive endowment by the same multiplicative fraction.
The fraction is subject to budget constraints, as we describe below.

\emph{Important distinction:} The pro-rata rule defined here is \emph{wealth-pro-rata} (losses proportional to endowment/equity).
Some exchanges (e.g., Drift) implement \emph{contracts-pro-rata} (losses proportional to position size in contracts).
These are different mechanisms with different fairness meanings and should not be conflated.
The paper's fairness analysis uses wealth-pro-rata as the benchmark.

Given a valid feasible shortfall $\theta_{\pi} D_T(\mathcal{P}_n) \leq U_T(\mathcal{P}_n)$, the wealth-pro-rata haircut policy is defined as:
\begin{equation}\label{eq:pro-rata}
h_{\pi_{PR}}(\mathfrak{p}) = \frac{\theta_{\pi} D_T(\mathcal{P}_n)}{U_T(\mathcal{P}_n)} \quad \text{for } w_T(\mathfrak{p}) > 0,
\end{equation}
or equivalently, $x_{\pi_{PR},i} = h_{\pi_{PR}} w_{i,T}$ where $h_{\pi_{PR}} = \theta_{\pi} D_T / U_T$ is the uniform haircut fraction.
This ranking simply says that the haircutable endowment is socialized proportional to its size.
Some exchanges slightly modify this formula to include leverage (\eg~Binance's ADL Priority Index includes leverage and unrealized PNL \cite{BinanceADL}; Aevo documentation discusses leverage‑weighted priority \cite{AevoADL}), which penalizes higher leverage positions much like the PNL‑leverage ranking:
\begin{equation}\label{eq:lev-pro-rata}
h_{\pi_{LPR}}(\mathfrak{p}) = \left(\frac{\ell w_T(\mathfrak{p})}{\sum_{\mathfrak{p} \in \mathcal{P}_n} \ell w_T(\mathfrak{p})}\right) \frac{\theta_{\pi} D_T(\mathcal{P}_n)}{\sum_{\mathfrak{p} \in \mathcal{P}_n} \ell w_T(\mathfrak{p})}
\end{equation}

\begin{figure}[t]
  \centering
  \begin{subfigure}{0.48\textwidth}
    \centering
    \resizebox{\linewidth}{!}{%
    \begin{tikzpicture}[x=1cm,y=3.6cm]
      \def\barw{0.35}
      \draw[->,thick] (0.5,-0.28) -- (0.5,0.52) node[left,font=\scriptsize] {$e_{T,i}$};
      \draw[->,thick] (0.5,0) -- (6.6,0) node[below right,font=\scriptsize] {ranking ($k \rightarrow 1$)};
      \draw[densely dotted] (0.5,0) -- (6.5,0);
      \foreach \i/\lab in {1/k,2/{k-1},3/{k-2},4/{\cdots},5/2,6/1}{
        \draw (\i+0.5,0) -- (\i+0.5,-0.015);
        \node[below,font=\scriptsize] at (\i+0.5,-0.02) {$\lab$};
      }
      \foreach \i/\val in {1/-0.28,2/-0.22,3/0.12,4/0.20,5/0.30,6/0.40}{
        \pgfmathsetmacro{\barheight}{\val}
        \pgfmathsetmacro{\xL}{0.5+\i-\barw}
        \pgfmathsetmacro{\xR}{0.5+\i+\barw}
        \ifdim \barheight pt>0pt
          \path[fill=ForestGreen!65,draw=ForestGreen!80!black] (\xL,0) rectangle (\xR,\barheight);
        \else
          \path[fill=BrickRed!70,draw=BrickRed!80!black] (\xL,0) rectangle (\xR,\barheight);
        \fi
      }
      \path[fill=RoyalPurple!55,opacity=0.65] (1.5-\barw,0) rectangle (1.5+\barw,-0.28);
      \path[fill=RoyalPurple!55,opacity=0.65] (2.5-\barw,0) rectangle (2.5+\barw,-0.22);
      \path[fill=blue!35,opacity=0.75] (6.5-\barw,0) rectangle (6.5+\barw,0.40);
      \path[fill=blue!35,opacity=0.75] (5.5-\barw,0.20) rectangle (5.5+\barw,0.30);
      \path[fill=blue!10,draw=blue!60!black,densely dashed] (5.5-\barw,0) rectangle (5.5+\barw,0.20);
      \path[fill=RoyalPurple!55,opacity=0.65] (0.65,0.42) rectangle (0.81,0.46);
      \node[font=\scriptsize,anchor=west] at (0.83,0.44) {deficit};
      \path[fill=blue!35,opacity=0.75] (0.65,0.36) rectangle (0.81,0.40);
      \node[font=\scriptsize,anchor=west] at (0.83,0.38) {haircut};
      \draw[decorate,decoration={brace,amplitude=4pt},thick] (6.5+\barw+0.1,0.40) -- (6.5+\barw+0.1,0) node[midway,right=3pt,font=\scriptsize,align=left] {queue\\haircut};
      \node[font=\scriptsize,anchor=north] at (2,-0.34) {severity $\theta=0.50$};
    \end{tikzpicture}
    }
    \caption{Queue ADL:\\ all haircuts fall on the highest-ranked winner.}
  \end{subfigure}
  \hfill
  \begin{subfigure}{0.48\textwidth}
    \centering
    \resizebox{\linewidth}{!}{%
    \begin{tikzpicture}[x=1cm,y=3.6cm]
      \def\barw{0.35}
      \draw[->,thick] (0.5,-0.28) -- (0.5,0.52) node[left,font=\scriptsize] {$e_{T,i}$};
      \draw[->,thick] (0.5,0) -- (6.6,0) node[below right,font=\scriptsize] {ranking ($k \rightarrow 1$)};
      \draw[densely dotted] (0.5,0) -- (6.5,0);
      \foreach \i/\lab in {1/k,2/{k-1},3/{k-2},4/{\cdots},5/2,6/1}{
        \draw (\i+0.5,0) -- (\i+0.5,-0.015);
        \node[below,font=\scriptsize] at (\i+0.5,-0.02) {$\lab$};
      }
      \foreach \i/\val in {1/-0.28,2/-0.22,3/0.12,4/0.20,5/0.30,6/0.40}{
        \pgfmathsetmacro{\barheight}{\val}
        \pgfmathsetmacro{\xL}{0.5+\i-\barw}
        \pgfmathsetmacro{\xR}{0.5+\i+\barw}
        \ifdim \barheight pt>0pt
          \path[fill=ForestGreen!65,draw=ForestGreen!80!black] (\xL,0) rectangle (\xR,\barheight);
        \else
          \path[fill=BrickRed!70,draw=BrickRed!80!black] (\xL,0) rectangle (\xR,\barheight);
        \fi
      }
      \path[fill=RoyalPurple!55,opacity=0.65] (1.5-\barw,0) rectangle (1.5+\barw,-0.28);
      \path[fill=RoyalPurple!55,opacity=0.65] (2.5-\barw,0) rectangle (2.5+\barw,-0.22);
      \path[fill=blue!35,opacity=0.75] (3.5-\barw,0.04) rectangle (3.5+\barw,0.12);
      \path[fill=blue!35,opacity=0.75] (4.5-\barw,0.08) rectangle (4.5+\barw,0.20);
      \path[fill=blue!35,opacity=0.75] (5.5-\barw,0.16) rectangle (5.5+\barw,0.30);
      \path[fill=blue!35,opacity=0.75] (6.5-\barw,0.24) rectangle (6.5+\barw,0.40);
      \path[fill=RoyalPurple!55,opacity=0.65] (0.65,0.42) rectangle (0.81,0.46);
      \node[font=\scriptsize,anchor=west] at (0.83,0.44) {deficit};
      \path[fill=blue!35,opacity=0.75] (0.65,0.36) rectangle (0.81,0.40);
      \node[font=\scriptsize,anchor=west] at (0.83,0.38) {haircut};
      \draw[decorate,decoration={brace,amplitude=4pt},thick] (6.5+\barw+0.1,0.42) -- (3.5-\barw-0.1,0.42) node[midway,above=3pt,font=\scriptsize] {pro-rata haircuts};
      \node[font=\scriptsize,anchor=north] at (2,-0.34) {severity $\theta=0.50$};
    \end{tikzpicture}
    }
    \caption{Pro-rata ADL:\\ haircuts are shared across the surviving winners.}
  \end{subfigure}

  \caption{ADL severity example comparing queue and pro-rata coloring. Purple shading equals the negative equity mass (deficit) while blue shading shows the haircut mass allocated to winning traders.
  The queue panel’s dashed blue block at rank~2 highlights residual equity when the queue method allows partial closures; exchanges that close winners fully (\eg~Hyperliquid) would shave this bar completely. 
  Haircut mass matches deficit mass in each panel, illustrating severity $\theta=0.50$.}
  \label{fig:adl-coloring}
\end{figure}
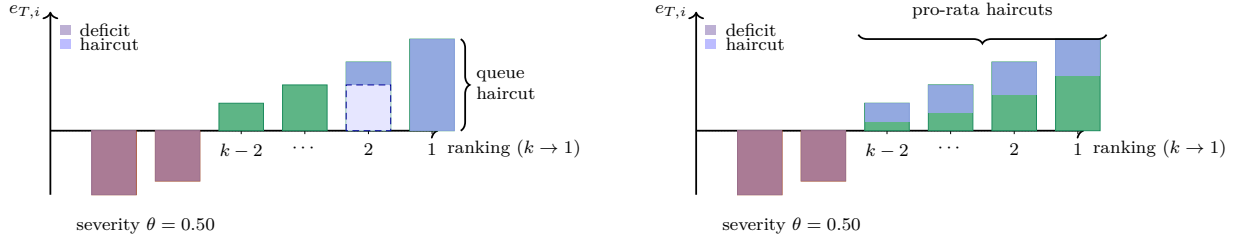

\iparagraph{Per-account constraints.}
Suppose that one wants to enforce some per-account constraints on the amount haircut.
For instance, suppose that we guarantee that for user $i$, their haircut $h_i$ always satisfies 
\begin{equation}\label{eq:haircut-constraint}
h_i \leq \overline{h}_i
\end{equation}
This corresponds to an exchange guaranteeing that a user will never lose more than $\overline{h}_i$\% of their haircutable endowment in a single ADL round.
Moreover, an exchange might also offer an absolute guarantee to users: if your equity is positive and sufficiently large, you will never have your equity cut to below some threshold $\underline{e}_i$.
This can be represented by the constraint 
\begin{equation}\label{eq:equity-constraint}
e'_{i,T} = c_{i,T} + \mathrm{PNL}_{i,T} - h_{\pi, i} w_{i,T} \geq \underline{e}_i
\end{equation}
For PNL-only haircuts ($w_{i,T} = (\mathrm{PNL}_{i,T})_+$), this becomes $c_{i,T} + \mathrm{PNL}_{i,T} - h_{\pi, i} (\mathrm{PNL}_{i,T})_+ \ge \underline{e}_i$.
Numerous exchanges offer one or both of these guarantees on a per‑ADL‑round basis, including Hyperliquid, BitMEX, and Aevo~\cite{HyperliquidDocsLiquidations,BitMEXADL,AevoADL}.

\paragraph{Invariance to protected cash.}
A key observation is that adding a common protected cash component $c$ to all winners is an affine shift that does not affect allocation-ordering theorems.
Specifically, if two policies $A$ and $B$ allocate seized amounts $x^A, x^B$ from endowment vectors $w$, then the post-ADL equity difference is:
\[
e'^A - e'^B = (c + \mathrm{PNL} - x^A) - (c + \mathrm{PNL} - x^B) = x^B - x^A,
\]
which depends only on the endowment allocation, not on $c$.
This means any theorem whose proof uses only (i) budget balance in the endowment ($\sum_i x_i = H$), and (ii) ordering/convexity/majorization properties of the survivor vector $w - x$, is unchanged after substituting the endowment vector from $(e_T)_+$ to $(\mathrm{PNL}_T)_+$.
The cash component only affects feasibility/coverability: under PNL-only haircuts, $\mathsf{Solv}_T = W_T - D_T > 0$ does not imply $D_T \le U_T$.

\iparagraph{Numerical example.}
Consider five accounts with equities $e_T = (10, 5, 1, -3, -12)$, cash components $c_T = (2, 1, 0.5, 0, 0)$, and PNL components $\mathrm{PNL}_T = (8, 4, 0.5, -3, -12)$.
Under PNL-only haircuts, the endowment vector is $w_T = (8, 4, 0.5, 0, 0)$ and the deficit is $D_T = 15$ (from losers with negative equity).
Let the queue policy choose severity $\theta_{\pi_Q} = \tfrac{1}{2}$, so the haircut budget is $B = \theta_{\pi_Q} D_T = 7.5$.
With queue order $\sigma = (2, 1, 3, 5, 4)$ induced by the PNL‑leverage scores, the construction in~\eqref{eq:pl-ranking} yields $A = (3.5, 0, 0, 0, 0)$ and haircuts $h = (0.4375, 1, 0, 0, 0)$ (seizing $x = (3.5, 4, 0, 0, 0)$).
Post-ADL equities are $e'_T = (2+8-3.5, 1+4-4, 0.5+0.5-0, -3, -12) = (6.5, 1, 1, -3, -12)$.
Under the pro‑rata rule~\eqref{eq:pro-rata} with the same severity $\theta_{\pi_{PR}} = \tfrac{1}{2}$, we have $D_T = 15$ and $U_T = 12.5$, so the uniform haircut factor is $h_{\pi_{PR}} = \theta_{\pi_{PR}} D_T / U_T = 15/25 = 0.6$ and seized amounts are $x = (4.8, 2.4, 0.3, 0, 0)$, giving post-ADL equities $e'_T = (2+8-4.8, 1+4-2.4, 0.5+0.5-0.3, -3, -12) = (5.2, 2.6, 0.7, -3, -12)$.
In particular, the ranking of positive‑equity accounts is preserved under pro‑rata, since all winners' endowments are multiplied by the same constant factor $(1-0.6) = 0.4$, illustrating its rank‑preservation property.

\subsection{ADL Trilemma}
\label{sec:adl-trilemma}

The preceding subsections defined trader equity, exchange solvency, insurance funds, and autodeleveraging (ADL) policies.
In this subsection we introduce a high–level design principle that organizes the rest of the paper: an \emph{ADL trilemma}.
Informally, a perpetuals venue cannot simultaneously (i) keep insolvency and ADL breach events rare, (ii) protect solvent winners from large socialized losses, and (iii) extract maximal long–run exchange revenue from trading and funding activity.
Any ADL design must pick (at most) two of these three goals.

\paragraph{Exchange revenue.}
To define the trilemma, we have to first formally describe what constitutes exchange revenue.
Recall from~\S\ref{subsec:exchange-solvency} that, given parameters $\alpha, \beta, \eta \in [0,1]$, the insurance fund $\mathsf{IF}_t$
evolves by
\[
\mathsf{IF}_{t+1}
= \mathsf{IF}_t
+ \alpha \sum_{j \in T_t} \tau_t(\Delta q_j)
+ \eta\, p_t V_t
+ \beta \sum_{i=1}^n |\gamma_t|\, p_t |q_i|
- \min\{\mathsf{IF}_t, D_t\},
\]
where $T_t$ is the multiset of liquidations at $t$, $V_t$ is traded volume, $\gamma_t$ is the funding rate, and $D_t$ is the period bad debt.
For notational convenience, define the per–period gross fee flows
\begin{align*}
\mathsf{Fee}^{\text{liq}}_t &= \sum_{j \in T_t} \tau_t(\Delta q_j) && \mathsf{Fee}^{\text{trade}}_t = p_t V_t && \mathsf{Fee}^{\text{fund}}_t = \sum_{i=1}^n |\gamma_t|\, p_t |q_i|.
\end{align*}
Most exchanges fully collect liquidation and trading fees and potentially collect fees on funding rates.\footnote{In particular, we note that HLP-like systems where fees are charged on quoting offsetting positions can be viewed as a form of exchange revenue~\cite{HyperliquidHLPVaults}. Moreover, a number of perpetuals AMM DEXs collect protocol revenue from funding flows~\cite{GMXDocs,PerpV2Docs}.}

Let $\zeta \leq (1-\beta)$ be the expected fraction of funding rates that the exchange keeps as revenue.
Then the exchange’s \emph{gross revenue} at time $t$ is
\[
\mathsf{Revenue}^{\text{gross}}_t = \mathsf{Fee}^{\text{liq}}_t + \mathsf{Fee}^{\text{trade}}_t + \zeta\mathsf{Fee}^{\text{fund}}_t.
\]
By construction, the fractions $\alpha,\eta,\beta$ of these fee flows are diverted into the
insurance fund. The remaining share accrues as \emph{net exchange revenue}
\[
\mathsf{Revenue}_t
:= (1-\alpha)\,\mathsf{Fee}^{\text{liq}}_t
 + (1-\eta)\,\mathsf{Fee}^{\text{trade}}_t
 + (1-\beta-\zeta)\,\mathsf{Fee}^{\text{fund}}_t.
\]
We note that in~\S\ref{sec:multi-round} that we refine this notion of revenue to a discounted \emph{exchange long–term value} (LTV).
The LTV accounts for future expected fee flows and traders leaving the exchange as a function of the realized ADL haircuts.

\paragraph{Three competing desiderata.}
There are three main goals that an ADL policy aims to enforce for traders and exchanges: revenue, solvency, and fairness.
We formalize these notions in the sequel:
\begin{itemize}
  \item \emph{Solvency.} (\S\ref{subsec:severity}, \S\ref{sec:negative}) The exchange aims to ensure that $R_t = 0$ for most times $t$ and that the total shortfall $\sum_t R_t$ is small relative to the expected insurance fund size.
  This desideratum is beneficial to both the exchange and traders.

  \item \emph{Fairness and moral hazard.} (\S\ref{sec:fairness}, \S\ref{sec:glpr})
  Traders want to know that if they face socialization, they will not be asked to absorb a portion of $R_t$ that is too large relative to their notional exposure on the exchange.
  This desideratum is mainly beneficial to traders.

  \item \emph{Exchange revenue and participation.} (\S\ref{sec:multi-round})
  Heavy ADL on high–value winners can trigger traders exiting the exchange permanently, shrinking future fee flows.
  This desideratum is mainly beneficial to the exchange.
\end{itemize}

\noindent It is clear that $\alpha,\beta,\eta$ and the ADL policy $\pi$ jointly control these three dimensions.
But a natural question is what the trade-offs are between these three desiderata at different values of $\alpha, \beta, \eta$ and parametrizations of $\pi$.
Raising $\alpha, \beta, \eta$ builds the reserve faster and decreases breaches (\ie~times with $R_t > 0$), but diverts exchange revenue into the fund.
Increasing ADL severity (or concentrating haircuts) makes breaches rarer and accelerates solvency, but worsens fairness and drives away the highest–value winning traders.
Keeping severities (\ie~total solvency resolved by ADL) small preserves fairness and participation, but leaves the exchange exposed to repeated shortfalls.
The goal of this paper is to formalize these three statements via a trilemma:

\begin{proposition}[Trilemma, Informal]
\label{prop:adl-trilemma}
Fix a sequence of perpetuals exchanges $\mathcal{P}_n$ and static ADL policies $\pi_n$ with insurance parameters $(\alpha, \beta, \eta)$.
Under the heavy-tailed shortfall assumptions of \S\ref{sec:negative}, no policy family $(\pi_n)$ can simultaneously satisfy the following uniformly in $n$:
\begin{enumerate}
  \item[\textbf{(S)}] \textbf{Solvency:} Residual risk is controlled, \ie, $\sum_t R_t(\pi_n) = O_p(1)$ and $\mathbb{P}[R_t(\pi_n) > 0] = O(1)$.
  \item[\textbf{(F)}] \textbf{Fairness:} Moral hazard is bounded, \ie, $\upsilon^{\pi}_T/D^{\pi}_T = \Theta(1)$ and $\upsilon^{\pi}_T/\Delta^{\pi}_T = \Theta(1)$, where $\upsilon^{\pi}_T$ is the maximum post-ADL haircutable endowment.
  \item[\textbf{(R)}] \textbf{Revenue:} Exchange revenue is preserved relative to a benchmark, \ie, $\text{LTV}(\pi_n) \asymp \text{LTV}_{\text{benchmark}}$.
\end{enumerate}
Enforcing \textbf{(S)} requires sacrificing \textbf{(F)} (via concentrated haircuts) or \textbf{(R)} (via excessive insurance diversion). Conversely, preserving \textbf{(F)} and \textbf{(R)} necessitates frequent solvency breaches.
\end{proposition}

\noindent A formal statement with precise definitions and a complete proof of the proposition appears in Appendix~\ref{app:adl-trilemma-proof}.

\paragraph{Scope note.}
The ADL trilemma is a statement about an abstract policy model (static ADL families under Assumptions J.1--J.3).
It does not rely on any empirical ``two-party zero-sum fill'' identity or on observing a complete execution/settlement ledger.
Empirical sections map public observations onto model objects; that mapping is separate from the theorem.

\paragraph{Regime boundary.}
The trilemma is \emph{conditional} on Assumption J.3 (the structural deficit regime $\mu_- > \mu_\Phi$).
It is \emph{silent} when $\mu_- \le \mu_\Phi$: in low-leverage or light-tailed regimes where expected shortfall rates are smaller than the maximum sustainable fee diversion rate, insurance can cover deficits without ADL binding, and all three goals (solvency, fairness, revenue) may be simultaneously achievable.
Here, $\mu_\Phi$ can be interpreted as the largest fee diversion rate that remains compatible with non-declining long-run venue value (e.g., via an LTV sensitivity constraint).
This is an interpretation of the regime boundary, not a change to the formal assumption used in the proof.
The paper focuses on the high-leverage structural deficit regime ($\mu_- > \mu_\Phi$) because that is where major perp venues actually operate and where ADL events are empirically concentrated.

\paragraph{Proof sketch and roadmap.}
The remainder of the paper establishes the Trilemma by analyzing each desideratum in turn:

\begin{itemize}
  \item \emph{Solvency Ratios (\S\ref{sec:negative}, Appendix~\ref{app:proofs}).}
  We analyze the ratio $\upsilon^{\pi}_T/D^{\pi}_T$ under heavy-tailed assumptions, where $\upsilon^{\pi}_T$ is the maximum post-ADL endowment.
  Theorem~\ref{thm:master-ptsr} proves that for any budget-balanced static ADL with severity $\theta_n$, the ratio scales as $b_n/(\theta_n n)$.
  Consequently, enforcing rare breaches (large $\theta_n$) drives the ratio to zero, violating \textbf{(F)}.

  \item \emph{Fairness Models (\S\ref{sec:fairness}--\S\ref{sec:glpr}, Appendices~\ref{app:capped-pro-rata}--\ref{app:rap-optimality-and-convex-dominance}).}
  We formalize fairness in two regimes.
  First, for a single round (\S\ref{sec:fairness}), we show that capped pro-rata is the unique policy satisfying Sybil resistance and monotonicity.
  Second, under external shocks (\S\ref{sec:glpr}), we prove that leverage-weighted rules (RAP) minimize future shortfall but sacrifice the most systemically important winners, highlighting the tension between short-term robustness and long-term participation.

  \item \emph{Revenue and Price of Anarchy (\S\ref{sec:multi-round}, Appendix~\ref{app:stack-nash}).}
  We model the revenue-solvency tradeoff as a Stackelberg game between the exchange and traders.
  We show that static severities yield an unbounded Price of Anarchy ($J_n(\pi) \sim \theta_n n / b_n$) compared to dynamic policies.
  This confirms that achieving \textbf{(R)} and \textbf{(S)} simultaneously requires dynamic intervention, as static rules cannot balance the trilemma asymptotically.
\end{itemize}

\section{Risk and Fairness Preliminaries}\label{sec:risk-prelim}\label{sec:risk}
To analyze how different ADL policies balance solvency and winner survival, we require metrics that capture both the \emph{magnitude} of losses (risk) and their \emph{distribution} across traders (fairness).
Classical risk measures like Value-at-Risk (VaR) and Expected Shortfall (ES) quantify aggregate solvency but ignore how the burden is shared.
To address this, we introduce fairness-aware metrics and distributional orderings that allow us to rank policies based on how they concentrate or spread losses among winners.
We then provide with an brief presentation of extreme-value scaling that we use in the sequel, before concluding with a brief introduction to algorithmic fairness.
The preliminaries in this section are meant to be incomplete and we refer the reader to the literature cited for more details.

\subsection{Risk Metrics}
\subsubsection{Traditional Risk Metrics}
We briefly review standard metrics used in finance and regulation~\citep{Boyd2017MultiPeriodTrading,BCBS2019FRTB}.
For a loss $X\ge 0$ and confidence level $\alpha \in (0,1)$:
\begin{itemize}
    \item \emph{Value-at-Risk (VaR):} The quantile $\mathrm{VaR}_\alpha(X) = \inf\{x : \Prob(X>x) \le \alpha\}$.
    \item \emph{Expected Shortfall (ES):} The average loss in the worst $\alpha$-fraction of cases, $\mathrm{ES}_\alpha(X) = \Expect[X \mid X \ge \mathrm{VaR}_\alpha(X)]$.
\end{itemize}
In ADL, we apply these to the residual shortfall $R_t$ to measure solvency risk.

\subsubsection{ADL-Specific Efficiency Metrics}\label{sec:risk-metrics}
To capture the trade-off between solvency and trader welfare, we define two ratios that normalize winner endowment survival by the scale of the default:
\begin{itemize}
    \item \emph{Profitability-to-Total-Solvency Ratio (PTSR):}
    \[
    \mathsf{PTSR}_T(\pi) = \Expect\left[\frac{\upsilon^\pi_T}{D^\pi_T}\right],
    \]
    where $\upsilon^\pi_T = \max_i (w_{i,T} - x_{\pi,i})_+$ is the maximum post-ADL haircutable endowment.
    This measures the survival of the top winner's endowment (profit capacity) relative to the total socialized loss. A low PTSR indicates that the most profitable trader is being disproportionately cannibalized to cover deficits.

    \item \emph{Profitability-to-Maximum Solvency Ratio (PMR):}
    \[
    \mathsf{PMR}_T(\pi) = \Expect\left[\frac{\upsilon^\pi_T}{\Delta^\pi_T}\right].
    \]
    This compares the top winner's post-ADL endowment to the largest single loser's shortfall, capturing the concentration of risk on both sides of the trade.
\end{itemize}
These ratios mirror the VaR/ES distinction: PTSR captures aggregate efficiency, while PMR captures tail concentration.
Under PNL-only haircuts ($w = (\mathrm{PNL})_+$), these metrics measure profit capacity survival, not total equity survival.
We formalize this connection and derive exact relationships in Appendix~\ref{app:proofs}.

\subsection{Fairness and Distributional Comparisons}
While scalar metrics like PTSR provide summary statistics, they cannot fully capture the fairness of a policy across the entire population of winners.
Two policies might achieve similar solvency but distribute the pain very differently --- one by wiping out a few large winners (Queue), another by shaving everyone slightly (Pro-Rata).
To rank policies robustly, we use tools from majorization theory~\citep{MarshallOlkinArnold2011}.

\paragraph{Schur-Convexity and Submajorization.}\label{sec:schur-convex-ordering}
We compare haircut vectors $h \in \reals^n$ using \emph{submajorization} ($\prec_w$).
We say a policy $\pi_A$ is \emph{more fair} (or less concentrated) than $\pi_B$ if its haircut vector is submajorized by $\pi_B$'s, i.e., $h(\pi_A) \prec_w h(\pi_B)$.
A function $\phi:\reals^n \to \reals$ is said to be \emph{Schur-convex} if $x \prec_w y$ implies $\phi(x) \le \phi(y)$, \ie~it increases with concentration.
This implies that for \emph{any} convex, symmetric cost function $\phi$ (representing trader disutility), the aggregate distress is lower under $\pi_A$:
\[
h(\pi_A) \prec_w h(\pi_B) \implies \sum \phi(h_i(\pi_A)) \le \sum \phi(h_i(\pi_B)).
\]
This gives us a powerful, parameter-free way to claim that Pro-Rata is ``fairer'' than Queue: it minimizes the collective pain for all convex risk attitudes.
We provide a detailed treatment of these orderings and their application to ADL in Appendix~\ref{app:lpr-convex-dominance}.

\paragraph{Comonotonicity.}
In our worst-case analysis, we often consider \emph{comonotonic} couplings, where winners' profits and losers' deficits are perfectly correlated (move in lockstep).
This represents the most dangerous regime for ADL, as large deficits coincide with large winner profits, testing the policy's ability to extract liquidity without destroying the best traders.
Our negative results in \S\ref{sec:negative} exploit this structure to derive tight bounds on the Trilemma.
See Appendix~\ref{app:lpr-convex-dominance} for further details on comonotonic risk bounds.

\subsection{Extreme--value scaling}
Our asymptotic analysis relies on \emph{extreme--value scales}: deterministic sequences that characterize the typical magnitude of the largest winner endowment (for haircut capacity) and the largest loser shortfall in the limit of a large market.
Recall the winner and loser index sets $\mathcal{W}_T$ and $\mathcal{L}_T$ with cardinalities $k_n$ and $m_n$, respectively.
Under mild mixing we assume the aggregate winner endowment and loser deficit masses concentrate at linear scales:
\begin{equation}\label{eq:lln-scaling}
  \left(\frac{U_T(\mathcal{P}_n)}{n},\, \frac{D_T(\mathcal{P}_n)}{n}\right)
  \xrightarrow{p} (\mu_+,\mu_-),
\end{equation}
where $U_T = \sum_i w_{i,T}$ is the total haircut capacity.
Here $\mu_+,\mu_->0$ summarize the average winner endowment and loser deficit magnitudes.
We define the respective maxima as
\begin{equation}\label{eq:maxima-def}
  \big(\upsilon_T,\, \Delta_T\big)
  = \left(
    \max_{i\in\mathcal{W}_T} w_{i,T},\,
    \max_{i\in\mathcal{L}_T} \big(-e_T(\mathfrak p_{i,T})\big)_+
  \right),
\end{equation}
where $\upsilon_T$ is the maximum haircutable endowment (under PNL-only, this is the maximum positive PNL).
A pair of deterministic, increasing sequences $\{b^+_{k}\}_{k\ge1}$ and $\{b^-_{m}\}_{m\ge1}$ constitute \emph{extreme--value scales} if there exist constants $c_+,c_-\in(0,\infty)$ such that, as $n\to\infty$,
\begin{equation}\label{eq:ev-scaling}
  \left(\frac{\upsilon_T}{b^+_{k_n}},\, \frac{\Delta_T}{b^-_{m_n}}\right)
  \xrightarrow{p} (c_+,c_-).
\end{equation}
We abbreviate $b_n=b^+_{k_n}$ when the context (winner side) is unambiguous.
A canonical choice, consistent with classical extreme--value theory, is the set of \emph{upper--quantile normalizers}:
\begin{equation}\label{eq:quantile-scales}
  \big(b^+_{k},\, b^-_{m}\big)
  = \left(
    F_+^{-1}\!\left(1-\frac{1}{k}\right),\,
    F_-^{-1}\!\left(1-\frac{1}{m}\right)
  \right),
\end{equation}
where $F_+$ and $F_-$ denote the distribution functions of winner-side positive endowments (under PNL-only, positive PNL) and loser-side shortfalls, respectively.

\iparagraph{Extreme–value severity scales.}
We refer to $\theta_n = \Theta(b_n/n)$ as the \emph{extreme--value severity scale}.
This choice aligns the haircut magnitude with the typical largest winner endowment $b_n=b^+_{k_n}$ (under PNL-only, the largest positive PNL) in a large market where $k_n=\Theta(n)$.

\paragraph{Examples.}
Two prototypical cases appear throughout our analysis:
\begin{itemize}
  \item \emph{Light tails (sub--Gaussian).} For distributions with sub--Gaussian scales $\sigma_+,\sigma_-$,
  \[
  b^+_{k}\ \asymp\ \sigma_+\sqrt{2\log k},
  \qquad
  b^-_{m}\ \asymp\ \sigma_-\sqrt{2\log m}.
  \]
  \item \emph{Power--law tails (Pareto/Fréchet).} For distributions satisfying $\Prob\{w_T>x\}\sim C_+x^{-\alpha_+}$ (under PNL-only, this is the distribution of positive PNL) and $\Prob\{(-e_T)_+>x\}\sim C_-x^{-\alpha_-}$ with $\alpha_\pm>0$,
  \[
  b^+_{k}\ \asymp\ (C_+ k)^{1/\alpha_+},
  \qquad
  b^-_{m}\ \asymp\ (C_- m)^{1/\alpha_-}.
  \]
\end{itemize}

\subsection{Fairness}
\paragraph{Algorithmic Fairness in ADL.}
Algorithmic fairness literature distinguishes between \emph{axiomatic} approaches, which posit structural invariants (often yielding impossibility results~\citep{Arrow1951,KleinbergEtAl2018FairnessImpossibility}), and \emph{optimization} approaches, which minimize disparate impact via convex programs~\citep{DworkEtAl2012FairnessAwareness}.
ADL spans both: queue-based policies resemble discrete ranking mechanisms subject to Arrow-style impossibilities, while pro-rata policies admit continuous convex formulations.
We adopt an axiomatic perspective to characterize ``ideal'' fairness, then relax it to optimize for robustness.

\paragraph{Axiomatic Properties.}
We formalize three operational desiderata for fair ADL policies:
\begin{itemize}
    \item \emph{Sybil Resistance (Split Invariance):} Let $x_{\pi,i}=h_{\pi,i}w_{i,T}$ denote seized endowment. For any winner $i$, any split $w_{i,T}=\sum_{a=1}^r z_a$ with the same economic state, and any budget $B$, if $\tilde x_{\pi,a}(B)$ denotes seizure on split children, aggregate seizure must be unchanged:
    \[
      \sum_{a=1}^r \tilde x_{\pi,a}(B) = x_{\pi,i}(B).
    \]
    \item \emph{Monotonicity (Stable Ordering):} The policy preserves winner ordering in surviving endowment:
    \[
      w_{i,T} \ge w_{j,T} \implies (w_{i,T}-x_{\pi,i}) \ge (w_{j,T}-x_{\pi,j}).
    \]
    \item \emph{Scale Invariance:} If all deficits and equities scale by $\lambda > 0$, the relative haircuts should remain unchanged. This ensures the mechanism responds to \emph{leverage} and \emph{risk distribution}, not nominal sizes.
\end{itemize}
In \S\ref{sec:fairness}, we prove that \emph{capped pro-rata} is the unique policy satisfying all three properties.
Appendix~\ref{app:capped-pro-rata} makes the queue comparison explicit: Proposition~\ref{prop:queue-sybil-score-preserving} shows that queue Sybil resistance depends on the ranking score (false for absolute scores, where splitting can change queue rank, and true for ratio scores that keep rank unchanged under proportional splits, such as Hyperliquid's documented index~\cite{HyperliquidDocsADL}), while Proposition~\ref{prop:queue-monotonicity-failure} shows greedy queues can still violate monotonicity.

\subsection{Moral hazard}
Moral hazard is a generic incentive problem that arises when one party (the \emph{agent}) can take actions that affect risk or payoffs, but those actions are not perfectly observable, verifiable, or contractible by the party bearing some of the downside (the \emph{principal}); see, \eg~\citet{Holmstrom1982,LaffontMartimort2002}.
Classic examples include workers choosing effort under wage contracts, managers taking portfolio risk on behalf of investors, or insurers providing coverage to policyholders whose behavior may change once insured.
Because the agent’s action is hidden, the principal cannot simultaneously provide full insurance (so that the agent’s payoff is insulated from shocks) and preserve strong incentives for the agent to take socially desirable actions: risk sharing necessarily distorts effort incentives.
The central goal of the moral-hazard literature is therefore to characterize \emph{second-best} contracts (typically linear or otherwise simple sharing rules) that optimally trade off incentive provision against insurance under information constraints~\citep{DuttingEtAl2023MultiAgentContracts,Carroll2015RobustLinearContracts}.

Moral hazard also couples to robustness: contracts must work across environments and counterparties, pushing outcomes away from the first-best benchmark~\citep{DuttingEtAl2023MultiAgentContracts}.
The same friction arises whenever losses are socialized (\eg~insurance pools or centrally cleared derivatives) as risk sharing dulls the ability for individuals to manage risk efficiently.
ADL fits this template: winners and losers supply mutual insurance while the exchange chooses the loss-sharing rule.
The negative results in Section~\ref{sec:negative} formalize the resulting second-best trade-off.

\section{Severity Optimization}\label{subsec:severity}
We decompose the ADL problem into two orthogonal components: \emph{severity} (how much to socialize) and \emph{haircuts} (who pays).
Solvency is driven almost entirely by the scalar severity sequence $(\theta_t)$, while fairness and revenue depend on the haircut distribution $h_t$.
Since \S\ref{sec:fairness}--\S\ref{sec:glpr} cover haircut design, here we focus on optimizing $\theta_t$ to balance solvency against the risk of driving away traders.

\paragraph{Severity Policies.}
A \emph{causal severity policy} maps available information $\mathcal{F}_t$ (deficits, funding rates, risk limits) to a severity fraction $\theta_t \in [0, 1]$.
We contrast three approaches:
\begin{enumerate}
    \item \emph{Static Interpolation:} $\theta_t$ is fixed or linearly moves between deficit-matching $\theta_t^{\text{match}}=\min\{1,\ \sum_i e_{t,i}/D_t\}$ (so $H_t = B_t^\star$ when feasible) and a cap $\bar{\theta}$ (\eg~1).
    \item \emph{Exponential Backoff:} $\theta_t = \theta_0 \alpha^{k_t}$ decays as the number of recent shocks $k_t$ increases. This prevents cascading failures from wiping out the entire book during prolonged stress.
    \item \emph{Online Control (Mirror Descent):} We treat severity selection as an online convex optimization problem. The controller adjusts $\theta_t$ to minimize a regret bound combining solvency costs ($R_t$) and revenue loss (excessive haircuts). This approach adapts to non-stationary market conditions and respects caps $\Theta_t$ derived from the Price of Anarchy analysis in \S\ref{sec:multi-round}.
\end{enumerate}

\paragraph{Separation Principle.}
Under the budget-balance constraint~\eqref{eq:budget-balance}, any one-round ADL policy $\pi$ with severity $\theta_\pi$ and haircuts $h_\pi$ satisfies
\[
  H(\pi) \;=\; \sum_i h_{\pi,i}\,w_{i,T} \;=\; \sum_i x_{\pi,i} \;=\; \theta_\pi D_T(\mathcal{P}_n),
\]
so the post-ADL deficit is
\[
  R_T(\pi) \;=\; (D_T(\mathcal{P}_n) - H(\pi))_+ \;=\; (1-\theta_\pi)_+ D_T(\mathcal{P}_n),
\]
which depends only on $\theta_\pi$ and $D_T(\mathcal{P}_n)$, not on how $H(\pi)$ is distributed across winners.
Thus, any solvency functional that depends on $\pi$ only via $\{R_t(\pi)\}_t$ (\eg~breach frequencies, VaR/ES of $\sum_t R_t$) can be optimized over the scalar sequence $(\theta_t)$, holding the haircut rule fixed.
Conversely, fairness and revenue functionals depend on the winner-side post-ADL endowments $\{(1-h_{\pi,i})w_{i,T}\}_i$ (and consequently post-ADL equities $e'_{i,T} = c_{i,T} + \mathrm{PNL}_{i,T} - x_{\pi,i}$) and are invariant to the choice of $\theta_t$ once the total budget $H_t$ is fixed.
We exploit this separability in our numerics (\S\ref{sec:numerics}) by tuning $\theta_t$ (severity control) and $h_t$ independently.

\section{Negative Results}\label{sec:negative}
We begin with two structural limits for ADL.
They formalize a moral‑hazard trade‑off between solvency and winner payoffs that cannot be removed by better policy design.
The tension is instantaneous and zero‑sum at the shock: covering the aggregate deficit $D_T$ eventually requires shaving the profits of the best winner.
We first argue that this tension can be viewed as a form of moral hazard, in the sense used in principal-agent contracts.

An exchange that earns fees from trading volume has an incentive to let very risky, highly leveraged users accumulate large positions.
When these users default, their negative equity shows up as a large deficit $D_T$, which is then socialized onto the winners via ADL haircuts.
In our model this is the basic moral‑hazard question: the exchange enjoys fee revenue from risky losers, while the cost of their tail losses is borne by solvent winners.
The results below quantify how this wedge scales as the market thickens, showing that the severity needed for solvency grows with the number of positions and with the leverage imbalance between losers and winners.

\subsection{Impossibility of Avoiding Moral Hazard as \texorpdfstring{$n\to\infty$}{n Grows}}\label{sec:moral-hazard}
We first show that the tension between solvency and winner survival scales unfavorably with market size, creating an unavoidable moral hazard wedge.
Viewing traders as principals and the exchange as an agent, preserving the agent's utility (solvency) asymptotically destroys the principals' utility (winner profits) unless severity vanishes.

\paragraph{Example.}
Consider a book of size $n$ where loser positions are i.i.d. with mean loss $\mu > 0$, so aggregate deficit scales as $D_T \approx \mu n$.
Suppose the top winner's endowment (under PNL-only, positive PNL) scales sub-linearly, \eg~$b_n \sim c \log n$ (light tails).
If the exchange commits to a fixed severity $\bar{\theta} > 0$ to maintain solvency, the total haircut budget is $H_n \approx \bar{\theta} \mu n$.
As $n \to \infty$, this linear haircut cost overwhelms the logarithmic winner endowment:
\[
w_n^{\text{post}} \approx (c \log n - \bar{\theta} \mu n)_+ \to 0.
\]
Thus, to preserve any winner profit capacity, severity must vanish at the rate $\theta_n = O(b_n/n)$.
Constant severity policies therefore represent a form of growing moral hazard: they transfer an increasingly large share of tail risk from the exchange to the most profitable traders.
We formalize this scaling in Appendix~\ref{app:mh-example}.

\paragraph{Main Result.}
In our model, socialization is instantaneously zero‑sum: every dollar of deficit $D_T$ covered by ADL must come from the winners' haircutable endowment.
If ADL usage does not vanish as the market grows larger ($n\rightarrow\infty$), then one cannot simultaneously keep deficits small (solvency, low $D_T$) and preserve winner profit capacity (high $\upsilon_T$). 
The proposition below quantifies this tradeoff in terms of extreme-value scales and can be read as a moral‑hazard impossibility result in the principal–agent sense.

\begin{proposition}[Informal]\label{prop:ev-scale-severity-informal}\label{prop:ev-impossibility}
Suppose that the extreme value scale of the winning trader's endowment, $\upsilon_T(\mathcal{P}_n)$, is $b_n$ (under PNL-only, this is the maximum positive PNL).
Then, 
\[
\mathsf{PTSR}_T(\mathcal{P}_n, \pi) =  \Expect_{\pi}\left[\frac{\upsilon^{\pi}_T(\mathcal{P}_n)}{D^{\pi}_T(\mathcal{P}_n)}\right] \asymp \frac{b_n}{\theta_n\, n}
\]
In particular, the example above shows the order $b_n/n$ is tight: unless the severity vanishes at the extreme–value scale, $\theta_n=\Theta(b_n/n)$, one cannot preserve the best trader's profit capacity as $n\to\infty$.
\end{proposition}

\iparagraph{Example.}
The bound $b_n/n$ dictates the maximal safe severity for different distributions:
\begin{itemize}
    \item \emph{Gaussian Winners (Light Tails):} For $w_i \sim \mathcal{N}(0,1)$ (endowment distribution), the maximum scales as $b_n \sim \sqrt{2\log n}$. 
    Thus, severity must vanish rapidly, $\theta_n \lesssim \sqrt{2\log n}/n$, to avoid autodeleveraging the top winner's endowment.
    This implies that the shortfall that can be covered goes to zero as $n\rightarrow \infty$.
    \item \emph{Pareto Winners (Heavy Tails):} For $w_i \sim \text{Pareto}(\alpha)$ with $\alpha > 1$, the maximum scales as $b_n \sim n^{1/\alpha}$. 
    Here, the constraint is looser but still vanishing as $n\rightarrow \infty$: $\theta_n \lesssim n^{-(1-1/\alpha)}$.
\end{itemize}
Even in the heavy-tailed case (infinite variance, $\alpha < 2$), budget balance implies that a constant fraction severity $\theta$ wipes out the top winner's endowment almost surely as $n \to \infty$.

\subsection{Excessive leverage guarantees large maximal trader loss}\label{subsec:lev-imbalance}
Another natural form of moral hazard, one that quantifies the role of the exchange as an agent, is the ratio of leverage held by the winning and losing sides.
If losers generally have much more leverage than winners, then the exchange is giving winning traders a higher likelihood of paying socialized losses.
We quantify this deficiency in risk management by the exchange by providing a more quantitative extreme value result: we show that the ratio of winner to loser leverages is a multiplicative term in how fast the winning trader's position is likely to be autodeleveraged.

\paragraph{Example.}
Consider a regime where loser leverage mass $\ell^-$ dominates winner leverage mass $\ell^+$.
Let $b_n^+$ (resp.\ $b_n^-$) denote the extreme value (EV) scales of winner endowments and loser deficits, respectively.
Aggregate shortfall scales with loser liability, $D_T \approx \ell^- b_n^-$, while the top winner's endowment scales as $\upsilon_T \approx \ell^+ b_n^+$.
If severity $\theta_n$ is fixed, the haircut budget $H_n = \theta_n D_T \approx \theta_n \ell^- b_n^-$ is large.
Comparing budget to winner endowment:
\[
\frac{H_n}{\upsilon_T} \approx \theta_n \frac{\ell^-}{\ell^+} \frac{b_n^-}{b_n^+}.
\]
If $\ell^- \gg \ell^+$, even a small severity $\theta_n$ creates a haircut larger than the top winner's entire endowment.
Thus, safe severity is throttled not just by market size but by the \emph{leverage imbalance} $\ell^+/\ell^-$.
See Appendix~\ref{app:mh-leverage} for the rigorous construction.

\paragraph{Main Result.}
The next result makes the leverage–imbalance effect precise.
When loser positions carry much more leverage than winners, the ratio of winner to loser leverage masses \(\ell_n^+/\ell_n^-\) enters multiplicatively in the EV‑scale bound for the top winner’s survival.
In moral‑hazard terms, an exchange that tolerates very high loser leverage effectively exposes its best winners to autodeleveraging: even with small severities \(\theta_n\), a large imbalance \(\ell_n^-/\ell_n^+\) can still force substantial haircuts on the most profitable traders.

\begin{proposition}[Informal]\label{prop:excessive-leverage}
Let $b_n^+, b_n^-$ be the extreme value scales of the winning traders' endowments and losing traders' deficits, respectively.
Furthermore, let the winner and loser side leverage masses be $\ell_n^+,\ell_n^-$, respectively.
Then we have the general EVT scaling:
\[
\Expect_{\pi}\!\left[\frac{\upsilon^{\pi}_T(\mathcal{P}_n)}{D^{\pi}_T(\mathcal{P}_n)}\right]\ \asymp\ \frac{\ell_n^+}{\ell_n^-}\cdot\frac{b_n^+}{b_n^-}\cdot\frac{1}{\theta_n},
\]
where $\upsilon^{\pi}_T$ is the maximum post-ADL endowment.
\end{proposition}
\noindent The full proof is in Appendix~\ref{app:proofs}.

\paragraph{Example.}
The scaling is now modulated by the \emph{leverage imbalance} $\Lambda_n = \ell_n^+/\ell_n^-$.
\begin{itemize}
    \item \emph{Symmetric Leverage ($\Lambda_n \asymp 1$):} If winner and loser leverages are comparable, the bound reverts to the previous case ($b_n^+/b_n^-$).
    \item \emph{High Loser Leverage ($\Lambda_n \ll 1$):} If losers are $10\times$ more leveraged, the safe severity drops by $10\times$. In the extreme, if $\ell_n^- \sim n$ (\eg~a few massive whales) while $\ell_n^+ \sim 1$ (retail winners), the required severity $\theta_n$ becomes vanishingly small even faster than the EV scale suggests.
    \item \emph{High Winner Leverage ($\Lambda_n \gg 1$):} Conversely, if winners are highly leveraged (so their equity is sensitive to small price moves, but large in notional), the ratio $\Lambda_n$ buffers the impact, allowing for larger severities without autodeleveraging the top winner.
\end{itemize}

\paragraph{Moral Hazard Interpretation.}
These results formalize the ``second-best'' nature of ADL: perfect insurance without incentive distortion is impossible.
Tail shocks necessarily haircut the top winners' profit capacity (Proposition~\ref{prop:ev-scale-severity-informal}), and high loser leverage magnifies that hit (Proposition~\ref{prop:excessive-leverage}).
This is the same trade-off seen in clearinghouses: Variation Gains Haircutting trims winner gains to keep the CCP solvent \citep{CPMI_IOSCO_2014,Gregory2015}.
ADL is the perpetual futures analogue of Variation Gain Haircutting (\ie~an ex-post haircut on unrealized profits), so the same moral-hazard tension applies \citep{DuffieZhu2011,Pirrong2011}.
Under PNL-only haircuts, PTSR/PMR measure profit capacity survival (not total equity survival), playing the same role as VGH fairness/solvency gauges in CCP recovery \citep{CPMI_IOSCO_2014,Gregory2015,DuffieZhu2011}.

\subsection{Queue-based Methods are the worst ADL policies for the top winning trade}\label{subsec:queue-worst}
One implication of the theory is that the Queue mechanism is by far the worst ADL policy for the top winning trade.
We briefly formalize this with an informal proposition that we prove in Appendix~\ref{app:queue-vs-pr-top}.
\begin{proposition}[Queue Maximizes Damage to Top Winner]\label{prop:queue-worst-top-informal}
Among all valid ADL policies with fixed severity $\theta$, the \emph{Queue} policy (haircut largest winners first) uniquely minimizes the survival of the top winner's endowment.
Let $\upsilon_T^{\text{Queue}}$ and $\upsilon_T^{\text{PR}}$ be the top winner's post-ADL endowment under Queue and Pro-Rata, respectively.
Whenever the haircut budget $H = \theta D_T$ satisfies $H \le w_{(1)}$ (where $w_{(1)}$ is the largest endowment), we have the strict gap
\[
\upsilon_T^{\text{PR}} - \upsilon_T^{\text{Queue}} \;=\; H \left(1 - \frac{w_{(1)}}{U_T}\right) \;>\; 0.
\]
Consequently, Pro-Rata strictly dominates Queue in terms of fairness metrics:
\[
\mathsf{PTSR}_T(\text{PR}) - \mathsf{PTSR}_T(\text{Queue}) \;=\; \theta\left(1 - \frac{w_{(1)}}{U_T}\right).
\]
Thus, Queue maximizes the concentration of haircuts on the most profitable trader, leading to the worst possible PTSR/PMR scores.
\end{proposition}
\noindent The full proof is in Appendix~\ref{app:queue-vs-pr-top}.
Appendix~\ref{app:classical-risk-measures} further strengthens this statement by proving that Queue also minimizes the VaR and ES of the top survivor's endowment at every confidence level, so even classical risk metrics rank Queue as the worst policy for the top winner.

\section{Fairness}\label{sec:fairness}\label{subsec:fairness}
A natural question is whether a particular ADL policy is optimal or fair.
We demonstrate that a modified pro-rata rule respecting per-account constraints, the \emph{capped pro-rata rule}, is the most fair policy in two senses: it minimizes total convex disutility and is the unique rule satisfying key axiomatic properties.

\paragraph{Definition of the Capped Pro-Rata Rule.}
The capped pro-rata rule is a water-filling algorithm that enforces per-account constraints before applying a pro-rata haircut.
Combining haircut constraints~\eqref{eq:haircut-constraint} and equity constraints~\eqref{eq:equity-constraint}, we define an effective cap $\beta_i = \min\{\overline{h}_i, 1 - (\underline{e}_i - c_{i,T})/w_{i,T}\}$ for accounts with $w_{i,T} > 0$, where the equity constraint is reinterpreted in terms of the minimum post-ADL endowment.
This $\beta_i$ is the most we can haircut account $i$'s endowment while honoring both its per-account haircut limit and its minimum post-ADL equity.
The maximum haircut capacity is 
\begin{equation}\label{eq:max-cap}
C(\beta) = \sum_{i=1}^n w_{i,T} \beta_i.
\end{equation}
If the deficit exceeds this capacity, no feasible haircut exists.
Otherwise, we aim to equalize haircuts subject to these caps.
The capped pro-rata rule $\pi_{CP}$ is defined as:
\begin{equation}\label{eq:capped-pro-rata}
h_{\pi_{CP}, i} = \min\{\eta, \beta_i\}
\end{equation}
where $\eta \in [0,1]$ is chosen to satisfy the budget constraint
\begin{equation}\label{eq:eta}
\sum_{i=1}^n h_{\pi_{CP}, i} w_{i,T} = \theta_{\pi} D_T(\mathcal{P}_n).
\end{equation}
This can be computed efficiently by sorting the caps $\beta_i$ and finding the threshold $\eta$ (see eq.~\eqref{eq:eta} and the algorithm in Appendix~\ref{app:pro-rata-algorithms}).

\paragraph{Convex Optimality.}
Let $\phi : [0,1] \rightarrow \reals$ be an increasing convex function representing the disutility of a haircut.
We formulate the problem of minimizing total disutility subject to per-account constraints:
\begin{equation}\label{eq:haircut-optimization}
\begin{array}{ll}
    \displaystyle{\min_{h \in [0, 1]^n}} & \displaystyle{\sum_{i=1}^n \phi(h_i)} \\[6pt]
    \text{subject to} & \forall i \in [n]\;\; h_i \leq \beta_i \quad \text{(effective cap)}
\end{array}
\end{equation}
In Appendix~\ref{app:convex-optimality}, we show:
\begin{proposition}
    The unique solution to~\eqref{eq:haircut-optimization} is the capped pro-rata rule~\eqref{eq:capped-pro-rata}.
\end{proposition}
This implies that for any reasonable model of financial disutility, capped pro-rata maximizes social welfare.

\paragraph{Axiomatic Fairness.}
Fix round $T$ and policy $\pi$. Let $x_{\pi,i}=h_{\pi,i}w_{i,T}$ denote seized endowment and $s_{\pi,i}=w_{i,T}-x_{\pi,i}$ denote surviving endowment.
We define three natural fairness properties:
\begin{itemize}
    \item \emph{Sybil resistance:} For any split $w_{i,T}=\sum_a z_a$, with $\tilde x_{\pi,a}$ denoting seizure on split children, aggregate seizure is invariant: $\sum_a \tilde x_{\pi,a}=x_{\pi,i}$.
    \item \emph{Scale invariance:} Haircuts are unchanged if both shortfall and equities are scaled by a common factor.
    \item \emph{Monotonicity (Stability):} The policy preserves endowment ordering: if $w_{i,T} \ge w_{j,T}$, then $s_{\pi,i} \ge s_{\pi,j}$ (and consequently, if $e_i \ge e_j$ pre-ADL and cash components are comparable, then $e'_i \ge e'_j$ post-ADL).
\end{itemize}
In Appendix~\ref{app:capped-pro-rata}, we formalize the following result:
\begin{proposition}[Informal]
    Under mild smoothness conditions, the unique sybil-resistant, scale-invariant, and monotone ADL policy satisfying per-account constraints is capped pro-rata.
\end{proposition}
\noindent This establishes capped pro-rata as the unique fair ordering rule.
For queue-based rules, Sybil resistance depends on the ranking score (Appendix~\ref{app:capped-pro-rata}, Proposition~\ref{prop:queue-sybil-score-preserving}): absolute-value scores (\eg~total PNL) are Sybil vulnerable, while ratio scores are Sybil resistant when proportional splitting does not create score ties. Hyperliquid's documented ADL score
\[
  \frac{\mathrm{mark}}{\mathrm{entry}}\cdot\frac{\mathrm{notional}}{\mathrm{account\ value}}
\]
is degree-zero in $(\mathrm{notional},\mathrm{account\ value})$ and therefore split-invariant under proportional scaling~\cite{HyperliquidDocsADL}.\footnote{Worked example (also discussed in~\cite{Doug2025SybilTweet}): mark/entry $=1.2$, notional/account value $=12{,}000{,}000/2{,}500{,}000=4.8$, so score $5.76$. Splitting into $100$ accounts gives $120{,}000/25{,}000=4.8$ per account, so each score remains $5.76$.}
However, Sybil resistance does not imply monotonicity: Appendix~\ref{app:capped-pro-rata}, Proposition~\ref{prop:queue-monotonicity-failure} gives a two-account counterexample where queue budget exhaustion reverses post-ADL ordering.
RAP introduces a risk tilt, deviating from strict equity stability to reduce risk, but remains ``as stable as possible'' within its risk-weighted framework.
In the numerical examples we present in the next section, we will see explicit examples where capped pro-rata is better for both exchange solvency and trader profitability than the Queue and Pro-Rata.

\section{Risk-aware Policies (RAP)}\label{sec:glpr}
Fairness-focused rules like capped pro-rata treat all winning dollars equally, ignoring the heterogeneous risks imposed by different winners.
Higher effective leverage implies thinner liquidation windows and greater sensitivity to execution costs~\cite{AlmgrenChriss2001,Gatheral2010}.
We introduce \emph{Risk-Aware Pro-Rata} (RAP) rules, a family of rules that preserves the fairness structure of capped pro-rata while tilting haircuts toward higher-risk positions.
This family includes Levered Pro-Rata (LPR) when the risk weighting is linear in leverage.
We derive RAP from a robustness criterion that minimizes a one-step excess-deficit proxy.

\paragraph{One-step Next Deficit.}\label{subsec:next-deficit}
To formalize robustness, we consider the \emph{next-step deficit} $D^{\text{next}}_{T+1}$ resulting from a price shock $Z_T$ occurring immediately after the ADL procedure.
Let $\pi$ be an ADL policy yielding post-haircut notional $\tilde{n}_{T,i}$ and equity $e_{T,i}$.
We model the shock $Z_T$ using a Markov kernel sensitive to the state.
The next-step deficit is the negative equity after the shock:
\[
D^{\text{next}}_{T+1} = \sum_{i \in \mathcal{W}_T} (-e_{T,i} - \tilde{n}_{T,i} Z_{T,i})_+ = \sum_{i \in \mathcal{W}_T} \tilde{n}_{T,i} \left( \frac{1}{\lambda_{T,i}} + Z_{T,i} \right)_-.
\]
Our objective is to minimize the conditional expected deficit $\delta_T = \Expect[D^{\text{next}}_{T+1} \mid \mathcal{F}_T]$.
Defining $\psi_i(u) = \Expect[(u + Z_{T,i})_-]$, we have:
\begin{equation}\label{eq:delta-T}
\delta_T = \sum_{i \in \mathcal{W}_T} \tilde{n}_{T,i} \psi_i\left(\frac{1}{\lambda_{T,i}}\right).
\end{equation}
Appendix~\ref{app:rap-examples} illustrates how pro-rata fails to control this deficit when shocks scale with leverage.

\paragraph{Risk Models.}
We assume the exchange employs a convex risk model $g(\lambda)$ characterizing the risk contribution of leverage $\lambda$.
Common choices include:
\begin{enumerate}
  \item \emph{Linear}: $g(\lambda)=\lambda$, corresponding to LPR. Matches VaR/CVaR scaling under linear shocks.
  \item \emph{Power}: $g(\lambda)=\lambda^c$ ($c>1$), penalizing high leverage super-linearly.
  \item \emph{CVaR}: $g(\lambda)=(\lambda-\tau)_+$, a surrogate for Conditional Value-at-Risk.
\end{enumerate}
Appendix~\ref{app:rap-examples} provides numerical examples of these weightings.

\subsection{Risk-Aware Pro-Rata (RAP) Rule}
The RAP rule allocates haircuts using risk weights $\tilde{w}_i = \lambda_{T,i} g(\lambda_{T,i})$ (not to be confused with the endowment $w_{i,T}$).
With per-account caps $\beta_i$, the haircuts on the endowment are:
\begin{equation}\label{eq:glpr}
    h_{\mathrm{RAP}(g), i} = \min\{\beta_i, \tau \tilde{w}_i\},
\end{equation}
where $\tau$ is calibrated to meet the target budget $\sum_i h_{\mathrm{RAP}(g), i} w_{i,T} = \theta D_T$.
This concentrates haircuts on high-risk positions (those with high leverage and high risk scores).
See Appendix~\ref{app:pro-rata-algorithms} for precise implementation of the algorithm.

\subsection{Convex Dominance of RAP}\label{sec:rap-dominance}
We can construct an optimal risk model $g^*$ that minimizes $\delta_T$.
Intuitively, an optimal risk model should maximize the marginal reduction of $\delta_T$ per unit of endowment $w_{T,i}$ seized.
Differentiating $\delta_T$ with respect to the seized amount $x_{T,i} = h_{T,i} w_{T,i}$ reveals:
\begin{equation}\label{eq:delta-t-grad}
-\frac{\partial\,\delta_T(x)}{\partial x_{T,i}}
\ =\ n_{T,i}\,\psi_i\!\left(\tfrac{1}{\lambda_{T,i}}\right)
\ =\ e_{T,i}\,\lambda_{T,i}\,\psi_i\!\left(\tfrac{1}{\lambda_{T,i}}\right)
\ =\ e_{T,i}\,\rho_i(\lambda_{T,i}),
\end{equation}
where we define the \emph{perspective transform} $\rho_i(\lambda) = \lambda\,\psi_i(1/\lambda)$.
The greedy rule that maximizes the marginal reduction per endowment-dollar seized is exactly the rule that prioritizes winners with the largest $\rho_i(\lambda_{T,i})$.
Note that while the gradient is expressed in terms of equity $e_{T,i}$ (since insolvency is an equity concept), the control variable is the endowment $w_{T,i}$.

\paragraph{Perspective Transform and Economic Intuition.}
The function $\rho_i$ is the perspective transform of the convex function $\psi_i$~\cite[\S2.3.3]{BoydVandenberghe2004}, widely used to induce 1-homogeneity.
Economically, $u=1/\lambda$ is the bankruptcy buffer, which is the percentage price shock required to wipe out the position's equity (since loss equals equity when shock magnitude $Z = e/\tilde{n} = 1/\lambda$).
Using convex duality, we can write $\rho_i(\lambda) = \sup_{y\ge 0} \{y - \lambda\,\psi_i^{*}(y)\}$, where $y$ is a `yield' (equity change per unit buffer change) and $\psi_i^*(y)$ is the cost to insure a position with yield $y$.
Thus, maximizing $\rho_i$ is equivalent to maximizing the net solvency gain per unit of endowment seized (the gradient in~\eqref{eq:delta-t-grad} is expressed in equity terms because insolvency is an equity concept, but the control variable is the endowment $w_{T,i}$).
For example, if a trader with leverage $\lambda=10$ faces an insurance cost of 0.5\% for a yield of 8\%, the net solvency gain is $0.08 - 10(0.005) = 0.03$ per unit equity, which translates to a gain per unit of endowment seized under the PNL-only model.

\paragraph{Constructing the optimal risk model.}
This suggests the optimal weights $w^\star_i = \rho_i(\lambda_{T,i})$ and risk model $g^\star(\lambda) = \rho(\lambda)/\lambda$.

\begin{proposition}[Informal]
For any weighted pro-rata rule comonotone with $w^\star$, the residual risk vector is weakly submajorized by that of $\mathrm{RAP}(g^\star)$.
This implies that $\mathrm{RAP}(g^\star)$ minimizes the total risk for any convex risk measure, assuming the shock process follows the kernel $K$.
\end{proposition}
\noindent The proof is in Appendix~\ref{sec:rap-optimality-and-convex-dominance}.
This result implies that if we have prior knowledge about the kernel $K$ (\eg~from historical data or backtesting), then constructing $g^*$ via the perspective transform yields the optimally robust and fair haircut model.

\subsection{Limitations of One-Step Optimality}
While RAP minimizes the one-step next deficit, it may not be optimal under multiple correlated shocks or broader exchange objectives.

\paragraph{Correlated Shocks.}
Consider sequential shocks $Z_{t+1}, Z_{t+2}$ with positive correlation.
High-leverage positions might act as a hedge against future shocks.
In such cases, aggressively liquidating high-leverage winners (as in RAP) can increase the two-step deficit compared to pro-rata.
Appendix~\ref{app:rap-examples} details an AR(1) shock example where the sum of deficits is lower for pro-rata.

\paragraph{Exchange Incentive Compatibility.}
Exchanges must also consider \emph{exchange long–term value (LTV)} --- the future fees and liquidity provided by traders.
RAP targets high-leverage, often high-volume, traders.
Fully liquidating these accounts to minimize immediate risk can reduce the exchange's long-term utility.
Appendix~\ref{app:exchange-incentives} provides an example where preserving a high-value trader (via pro-rata) yields higher total utility than the risk-minimizing RAP.

\section{Multiple–Round ADL as a Stackelberg Control Problem}\label{sec:multi-round}\label{sec:stackelberg-rewrite}\label{sec:stackelberg-adl}
As the previous section illustrates, risk-aware pro-rata mechanisms work well as robust, static strategies but not in a dynamic setting.
Both future exchange revenue and correlated shocks require including a model of the dynamics of ADL over multiple rounds.
We use the intuition gained from RAP mechanisms for a single round to help extend the model to the multiple round setting.
A common place where there is a difference between one-shot equilibria and multi-round equilibria is in the difference between Nash and Stackelberg games.
We will illustrate how the multi-round ADL problem is naturally a Stackelberg game whereas the ADL policies that we have studied so far are more like Nash optimal strategies.
To do this, we will first briefly describe Stackelberg games and then provide examples demonstrating that it is the right formalism for this extension.

\subsection{ADL is a Stackelberg Game}\label{subsec:adl-stackelberg}

\paragraph{Definition.}
A Stackelberg (leader–follower) game is sequential: the leader first \emph{commits} to a policy, the follower observes it, and then best–responds; the leader chooses its commitment anticipating that response~\citep{Stackelberg1934,FudenbergTirole1991,BasarOlsder1999}.\footnote{In finance and control, this appears as ``dynamic optimal control with commitment'' or a principal–agent model with rational expectations~\citep{Puterman1994,Bertsekas2017,LaffontMartimort2002}.}
Commitment here means the policy is observable and credible to followers.
The move order is therefore leader first, follower second, and the leader optimizes against the follower’s best–response set.
Over many rounds, the leader can \emph{learn} a near–optimal commitment by updating $\{\theta_t, h_t\}$ using feedback from realized follower responses and shocks, recovering standard online-learning guarantees when losses are convex (or have convex surrogates)~\citep{Zinkevich2003,ShalevShwartz2012,Hazan2019,conitzer2006}.
The Stackelberg view also supports robust design: the exchange can pick policies that hedge against worst–case follower responses and price/oracle shocks, steering the system toward states with smaller tail externalities~\citep{aghassi2006,Tambe2011}.

\paragraph{ADL and Stackelberg games.}
ADL naturally fits the setup of a Stackelberg game: the exchange (leader) \emph{publishes} its ADL policy each round by choosing a severity budget \(\theta_t\in[0,1]\) and a haircut rule.
As we saw in~\S\ref{subsec:adl}, the queueing rules used by Binance and Hyperliquid imply both a static severity and haircut rule and are published on their websites (which is their form of commitment).
Market participants (\eg~traders, LPs) and the stochastic environment (order flow, depth, liquidation behavior, oracle and price shocks) then respond as followers given the announced policy.
After market participants (followers) respond, the residual exposure (\ie~the remaining negative equity) is measured and the leader decides on a next action (which could include to stop the game and not perform an ADL policy).
The commitment of the exchange allows followers to construct a \emph{local} best response that can adapt as the number of rounds of the game increases.

\paragraph{Single vs.\ Multi-Round Models.}
One-shot (static) ADL treats each deficit in isolation, optimizing a fixed severity $\theta$ and haircut rule $h$ for that round.
The multi-round Stackelberg model links rounds: the leader commits to a policy path $\{\theta_t, h_t\}$, and followers adapt, creating state dependence via churn, liquidity, and leverage.
This coupling makes commitment valuable: a policy that is optimal in one round can be dominated once cross-shock feedback is considered.

\paragraph{Stackelberg vs.\ Nash.}
Nash equilibria are one–shot, \emph{simultaneous–move} fixed points where each player best–responds to the other’s action but no one commits \emph{ex ante}~\citep{FudenbergTirole1991}.
By contrast, Stackelberg equilibria embed the leader’s commitment advantage, often producing different (and more robust or revenue–enhancing) outcomes \citep{Stackelberg1934,BasarOlsder1999}.
Static ADL strategies such as venue queues (Binance/Hyperliquid) or single–round RAP can be analyzed under a Nash lens, as we assume there is no further response by a leader after a single execution of an ADL policy.
In practice, ADL controllers are used repeatedly (such as on October 10, 2025; see~\S\ref{sec:empirical-oct10}) which more resembles a Stackelberg game.
In Appendix~\ref{app:stack-nash} we give an explicit example for ADL where a Stackelberg equilibrium differs from a Nash outcome and in which the Stackelberg equilibrium strictly improves exchange revenue and robustness.

\subsection{Why do we need a dynamic model?}\label{subsec:why-dynamic}
Stackelberg games are inherently dynamic due to the alternating move order between the leader and follower.
A natural question is whether we really need to use a dynamic method or if a static method is `good enough.'
In this section, we will provide our first formal results showing that dynamic ADL rules are necessary if you want to optimize both time to recover solvency and exchange revenue after repeated shocks.
This will motivate the construction of the dynamic ADL policy in the sequel.

\paragraph{Exchange Long-Term Value (LTV) and the Trilemma.}
Recall the ADL Trilemma from \S\ref{sec:adl-trilemma}, which posits a conflict between Solvency (\textbf{S}), Fairness (\textbf{F}), and Revenue (\textbf{R}).
To quantify the \emph{Revenue} objective, we define the \emph{Exchange Long-Term Value} (LTV).
Unlike immediate fee capture, LTV accounts for the long-run impact of ADL on trader participation: heavy haircuts on winners cause them to exit or reduce activity (``churn''), eroding future fees.

For a policy $\pi$, let $h_{t,i} \in [0,1]$ be the haircut fraction on winner $i$ at time $t$.
We model trader retention via a revenue retention curve $r_{t,i}(h_{t,i})$ which is decreasing in the haircut size (larger haircuts drive away more volume).
The per-period expected revenue from trader $i$ is then $R_{t,i} = e_{t,i} r_{t,i}(h_{t,i})$, where $e_{t,i}$ is their equity.
The LTV is the discounted sum of these revenues:
\[
\mathrm{LTV}_T(\pi) = \sum_{t=0}^T \beta^t \sum_{i \in W_t} e_{t,i} \, r_{t,i}(h_{t,i}).
\]
Maximizing LTV corresponds directly to the \emph{(R)} vertex of the trilemma.

\paragraph{Competing Time Scales.}
To compare policies dynamically, we track two opposing clocks starting from a default event at $\tau_{\mathrm{def}}$: $\tau_{\mathrm{solv}}, \tau_{\mathrm{rev}}$.
Intuitively, $\tau_{\mathrm{solv}}$ is how quickly deficits are cleared whereas $\tau_{\mathrm{rev}}$ is how quickly fee revenue rebounds.
\begin{itemize}
    \item \emph{Solvency Recovery Time} ($\tau_{\mathrm{solv}}$): The time required for the insurance fund (or deficit) to fully recover to a safe level $\delta$.
    \item \emph{Revenue Recovery Time} ($\tau_{\mathrm{rev}}$): The time required for the exchange's expected LTV to return to pre-shock levels $(1-\epsilon)\mathrm{LTV}_{\text{pre}}$.
\end{itemize}
Formally,
\[
  \tau_{\mathrm{solv}} = \inf\{t \ge \tau_{\mathrm{def}} : \mathsf{IF}_t \ge \delta\}, \qquad
  \tau_{\mathrm{rev}} = \inf\{t \ge \tau_{\mathrm{def}} : \Expect[\mathrm{LTV}_t] \ge (1-\epsilon)\mathrm{LTV}_{\text{pre}}\}.
\]
As we will show, static policies trade these quantities off one another off.
We also show that the concentration of how haircuts are distributed amongst winning traders controls this trade-off.

\iparagraph{Example.}
We can rank policies by how concentrated their haircuts are using the Schur order (\S\ref{sec:risk-prelim}):
\begin{itemize}
    \item \emph{Queue (concentrated):} Haircuts hit the largest winners first, yielding a highly unequal (Schur-convex) vector.
    Deficits clear quickly (small $\tau_{\mathrm{solv}}$), but the biggest fee payers churn, so LTV recovers slowly (large $\tau_{\mathrm{rev}}$).
    \item \emph{Pro-Rata (diffuse):} Haircuts are spread evenly (Schur-concave).
    Solvency takes longer because shallow accounts contribute little (larger $\tau_{\mathrm{solv}}$), but everyone stays engaged, so revenue rebounds faster (small $\tau_{\mathrm{rev}}$).
\end{itemize}
Any intermediate policy traces this frontier: more concentration shortens $\tau_{\mathrm{solv}}$ but lengthens $\tau_{\mathrm{rev}}$.
This intuition is formalized in the following proposition, which states that any policy that concentrates haircuts (like Queue) will recover solvency faster but revenue slower than a diffuse policy (like Pro-Rata).

\paragraph{Opposing Schur orderings.}\label{subsec:opposite-orders}
Let $z_t(\pi)$ be the vector of post-haircut residual equities.
If policy $A$ is more concentrated than policy $B$ (i.e., $z_t(A) \succ_{\text{Schur}} z_t(B)$), then $A$ recovers solvency faster but revenue slower.
\begin{proposition}[Informal]\label{prop:opposing-orders-informal}
Under separable convex stage losses and monotone shock dynamics, if the residual exposure of policy $A$ majorizes policy $B$ (\eg~Queue vs.\ Pro-Rata), then:
\[
\tau_{\mathrm{solv}}(A) \le \tau_{\mathrm{solv}}(B) \quad \text{but} \quad \mathrm{LTV}_t(A) \le \mathrm{LTV}_t(B) \quad \text{for all } t.
\]
This implies no static policy can simultaneously minimize both recovery times.
\end{proposition}


\subsection{Incentive Compatibility and Constraints}\label{subsec:stackelberg-ic}
We frame the problem as a Stackelberg game: the exchange (leader) commits to a policy, and traders (followers) best-respond.
We impose follower incentive compatibility (FIC) and leader-side revenue and solvency constraints (LIC) to ensure the policy is feasible for traders and sustainable for the venue.

\paragraph{Follower Incentive Compatibility (FIC).}
Winners must remain solvent and within haircut caps:
\[
  0 \;\le\; h_{t,i} \;\le\; \beta_{t,i}, \qquad (1-h_{t,i}) e_{t,i}^+ \;\ge\; \underline{e}_{t,i},
\]
where $\beta_{t,i}$ enforces per-account haircut limits and $\underline{e}_{t,i}$ enforces minimum post-ADL equity (cf.\ \eqref{eq:haircut-constraint}--\eqref{eq:equity-constraint}).
These bounds guarantee feasibility (no negative winners) and rule out haircut paths that would violate venue-level haircut caps (as we saw in~\S\ref{subsec:adl}).

\paragraph{Leader Incentive Compatibility (LIC).}
Revenue must stay above a floor and solvency above a safety threshold:
\[
  \Expect[\mathrm{LTV}_t] \;\ge\; (1-\epsilon)\,\mathrm{LTV}_{\text{pre}}, \qquad
  \mathsf{IF}_t \;\ge\; \delta \;\;\text{or}\;\; \mathrm{CVaR}_\alpha(-\mathsf{IF}_t) \;\le\; \kappa.
\]
These ensure the policy is not revenue-destructive (per the churn model in \S\ref{sec:numerics}) and keeps the exchange solvent both on-path and in the tails (\S\ref{subsec:exchange-solvency}).
These constraints define the feasible set for any policy path $\{\theta_t, h_t\}$; we defer the full optimization program to Appendix~\ref{app:stack-nash}.

\paragraph{Stackelberg vs.\ Nash Equilibria.}
We model the interaction between the exchange and traders as a Stackelberg game where the exchange (leader) commits to a policy $\pi_t$, and traders (followers) subsequently respond (\eg~by adjusting liquidity or closing positions).
This commitment power is crucial.
As we show in Appendix~\ref{app:stack-nash} (Proposition~\ref{prop:stack-nash}), a simultaneous-move (Nash) game admits a ``bad'' equilibrium where traders withdraw liquidity anticipating high severity, which in turn forces the exchange to set high severity to cover the deficit, often violating the LIC floor on revenue.
In the Stackelberg model, the exchange commits to a policy that satisfies FIC/LIC and rules out this inefficient equilibrium, coordinating the market toward a high-liquidity, lower-severity outcome that keeps both trader feasibility and venue sustainability intact.

\subsection{Learning the Optimal ADL Policy}\label{subsec:learning-mdic}
Given these constraints, we can treat severity and haircuts as an online control problem: choose $\{\theta_t, h_t\}$ after each price shock, conditioned on past shocks and responses, to minimize loss while respecting FIC/LIC.
Because losses admit convex surrogates, we can apply classical online schemes (\eg~mirror descent) and enforce constraints via primal–dual updates so FIC/LIC remain satisfied at each round.
We use mirror descent because it naturally handles convex losses with simple geometry (RAP/Pro-Rata) and pairs cleanly with primal–dual updates, matching the Stackelberg setting of committing to a policy while respecting per-round constraints~\citep{Haghtalab2016StackelbergLearning,Haghtalab2016ThreeStrategies,Zinkevich2003,ShalevShwartz2012,Hazan2019}.

\paragraph{Mirror Descent with IC constraints (MDIC).}
We adopt a primal–dual mirror–descent scheme tailored to the RAP geometry and non–smooth constraints:
\begin{enumerate}
  \item Choose a mirror map \(\psi\) on the policy domain \(\Pi\) (\eg~a separable entropy or a Bregman geometry aligned with capped pro–rata).
  \item Initialize \(\pi_1\in\Pi\) and dual multipliers \(\lambda_1\!\ge 0\) for constraints \(c(\cdot;\cdot)\le 0\).
  \item For rounds \(t=1,2,\dots\):
  \begin{enumerate}
    \item Observe state \(s_t\); publish \(\pi_t\) (commitment). The follower/environment best–responds; incur loss \(\ell_t(\pi_t)\) and observe a subgradient \(g_t\in\partial \ell_t(\pi_t)\) and constraint feedback \(c_t=c(\pi_t;s_t)\).
    \item Primal update (mirror step with penalties):
    \[
      \pi_{t+1} \;=\; \argmin_{\pi\in\Pi}\; \Big\{\langle g_t + \nabla_\pi \langle \lambda_t, c(\pi;s_t)\rangle,\, \pi\rangle \;+\; \tfrac{1}{\eta_t} D_\psi(\pi\,\|\,\pi_t)\Big\}.
    \]
    \item Dual update (projected subgradient ascent on violations):
    \[
      \lambda_{t+1} \;=\; \big[\lambda_t + \gamma_t\, c_t\big]_+.
    \]
  \end{enumerate}
  \item Tune \((\eta_t,\gamma_t)\) as \(O(t^{-1/2})\) to balance regret and constraint violations.
\end{enumerate}
This controller reduces to standard OCO when constraints are absent, while handling non–smooth IC/solvency constraints via the dual iterates.

\paragraph{Why MDIC.}
The IC set is convex but non‑smooth (\eg~caps, max over scenarios).
Mirror descent matches this geometry (KL on the simplex, log‑barrier on $\theta$) and keeps updates stable while preserving $O(\sqrt{T})$ rates.

\paragraph{Algorithmic Guarantees and Why Regret Matters.}
If $f_t$ is Lipschitz, the follower is $\varepsilon$‑best‑response, and the IC set has bounded diameter, MDIC with $\eta_t\propto 1/\sqrt{T}$ yields
\[
\mathrm{Regret}_T = O\!\Big(\sqrt{T\,D_\phi(\theta^\star\|\theta_1)}+\sqrt{T\,D_\Phi(h^\star\| h_1)}\Big)+\varepsilon T,
\]
with per-round IC feasibility (or sublinear average violation in the long-term CVaR variant).
We care about regret because it benchmarks any adaptive policy against the best dynamic policy in hindsight.
Showing that MDIC achieves $O(\sqrt{T})$ regret means MDIC learns fast, while $\Omega(T)$ regret shows static rules (Queue, RAP, Pro-Rata) can be worst-case bad.
Proofs are in Appendix~\ref{app:regret}.

\paragraph{Static Policies Have Linear Regret.}
In an alternating two-regime example (liquidity vs.\ stress) any fixed $\theta$ incurs $\Omega(T)$ regret, whereas MDIC adapts and achieves $O(\sqrt{T})$ (proven in Appendix~\ref{app:regret}).
This captures why static Queue, RAP, Pro-Rata are fragile in adversarial environments.

\subsection{Follower Robustness}
So far, we have focused on the leader (\eg~the exchange), as they have to commit to a policy.
However, it is possible that if the follower (\ie~a trader who reduces the deficit) knows enough about the leader's strategy to adversely select against other traders.
We address the strategic behavior of traders (``followers'') in response to ADL policies, focusing on adverse selection and timing games.
Our description of follower behavior will be minimal with most details in Appendix~\ref{app:follower-robustness}.

\paragraph{Adverse Selection and the Pro-Rata Death Spiral.}
As we saw earlier, if an ADL policy socializes losses uniformly without regards to risk (like Pro-Rata), low-risk traders subsidize high-risk traders. 
If the cost of subsidy exceeds the utility of trading for low-risk traders, they exit.
This increases both the exchange's average risk and probability of future deficits and can be viewed as a form of adverse selection.
This is a direct consequence of the follower's strategic response: low-risk traders, observing a policy that penalizes them disproportionately, rationally choose to exit (or not participate), leaving the exchange with a riskier pool of traders.

Let $\mu_i$ be the baseline utility of trader $i$ and let $H_{\pi}(\lambda_i \mathcal{P}_n)$ be the random haircut variable for a trader with leverage $\lambda_i$.
The "utility is greater than subsidy" condition can be described as $\mu_i - \Expect[H_{\pi}(\lambda_i, \mathcal{P}_n)] \ge u_0$, where $u_0$ is a baseline minimum utility a low-risk trader accepts.
Under Pro-Rata, $\Expect[H_{\pi_{PR}}]$ depends on the average leverage $\overline{\lambda} = \frac{1}{n}\sum_i \lambda_i$, which penalizes low-risk traders with $\lambda_i < \overline{\lambda}$.
The worst case is if there is a `death spiral' of exits, where low-risk traders leaving sequentially increase the probabilty of other traders leaving, as the average $\overline{\lambda}$ increases on each exit.
This is adverse selection that the exchange faces as the pool of traders becomes increasingly risky from a leverage standpoint.
We formalize this in Appendix~\ref{app:follower-adverse}~and provide explicit examples of a death spiral.

\paragraph{The Waiting Game.}
Traders may act as ``Backstop Liquidity Providers'' (BLP) by voluntarily filling liquidation orders to reduce the deficit $D_t$.
For instance, traders who deposit to Hyperliquid's HLP pool around the time of an ADL event act as BLPs.
This presents an optimal stopping problem: a BLP can absorb the deficit immediately (incurring market risk and execution costs) or wait, hoping that the exchange reduces ADL severity $\theta_t$ (and thus the potential haircut) in subsequent rounds.
If the exchange employs a policy like an exponential backoff that lowers severity over time, it incentivizes traders to wait, potentially exacerbating the deficit.
To prevent such strategic delays, the exchange needs to provide a financial incentive to ensure that BLPs add liquidity in a timely manner.

This incentive often comes in the form of yield premium $\Gamma_t$. 
This premium can be viewed as an incentive (e.g., a rebate or a yield) that the exchange must pay to compensate BLPs for the risk of absorbing the deficit early.
Intuitively, $\Gamma_t$ quantifies the cost of immediate intervention, capturing the execution slippage and inventory risk premium required for a BLP to step in.
By adjusting these incentives, an exchange can encourage "no-wait" conditions to speed up solvency recovery.
Intuitively, these ``No-Wait'' conditions correspond to solvency-like constraints: $\theta_t D_t \ge \Gamma_t$.
We formalize this with the following proposition:
\begin{proposition}[No-Wait Condition]\label{prop:no-wait}
  Let $\Gamma_t$ be the liquidity premium required for a Backstop Liquidity Provider (BLP) to absorb deficit $D_t$ immediately rather than waiting for $\Delta t$.
  If the severity policy $\theta_t$ is a decreasing function of time (\eg~exponential backoff), a BLP will withhold liquidity unless the immediate haircut cost is cheaper than the discounted future cost:
  \[
    \theta_t D_t \;\le\; \Expect_t[\theta_{t+\Delta t} D_{t+\Delta t}] - \Gamma_t.
  \]
  Violating this condition induces a ``Waiting Game'' where solvency recovery is delayed purely by the follower's strategy.
\end{proposition}  
\noindent Appendix~\ref{app:follower-wait} details the construction of $\Gamma_t$ and proves this proposition.

\subsection{Dynamic Phase Transition in ADL}
\label{subsec:dynamic-phase-transition}

In Section~\ref{sec:negative}, we established that static ADL rules face fundamental scaling limits.
However, our results in that section were qualitative in nature, making it hard to use those results to design better mechanisms.
One question that one might have is if there is a way to classify, based on properties of an asset (\eg~volatility, liquidity, etc.) whether one can use a simple static ADL strategy or if one needs a full dynamic strategy.
Here we will construct conditions for quantifying an answer to this question.

To quantify the efficiency gap between static (Nash) and dynamic (Stackelberg) control, we evaluate the system using welfare functions $W(\pi)$ corresponding to the trilemma (Proposition~\ref{thm:adl-trilemma-formal}) objectives.
Welfare functionals quantify how fair, profitable, or solvent the entire system is (as opposed to individual traders).
Specifically, we show there is a \emph{phase transition} that depends on the liquidity and volatility of the asset pair that determines when one needs a dynamic methods (like MDIC) or if one can use a simpler algorithm (like RAP).
The phase transition is defined by an \emph{order parameter}, which is a dimensionless quantity that captures the relative strength of the ADL severity to the market's liquidity and volatility.
This phase transition occurs in the \emph{efficiency gap} between static and dynamic methods --- \ie~how much better at preserving solvency and/or revenue that a dynamic method is over static method.

\paragraph{Price of Anarchy for Nash vs.\ Stackelberg.}
The Price of Anarchy (PoA) is the standard metric for quantifying efficiency loss from strategic behavior across equilibrium concepts~\cite{Roughgarden2005SelfishRouting,roughgarden15intrinsic}.
In Stackelberg settings, a leader's commitment reshapes follower best responses and can systematically lower inefficiency relative to simultaneous-move (Nash) play.
This is well documented in dynamic game models~\cite{BasarOlsder1999} and in smoothness-based PoA analyses that extend uniformly across equilibrium notions, including sequential and composed mechanisms~\cite{SyrgkanisTardos2013,roughgarden15intrinsic}.
In our ADL context, the venue is a natural leader that can commit to a severity path, while traders react myopically to current and anticipated haircuts.
This makes a PoA comparison between one-shot (Nash-like) and multi-round (Stackelberg) control both natural and informative.

\noindent We define the PoA for a welfare function $W$ as the ratio of the optimal dynamic Stackelberg welfare to the welfare of a static, one-shot policy $\pi^{\mathrm{stat}}_n$:
\begin{equation}
  \label{eq:poa-def}
  \mathrm{PoA}_n(W) = \frac{W(\pi^{\mathrm{dyn}}_n)}{W(\pi^{\mathrm{stat}}_n)}
\end{equation}
We consider both a fairness welfare $W_{\text{Fair}}(\pi) = \mathsf{PTSR}_T(\pi)$ and a revenue welfare $W_{\text{LTV}}(\pi) = \text{LTV}_T(\pi)$.
The severity load $\kappa_n = \theta_n n / b_n$ serves as the order parameter for the phase transition in fairness, while the heavy-tail index $\alpha$ (see~\S\ref{sec:risk-metrics}) governs the phase transition in revenue.
Using extreme-value scaling and mean-field equilibrium assumptions, we obtain PoA bounds that are robust to the precise equilibrium concept. Appendix~\ref{app:smoothness-poa} formalizes the scaling assumptions and presents the theorem. Informally:

\begin{proposition}[Informal]\label{prop:poa-phase-informal}
  Let $\kappa_n = \theta n/b_n$ be the normalized severity relative to the winner equity scale $b_n$.
  \begin{itemize}
    \item \emph{Bounded PoA (Static Sufficient):} If the severity load stays bounded ($\sup_n \kappa_n < \infty$), leverage masses are comparable ($\ell_n^+ \asymp \ell_n^-$), and deficits are light-tailed, then static ADL is constant-factor optimal: $\sup_n \mathrm{PoA}^{\text{Fair}}_n < \infty$ and $\sup_n \mathrm{PoA}^{\text{LTV}}_n < \infty$.
    \item \emph{Unbounded PoA (Dynamic Necessary):} If the load diverges ($\kappa_n \to \infty$) or leverage is concentrated ($\ell_n^+/\ell_n^- \to 0$), Fairness PoA blows up. If deficits are heavy-tailed ($\alpha<2$) and solvency is enforced, Revenue PoA also diverges. Static policies must sacrifice at least one objective, while a dynamic controller can maintain both.
  \end{itemize}
  This implies a screening rule: if estimates satisfy $\widehat{\kappa}_n \le K$ and $\widehat{\Lambda}_n \in [c, C]$, static ADL suffices; otherwise, dynamic control is required to navigate the Trilemma.
\end{proposition}

\noindent Practically, this result can be used to provide simple heuristics for screening assets using historical data to determine whether a simple static rule suffices or if a dynamic policy is required.
Formal statements and proofs are given by Theorem~\ref{thm:poa-phase}
and Proposition~\ref{prop:ltv-poa}
in Appendix~\ref{app:smoothness-poa}, and by
Example~\ref{ex:light-tailed-unbounded-poa}.

\section{Empirical Analysis under Partial Observability: The October 10 Event}\label{sec:empirical-oct10}\label{sec:numerics}

\paragraph{Partial Observation model.}
Our reconstruction uses publicly available fills, marks, and account-level aggregates.
It is not a full execution and clearing ledger: complete settlement logic, internal netting, fee and buffer accounting, and some timing details are unobserved.
These observables are sufficient to compute our empirical targets (deficits and winner-side PNL capacity) under explicit definitions.
Accordingly, strict accounting equalities are treated as diagnostics rather than identities.
Throughout this section, ``counterfactual'' refers to counterfactuals under this public-data observation model, not to a complete ledger identity.

We evaluate the proposed ADL mechanisms using a reconstruction\footnote{We thank Mauricio Trujillo (\href{https://x.com/ConejoCapital}{@conejocapital}) and Hydromancer for providing, cleaning, and validating this dataset.} of the October 10, 2025 liquidation cascade on Hyperliquid~\cite{CoinDesk2025LargestLiquidations}.
We use this replay to reconstruct production wealth removed from winners via a two-pass no-ADL counterfactual, compute the instantaneous bankruptcy-gap proxy \(B_t^{\mathrm{needed}}\), and compare production to transparent benchmark allocations under a profits-only (PNL) haircut-capacity constraint.

\subsection{Experimental Setup}

\paragraph{Regret Framework.}
We evaluate ADL policies using two complementary regret components:

\begin{definition}[Fairness Regret]
The fairness regret measures the maximum deviation from pro-rata allocation:
\[
\mathcal{R}_t^{\mathrm{fair}}(\pi) := \max_{i \in W_t} \frac{h_{\pi,i}}{e_i} - \frac{h_{\mathrm{pr},i}}{e_i}
\]
where $h_{\pi,i}$ is the haircut under policy $\pi$ and $e_i$ is user $i$'s equity.
\end{definition}

\begin{definition}[Overshoot Regret]
The overshoot regret measures excess budget beyond the needed amount:
\[
\mathcal{R}_t^{\mathrm{overshoot}}(\pi) := B_t^{\pi} - B_t^{\mathrm{needed}}
\]
\end{definition}

Total regret is the sum of these components. Policies with low overshoot but highly concentrated losses, and policies with diffuse losses but systematic overshoot, are both suboptimal under this decomposition.

\paragraph{Dataset and wave partition.}
We utilize the canonical public-data reconstruction of the event~\cite{HyperMultiAssetedADL}, covering a 12-minute window (21:16--21:27 UTC) with \$2{,}103{,}111{,}431 in liquidations across roughly 160 assets.
We partition ADL fills into \emph{global time waves} by gap clustering: a new wave begins when the inter-fill gap exceeds 5 seconds, yielding $T=16$ waves for this window.

\paragraph{Deficit and haircut capacity.}
For each wave $t$, the realized loser deficit is
\[
D_t=\sum_{j \in \mathrm{losers}(t)}(-e_t(j))_+,
\]
computed from loser-side liquidation equity fields (using the canonical table's \texttt{liquidated\_total\_equity}, with a per-liquidated-user minimum within a wave to avoid double counting).
Across the event, the realized deficit sums to \(\sum_t D_t \approx \$100.1\)M.
Haircuts apply to profits (PNL), not protected cash/collateral.
For benchmark allocations we use a per-user capacity proxy \(c_u := \min(U(u),E(u))\), where \(U(u)\) is positive unrealized PNL and \(E(u)\) is positive equity, both taken from the canonical REALTIME table within the wave window.

\paragraph{Needed budget and production wealth removed.}
We compute \(B_t^{\mathrm{needed}}\) from mark--execution gaps in the raw fill stream and \(H_t^{\mathrm{prod}}(\Delta)\) via two-pass replay as described below.
At \(\Delta=0\), totals satisfy \(\sum_t B_t^{\mathrm{needed}} \approx \$15.1\)M and \(\sum_t H_t^{\mathrm{prod}} \approx \$60.1\)M, hence \(O(0)\approx \$45.0\)M.

\paragraph{Regime check .}
To assess whether the October 10 episode plausibly lies in the structural-deficit regime ($\mu_- \gg \mu_\Phi$), we run an order-of-magnitude check.
The realized event deficit is \(\sum_t D_t \approx \$100.1\)M.
Using the venue fee schedule and event turnover, a conservative upper bound for sustainable fee diversion is approximately 2--5\% of gross fees; with gross fees in the \$1--5M range over the window, this implies \(\mu_\Phi\) on the order of \$20k--\$250k.
The realized deficit is therefore roughly two orders of magnitude larger, consistent with $\mu_- \gg \mu_\Phi$ in this episode.
This is a plausibility check, not identification: the formal trilemma statement still relies on Assumption J.3.

\paragraph{Benchmark allocations.}
We compare production to transparent benchmark allocations that target \(B_t^{\mathrm{needed}}\) on each wave under the same profits-only capacity constraint:
\begin{itemize}
    \item \emph{Wealth pro-rata (continuous):} a capped pro-rata allocation in wealth-space USD.
    \item \emph{Vector mirror descent (vector-md):} a projection-based allocator that matches the budget using a single projection step onto \(\{x\in[0,1]^n: c^\top x=B\}\).
    \item \emph{Contract pro-rata:} standard exchange-style integer allocation proportional to position size, used by Binance, Bybit, and similar venues.
    \item \emph{Min-max ILP:} MIP solver minimizing the maximum haircut percentage $\max_i h_i/e_i$ subject to budget and capacity constraints.
\end{itemize}
These benchmarks are used to quantify the mechanically necessary overshoot under discreteness and capacity constraints; they are not claims about production implementability.

\paragraph{Integer rounding and budget shortfall.}
A key difference between continuous and integer-contract allocation is rounding error.
In contract space, standard pro-rata rounding (downward rounding of fractional contracts) can induce substantial under-allocation relative to \(B_t^{\mathrm{needed}}\).
In our reconstruction, contract pro-rata exhibits per-wave shortfalls up to \$371{,}000 (wave 1) and cumulative under-allocation of approximately \$1.5M across the episode.
Thus, standard exchange-style integer allocation does not reliably hit the target budget without additional adjustment.
The min-max ILP benchmark enforces budget matching while minimizing the maximum haircut percentage, achieving near-exact targeting with only small positive overshoot (cumulative \$50{,}000).
This shows that tight budget control is feasible in integer contract space, but not via simple rounding alone.

\paragraph{Revenue proxy.}
We construct a \emph{revenue proxy} to estimate expected fee loss induced by ADL haircuts.
The proxy maps haircut intensity to post-event account churn, consistent with retention modeling in the economics and marketing literature~\cite{Acerbi2002,Bolton1998Duration,Schmittlein1987Customers}.
Formally, let the churn probability for winner \(i\) in wave \(t\) be the probability that the account exits after receiving haircut \(h_{t,i}\).
We model this probability with an exponential hazard:
\[
p_{t,i}=1-\exp(-\beta h_{t,i}/w_{t,i}).
\]
This specification is approximately linear for small haircuts and saturates toward one for severe haircuts, matching the account-level attrition patterns observed in the replay.
It is also consistent with survival-style retention curves used in both classical and modern churn studies~\cite{Bolton1998Duration,Schmittlein1987Customers,AscarzaHardie2013UsageChurn,LemmensGupta2020Churn,Ascarza2018RetentionFutility}.
Full details of the calculation and the fee proxy are in Appendix~\ref{app:revenue-proxy}.

\paragraph{Overshoot measurement methodology.}
To measure overshoot, we compare production ADL (executed in contract space) with wealth-space comparator policies.
This requires a \emph{two-pass replay} on the same realized price path.

\emph{Two-pass replay.}
For each ADL wave $t$, we define two equity outcomes for each position $p$:
\begin{itemize}
  \item $e_t^{\mathrm{ADL}}(p)$: the realized equity after production ADL executes in wave $t$ (observed outcome).
  \item $e_t^{\mathrm{no\text{-}ADL}}(p)$: a counterfactual equity if ADL were disabled in wave $t$, holding the realized price path fixed (reconstructed outcome).
\end{itemize}
The \emph{induced haircut} from production ADL is then:
\[
h_t^{\mathrm{prod}}(p) = \begin{cases}
\dfrac{e_t^{\mathrm{no\text{-}ADL}}(p) - e_t^{\mathrm{ADL}}(p)}{e_t(p)}, & e_t(p) > 0, \\[6pt]
0, & \text{otherwise},
\end{cases}
\]
where $e_t(p)$ is the pre-ADL equity (which equals $e_t^{\mathrm{no\text{-}ADL}}(p)$ by construction of the counterfactual).
Total wealth removed from winners in wave $t$ is:
\[
H_t^{\mathrm{prod}} = \sum_{p \in \mathcal{W}_t} \bigl(e_t^{\mathrm{no\text{-}ADL}}(p) - e_t^{\mathrm{ADL}}(p)\bigr)_+.
\]

\emph{The no-ADL counterfactual and identification assumption.}
The counterfactual \(e_t^{\mathrm{no\text{-}ADL}}(p)\) is not directly observed because production ADL executed.
We recover it by replaying the realized price path with ADL disabled, under the short-horizon assumption that positions remain open and equity evolves with marks over the execution window.
This assumption is necessary because trader actions in the unobserved world (manual closes, leverage changes, transfers) are not observed.
Given that ADL waves occur on a seconds-scale during high volatility, this assumption is appropriate for measuring immediate wealth transfer effects.

\emph{Needed budget and production overshoot vs needed.}
For each ADL fill $k$ with liquidation mark $p_k^{\mathrm{mark}}$, execution price $p_k^{\mathrm{exec}}$, and size $q_k$, define the instantaneous bankruptcy-gap proxy
\[
\mathrm{needed}_k := \lvert p_k^{\mathrm{mark}} - p_k^{\mathrm{exec}}\rvert\cdot \lvert q_k\rvert.
\]
We aggregate within each wave $t$:
\[
B_t^{\mathrm{needed}} := \sum_{k \in \text{fills in wave }t}\mathrm{needed}_k.
\]
We define the \emph{production overshoot vs needed}:
\[
O_t(\Delta) := H_t^{\mathrm{prod}}(\Delta) - B_t^{\mathrm{needed}},\qquad O(\Delta):=\sum_t O_t(\Delta).
\]
This compares realized wealth removed from winners to the minimal instantaneous transfer implied by closing across the mark--bankruptcy gap.

\emph{Evaluation horizon \(\Delta\).}
The quantity \(B_t^{\mathrm{needed}}\) is computed from realized ADL fills and is independent of \(\Delta\).
By contrast, \(\Delta\) changes the evaluation time for winner equities in \(H_t^{\mathrm{prod}}(\Delta)\): we evaluate \(e^{\mathrm{ADL}}\) and \(e^{\mathrm{no\text{-}ADL}}\) at \(t_{\mathrm{eval}}=t_{\mathrm{end}}+\Delta\) while holding the realized post-wave path fixed.
This isolates a short-horizon opportunity-cost channel.
Over \(\Delta\in\{0,500,1000,2000,5000\}\)ms, we obtain \(O(0)\approx \$45.0\)M and a range of approximately \$45.0M--\$51.7M.

\subsection{Results}

\paragraph{Headline numbers (two numeraires).}
Figure~\ref{fig:emp-headlines} reports two headline numeraires:
the stylized \emph{queue overshoot} diagnostic in equity dollars (wealth space) and the \emph{production overshoot vs needed} metric in PNL dollars (two-pass replay at \(\Delta=0\)).
The figure also reports the simple mapping obtained by dividing the equity-dollar queue diagnostic by the trimmed-mean winner overcollateralization ratio (equity/PNL).

\begin{figure}[t]
  \centering
  \IfFileExists{oss_figures/01_headlines.png}{\includegraphics[width=0.82\linewidth]{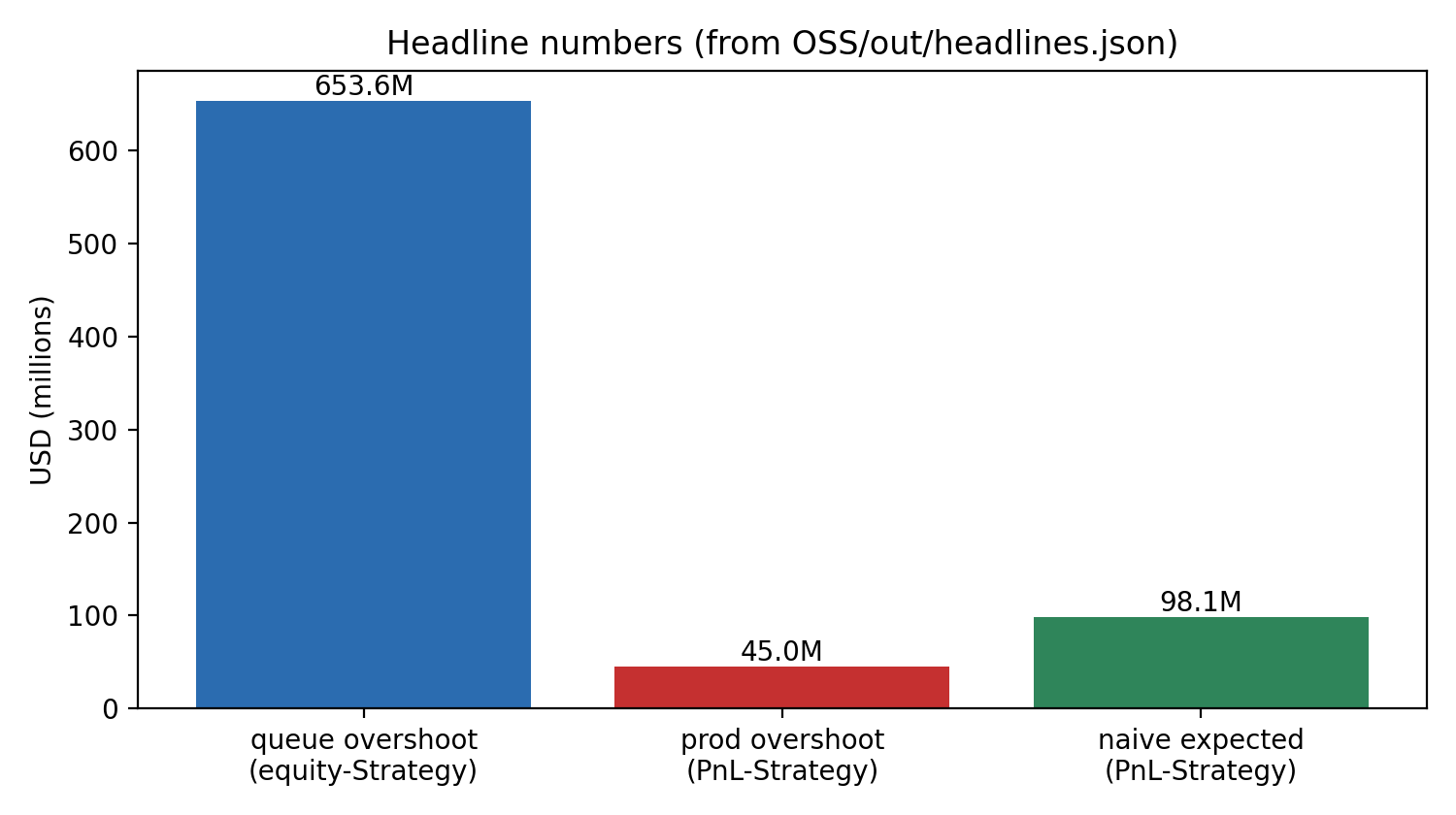}}{\rule{0pt}{2in}\rule{0.9\linewidth}{0pt}}
  \caption{Headline metrics from the October 10, 2025 event reconstruction. The stylized queue overshoot (\$653.6M in equity dollars) is a theoretical wealth-space diagnostic; production overshoot vs needed (\$45.0M in PNL dollars) is the empirically measured excess from two-pass replay at \(\Delta=0\). The 6.66\(\times\) overcollateralization ratio connects these two numeraires.}
  \label{fig:emp-headlines}
\end{figure}

\paragraph{Overshoot vs horizon (opportunity-cost sensitivity).}
Figure~\ref{fig:emp-overshoot} plots \(O(\Delta)\) from the horizon sweep.
As \(\Delta\) increases, \(H_t^{\mathrm{prod}}(\Delta)\) may rise because positions are marked later along the realized path, revealing the short-horizon opportunity-cost channel.

\begin{figure}[t]
  \centering
  \IfFileExists{oss_figures/02_overshoot_vs_horizon.png}{\includegraphics[width=0.82\linewidth]{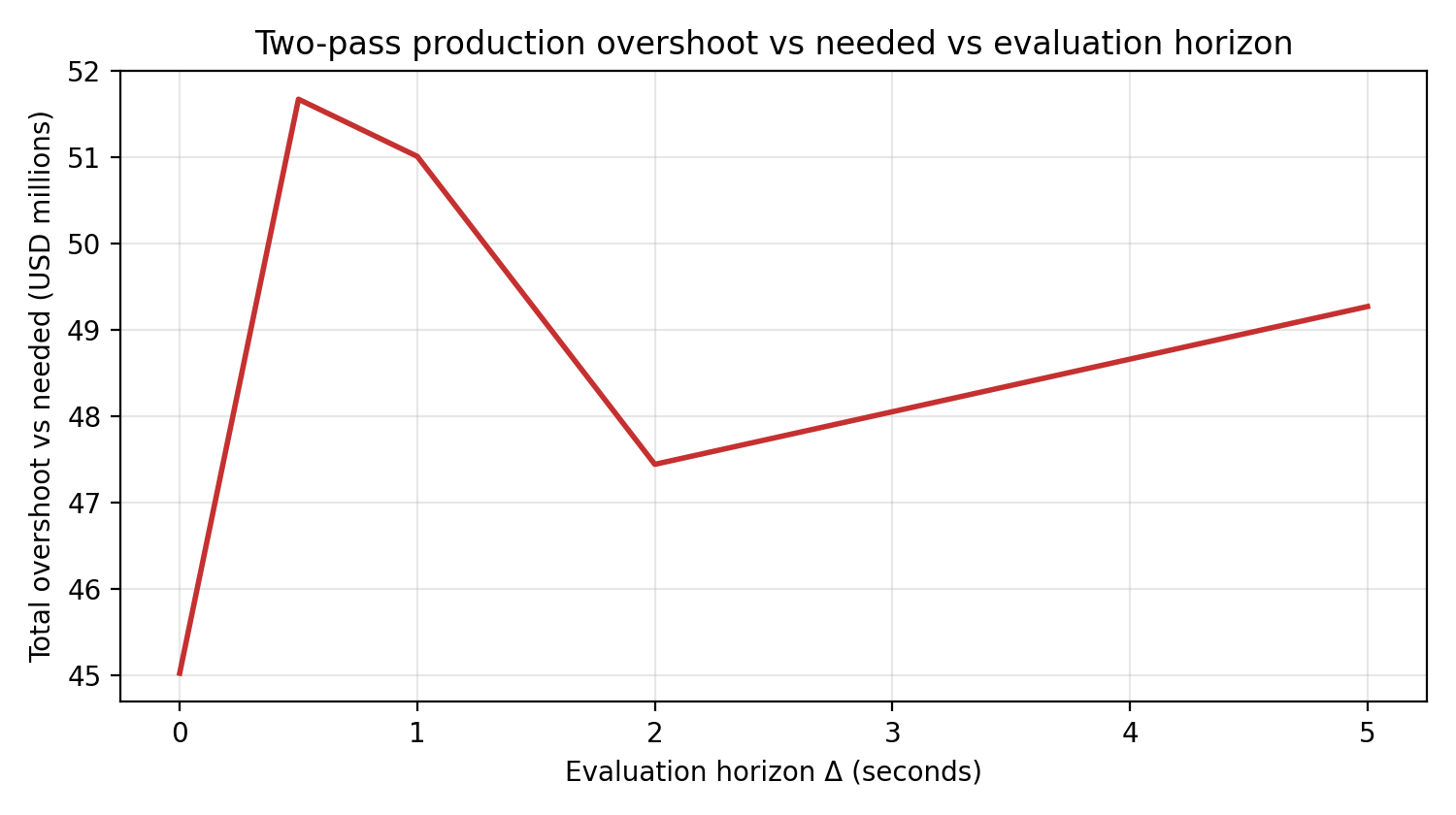}}{\rule{0pt}{2in}\rule{0.9\linewidth}{0pt}}
  \caption{Total production overshoot vs needed \(O(\Delta)\) as a function of evaluation horizon \(\Delta\). The overshoot increases from \$45.0M to \$51.7M as \(\Delta\) grows, reflecting opportunity-cost channels: positions marked later along the realized price path capture additional value that would have accrued to winners.}
  \label{fig:emp-overshoot}
\end{figure}

\paragraph{Production vs benchmark allocations (per wave).}
To quantify excess relative to transparent alternatives, we compare production with benchmark allocations that target \(B_t^{\mathrm{needed}}\) on each wave.
Benchmarks use a PNL-only haircut-capacity proxy constructed from the canonical REALTIME table (Appendix~\ref{app:methodology}).
Figure~\ref{fig:emp-policy-per-wave} shows per-wave budgets, overshoot vs needed, and a concentration diagnostic.

\begin{figure}[t]
  \centering
  \IfFileExists{oss_figures/05_policy_per_wave_performance.png}{\includegraphics[width=0.98\linewidth]{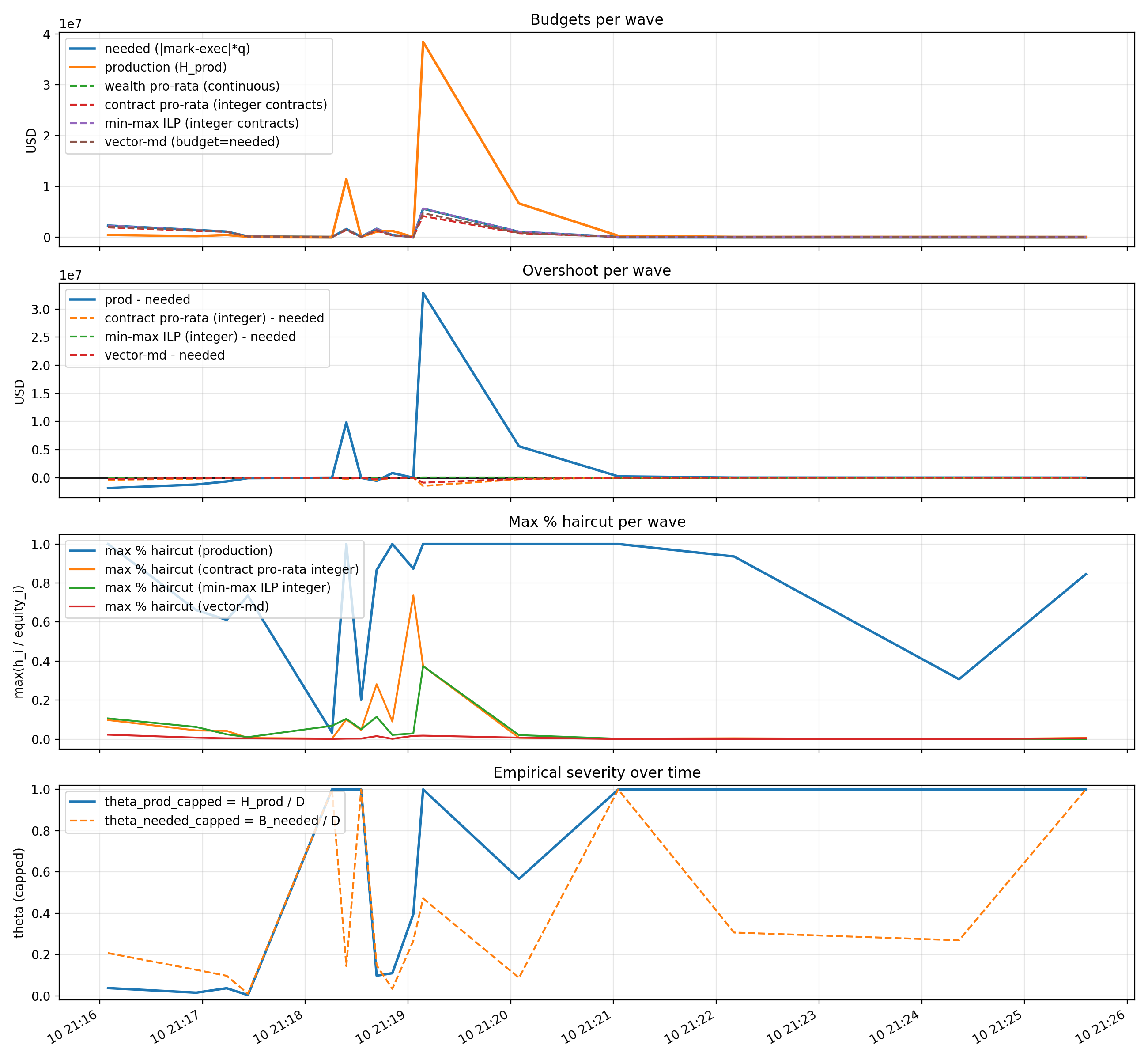}}{\rule{0pt}{2in}\rule{0.9\linewidth}{0pt}}
  \caption{Production vs benchmark allocations per wave. The ``needed'' budget \(B_t^{\mathrm{needed}}\) (bankruptcy-gap proxy) is compared against production haircuts \(H_t^{\mathrm{prod}}\) and four benchmark policies: wealth pro-rata, vector mirror descent, contract pro-rata, and min-max ILP. Benchmarks achieve near-exact budget matching while production consistently overshoots, empirically validating the trilemma's prediction that queue-based allocation sacrifices efficiency (Proposition~\ref{prop:adl-trilemma}).}
  \label{fig:emp-policy-per-wave}
\end{figure}

\paragraph{Cumulative overshoot vs needed.}
Figure~\ref{fig:emp-cumulative-overshoot} shows the aggregate effect.
Benchmarks that target \(B_t^{\mathrm{needed}}\) maintain near-zero cumulative overshoot vs needed (up to integer rounding and capacity caps), while production accumulates substantial excess.

\begin{figure}[t]
  \centering
  \IfFileExists{oss_figures/06_policy_per_wave_cumulative_overshoot.png}{\includegraphics[width=0.98\linewidth]{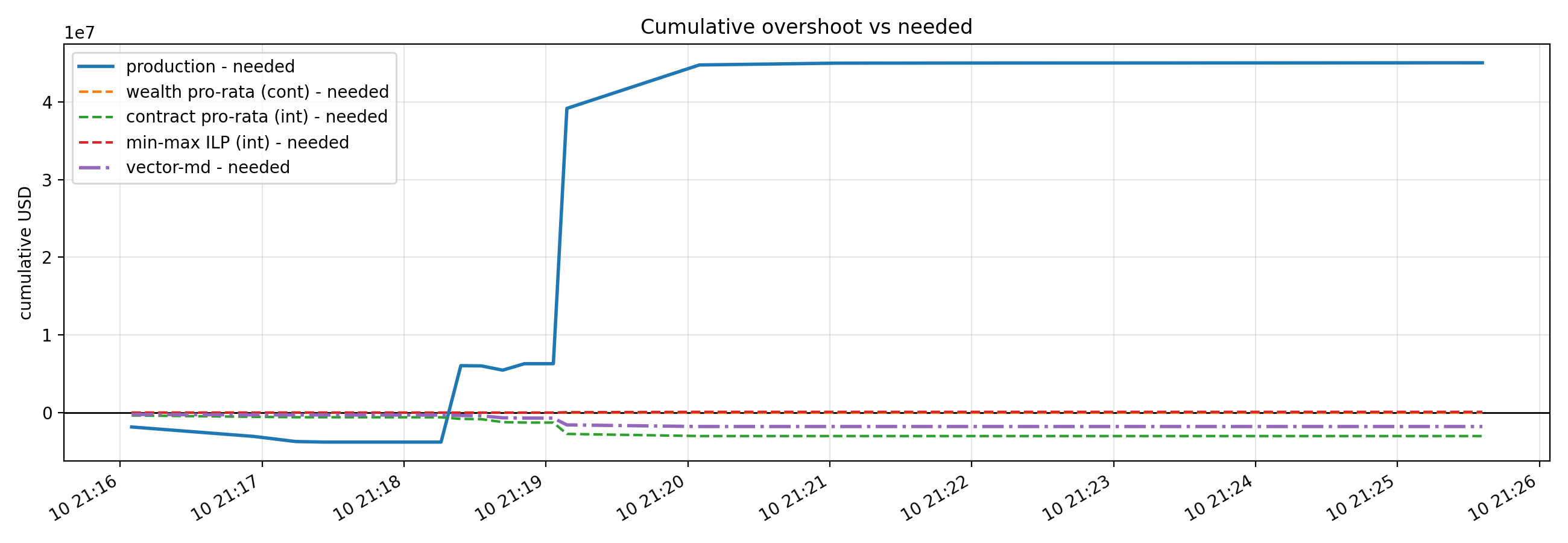}}{\rule{0pt}{2in}\rule{0.9\linewidth}{0pt}}
  \caption{Cumulative overshoot vs needed over the 16 waves. Production accumulates \$45.0M in excess haircuts while benchmark policies maintain near-zero cumulative overshoot (bounded only by integer rounding and capacity constraints). This demonstrates that large overshoot is not mechanically inevitable---it is a policy choice.}
  \label{fig:emp-cumulative-overshoot}
\end{figure}

\paragraph{Regret decomposition.}
We decompose total regret into overshoot regret and fairness regret, as defined in the Regret Framework.
Figure~\ref{fig:emp-cumulative-regret} shows cumulative regret trajectories over the 16 waves for production versus benchmark allocations.
Figures~\ref{fig:emp-overshoot-regret}--\ref{fig:emp-total-regret} provide detailed breakdowns of each regret component separately.

\begin{figure}[t]
  \centering
  \IfFileExists{oss_figures/09_cumulative_regret_historical.png}{\includegraphics[width=0.95\linewidth]{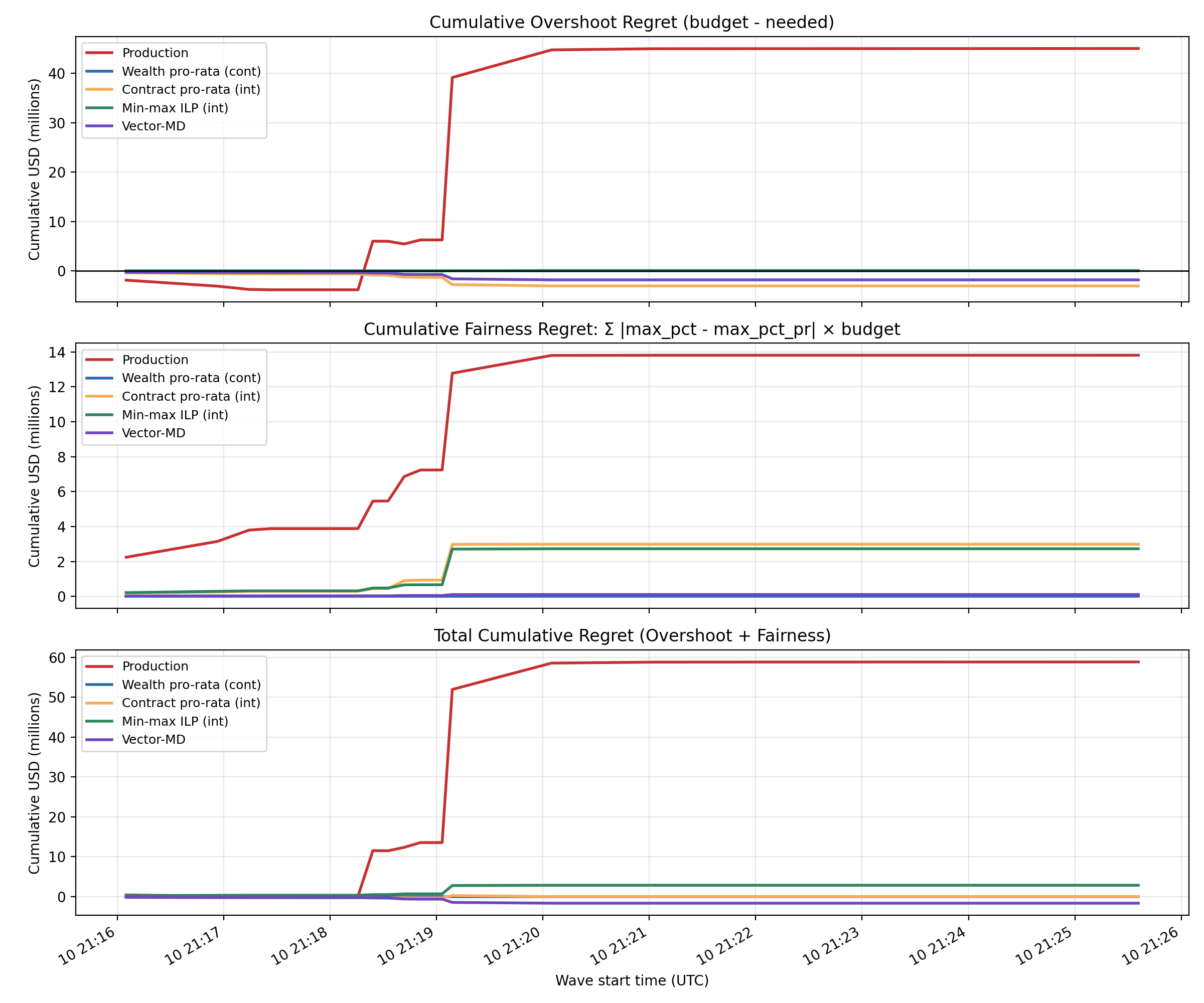}}{\rule{0pt}{2in}\rule{0.9\linewidth}{0pt}}
  \caption{Cumulative regret over the 16 waves for production and benchmark policies (wealth pro-rata, contract pro-rata, min-max ILP, and vector-MD). Top panel: overshoot regret (budget minus needed). Middle panel: fairness regret (deviation from pro-rata). Bottom panel: total regret. Production accumulates regret roughly 10\(\times\) faster than pro-rata benchmarks.}
  \label{fig:emp-cumulative-regret}
\end{figure}

\begin{figure}[t]
  \centering
  \IfFileExists{oss_figures/10a_overshoot_regret.png}{\includegraphics[width=0.95\linewidth]{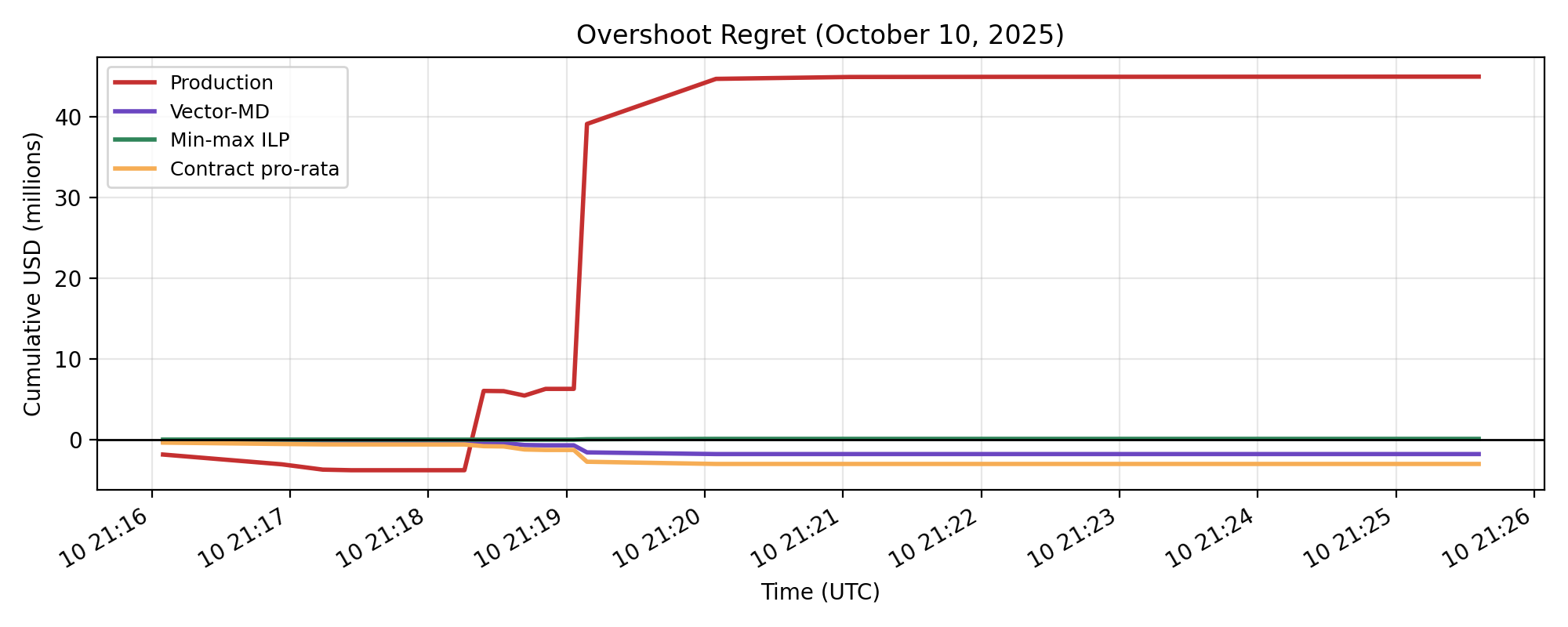}}{\rule{0pt}{1.5in}\rule{0.9\linewidth}{0pt}}
  \caption{Cumulative overshoot regret by policy. Production accumulates \$45.0M in overshoot (excess haircuts beyond needed), while benchmark policies achieve near-zero overshoot. This confirms that large overshoot is a policy choice, not a mechanical necessity.}
  \label{fig:emp-overshoot-regret}
\end{figure}

\begin{figure}[t]
  \centering
  \IfFileExists{oss_figures/10b_fairness_regret.png}{\includegraphics[width=0.95\linewidth]{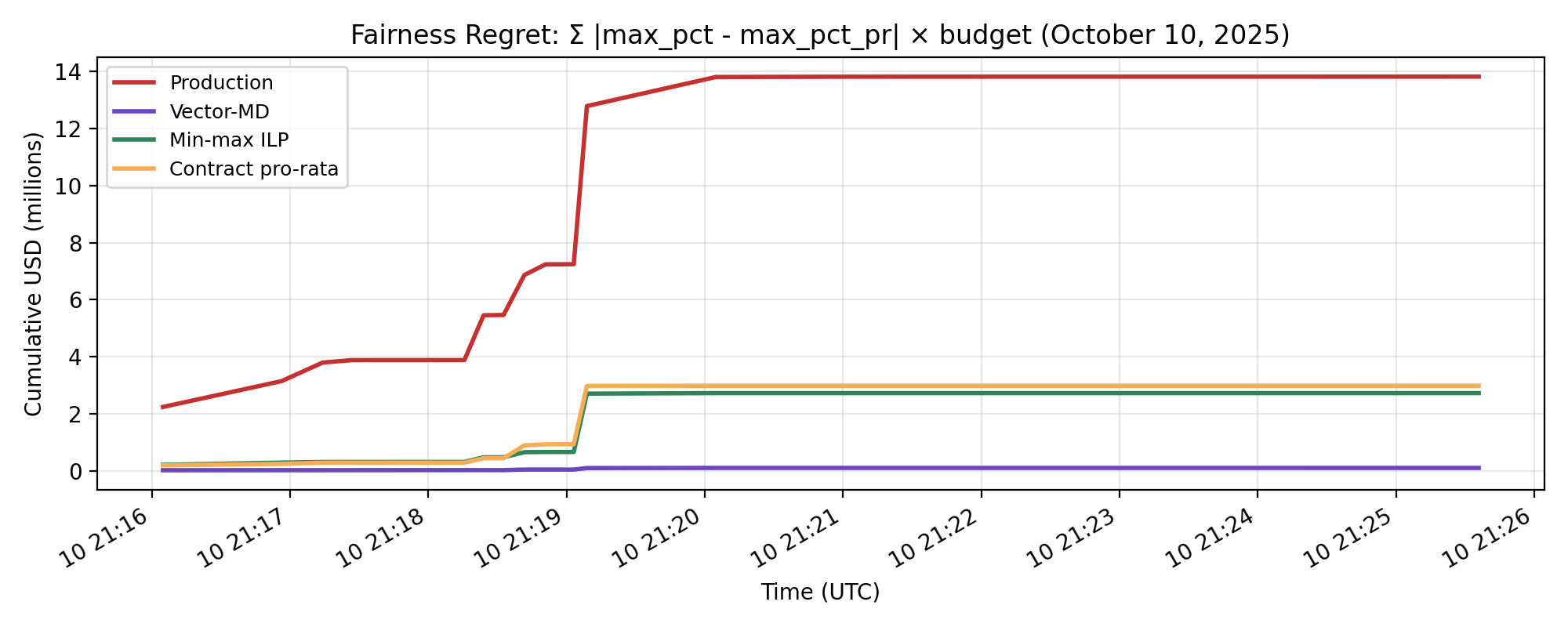}}{\rule{0pt}{1.5in}\rule{0.9\linewidth}{0pt}}
  \caption{Cumulative fairness regret by policy. Production's regret is dominated by fairness violations: the worst-hit winner receives 75.45\% haircuts on average, compared to 11.53\% for contract pro-rata and 6.21\% for min-max ILP. Vector-MD achieves the lowest fairness regret (0.71\%).}
  \label{fig:emp-fairness-regret}
\end{figure}

\begin{figure}[t]
  \centering
  \IfFileExists{oss_figures/10c_total_regret.png}{\includegraphics[width=0.95\linewidth]{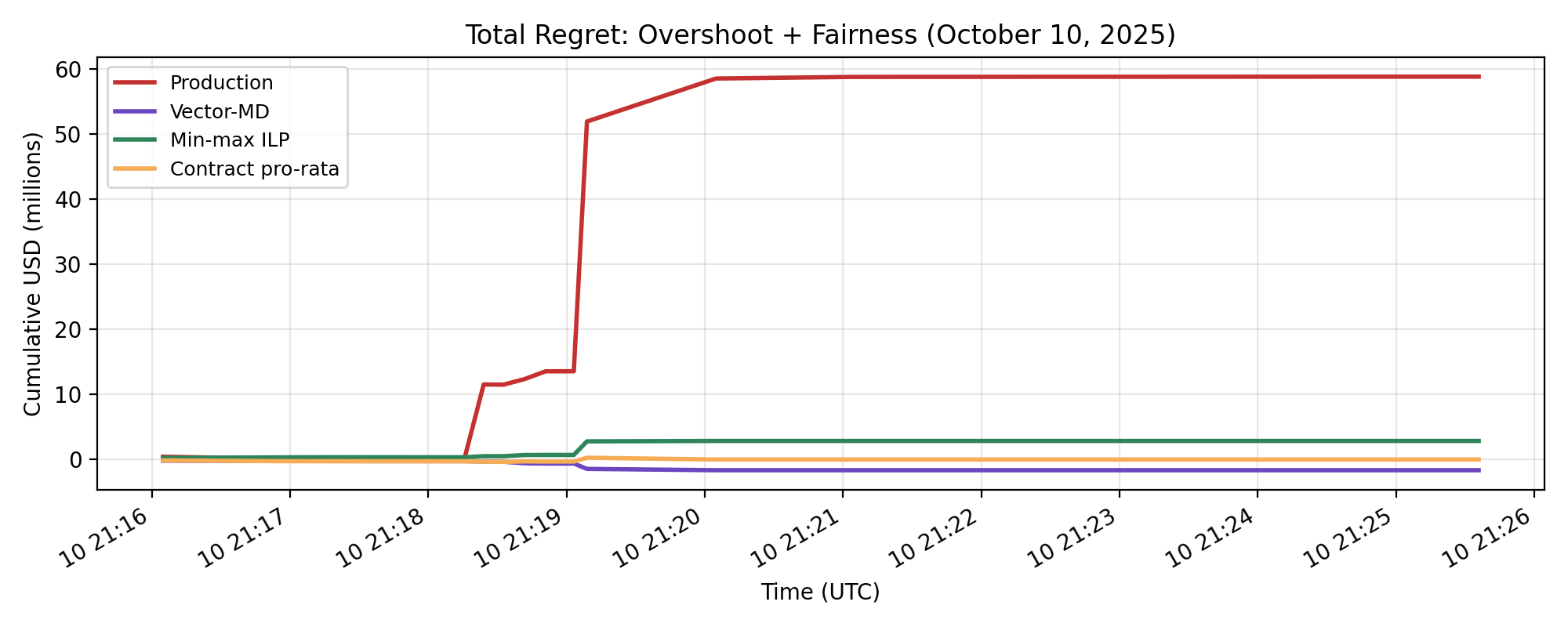}}{\rule{0pt}{1.5in}\rule{0.9\linewidth}{0pt}}
  \caption{Total cumulative regret (overshoot + fairness) by policy. The dominant contributor to production's total regret is fairness violations, not overshoot. Benchmark policies achieve bounded total regret through proportional allocation.}
  \label{fig:emp-total-regret}
\end{figure}
\subsection{Interpretation}
Taken together, the reconstruction indicates that production ADL removes materially more winner wealth than the instantaneous bankruptcy-gap proxy \(B_t^{\mathrm{needed}}\), even at the conservative horizon \(\Delta=0\).
At the same time, transparent benchmarks (wealth-space pro-rata and solver-free contract-space approximations) can track \(B_t^{\mathrm{needed}}\) closely and keep overshoot vs needed near zero.
This difference is descriptive rather than normative: it measures how a contract-space production policy maps into wealth transfers under partial observability.
Under this measurement, production exhibits both high overshoot regret (\$45.0M vs needed) and high fairness regret (loss concentration on top winners), while benchmark allocations maintain near-zero overshoot regret with bounded fairness regret.

Based on these results, we offer three concrete prescriptions for venues that use ADL:
\begin{enumerate}
\item \emph{Make the numéraire explicit.} When reporting ``haircut dollars,'' state whether the unit is equity dollars (wealth space) or PNL dollars (haircut-capacity space), and report both when possible.
\item \emph{Separate eligibility/capacity from allocation.} Enforce ``profits-only'' capacity constraints (haircuts apply to PNL endowment, not cash) and then choose an allocation rule conditional on that feasible set.
\item \emph{Publish a budget target.} A publicly specified per-wave target such as \(B_t^{\mathrm{needed}}\) makes it possible to evaluate (and improve) execution policies; benchmark allocations show that large overshoot is not mechanically inevitable.
\end{enumerate}

\section{Conclusion and Future Work}\label{sec:conclusion}
In this paper, we provide the first formalism for ADL mechanisms that allows us to compare the performance of different policies.
We first prove a negative result: it is impossible for a perpetual venue to simultaneously optimize for exchange solvency, long-term revenue, and fairness to traders.
We then try to get around the negative result by analyzing strategies that optimize each of these components individually.
Our results demonstrate that the heuristic strategy employed by Hyperliquid, Binance, and others that stems from 2015 is actually suboptimal in all dimensions.
We empirically validate the measurement framework on Hyperliquid data from October 10, 2025: the two-pass reconstruction shows a production profit-haircut overshoot vs the instantaneous bankruptcy-gap proxy \(B_t^{\mathrm{needed}}\) on the order of tens of millions of dollars, while transparent benchmark allocations can target \(B_t^{\mathrm{needed}}\) closely and thereby keep overshoot vs needed near zero.

Our results suggest a few directions for future inquiry.
Implementing these policies in production has a number of challenges that we don't consider in the paper.
For instance, as exchanges grow and have many accounts with very smaller notation sizes, it makes more sense to choose a subset of users who are ``eligible for ADL''.
This lowers the computational burden of executing ADL during times of duress (when execution latency is high) while also ensuring that smaller users are not subsidizing larger users.
Practical limits on ADL applicability likely change the results we have here and need to be studied further.

ADL policies also likely change in their efficacy if they were able to be executed in a fully privately manner.
In this paper, we ignored modeling execution costs for exchanges upon liquidation and position closure.
The modeling of such execution costs is nuanced and likely leads to incorrect conclusions about ADL if not carefully handled\footnote{See, \eg~\citet{Storm2025ADLThread1,Storm2025ADLThread2} for examples of overly naive analyses that conflate execution costs with ADL mechanism design.}.
When private execution is possible, the impact of these costs is muted as follower strategies are not able to react to ADL shocks.

While centralized exchanges effectively offer ADL privacy (\ie~other users do not know who else was autodeleveraged besides themselves), they still leak information to the market via public changes to the order book.
With superior privacy guarantess (likely provided by fully homomorphic encryption, zero knowledge proofs, and secure multi-party computation), one could imagine ADL policies being able to avoid the adverse selection effects we describe in the follower strategies section.
As ADL mechanisms rely increasingly on backstop liquidity vaults (such as Hyperliquid's HLP and Lighter's LLP), privacy of ADL execution becomes increasingly important to reduce costs for the exchange and traders.

Finally, we note that the results in this paper focused on a single margin model, where the posted collateral was the num\'eraire.
Many exchanges, including Hyperliquid (which recently launched portfolio margin~\cite{HyperliquidDocsPortfolioMargin}), have cross-margin ADL support.
This changes the modeling of execution costs and forces us to explicitly model market impact on cash and/or collateral balances.

\section{Acknowledgments}\label{sec:acknowledgments}
The author would like to thank Udai Parvathaneni, Nathan Sheng, JD Maturen, Kamil Yusubov, and Luke Sterle from Gauntlet for helpful discussions around how to quantify realistic ADL scenarios (such as October 10, 2025).
The author would also like to thank Nagu Thogiti, Victor Xu, Yi Sun, Anirudh Pai, Kshitij Kulkarni, Theo Diamandis, Matheus V. X. Ferreira, Guillermo Angeris, and Vinayak Kurup for helpful discussions.
Finally, the author really appreciates the help with procuring and cleaning Hyperliquid data provided by SonarX, Hydromancer, and Mauricio Trujillo (@conejocapital).

\clearpage
\appendix
\section*{Appendix}
\addcontentsline{toc}{section}{Appendix}
\section{Notation and conventions}
Throughout the appendix we use consistent notation:
\begin{itemize}
  \item Shocks are indexed by $t=1,2,\dots,T$; winners $W_t$ with endowments $w_{t,i}>0$ (under PNL-only, $w_{t,i}=(\mathrm{PNL}_{t,i})_+$) and losers $L_t$ with negative equity.
  \item Deficit $D_t\ge 0$ (computed from loser negative equity); severity $\theta_t\ge 0$; haircut vector $h_t\in[0,1]^{|W_t|}$ with endowment survivors $s_{t,i}=(1-h_{t,i})w_{t,i}$.
  \item Per–round budget $B_t=\min\{\theta_t D_t,\ \sum_i w_{t,i}\}$ for severity policies and $B_t=w_t^\top h_t$ for vector policies; $B_t^\star=\min\{D_t,\sum_i w_{t,i}\}$ (budget constraint in endowment space).
  \item Per–account caps are $\beta_{t,i}\in[0,1]$ so that $0\le h_{t,i}\le \beta_{t,i}$.
  \item Risk weights: $\rho(\lambda)=\lambda\,\psi(1/\lambda)$ and $g(\lambda)=\rho(\lambda)/\lambda$; when ordering by risk we use $\rho$.
  \item Orders: $x\prec_w y$ denotes weak submajorization (on decreasing rearrangements).
  \item Asymptotics: $X_n \asympp Y_n$ means there exist constants $c, C > 0$ such that $c |Y_n| \le |X_n| \le C |Y_n|$ with high probability.
\end{itemize}

\section{Liquidations, Autodeleveraging, and Insurance Funds}\label{app:liquidation-mechanics}

\subsection{Bankruptcy Price Example}
We illustrate the bankruptcy price calculation with an example.
Fix $\ell^{\max}=10$ (so $m_I=0.10$) and $p_0=1$.
Consider the five running positions from $\mathcal{P}_5$:
\begin{align*}
\mathfrak p_A&=(q,c,b)=(1,\,2,\,+1), &
\mathfrak p_B&=(1,\,2/3,\,+1), &
\mathfrak p_C&=(4,\,8/3,\,-1),\\
\mathfrak p_D&=(1,\,2/19,\,+1), &
\mathfrak p_E&=(1,\,10/99,\,-1).
\end{align*}
Applying Eq.~\eqref{eq:bankruptcy-price-no-funding} with $\Gamma=0$ and $p_t=p_0=1$ gives
\begin{align*}
 p^{bk}(\mathfrak p_A)&=\max\{0,\,1-2\}=0,\\
 p^{bk}(\mathfrak p_B)&=1-\tfrac{2}{3}=\tfrac{1}{3},\\
 p^{bk}(\mathfrak p_C)&=1-\tfrac{\,8/3\,}{-4}=1+\tfrac{2}{3}=\tfrac{5}{3}\ (1.6667),\\
 p^{bk}(\mathfrak p_D)&=1-\tfrac{2}{19}\approx 0.8947,\\
 p^{bk}(\mathfrak p_E)&=1+\tfrac{10}{99}\approx 1.1010.
\end{align*}
Thus A is robust to a full drop to zero; B (long 1.5x) has a low bankruptcy price; C (short 1.5x) bankrupts only if the mark rises above $\tfrac{5}{3}$; D/E (long ~9.5x/9.9x) have high bankruptcy prices close to $1$, making negative equity likely if liquidations lag.

\subsection{Liquidation Price Example}
For $\mu=0.10$ the liquidation prices evaluate to factors $\tfrac{1}{1-\mu}$ for longs and $\tfrac{1}{1+\mu}$ for shorts. Using the same five positions with $p_{t_i}=1$ and $\Gamma=0$,
\begin{align*}
 \hat p^{liq}(\mathfrak p_A,0.10)&=\tfrac{1}{0.9}\cdot 0=0,\\
 \hat p^{liq}(\mathfrak p_B,0.10)&=\tfrac{1}{0.9}\cdot \tfrac{1}{3}\approx 0.3704,\\
 \hat p^{liq}(\mathfrak p_C,0.10)&=\tfrac{1}{1.1}\cdot \tfrac{5}{3}\approx 1.5152,\\
 \hat p^{liq}(\mathfrak p_D,0.10)&=\tfrac{1}{0.9}\cdot \tfrac{17}{19}\approx 0.9942,\\
 \hat p^{liq}(\mathfrak p_E,0.10)&=\tfrac{1}{1.1}\cdot \tfrac{109}{99}\approx 1.0010.
\end{align*}
Long A would liquidate only if the oracle fell to $0$ (\ie~never occurs in practice). The short $C$ becomes liquidatable only when the oracle exceeds its threshold; the high‑leverage long $D$ and short $E$ become liquidatable close to $1$ if collateral is not topped up.

\subsection{Execution Price Example}
We summarize execution with a directional linear impact rule consistent with our notation: selling to close a long uses $P^{sell}(x)=p_t-\alpha x$ and buying to close a short uses $P^{buy}(x)=p_t+\alpha x$ with $\alpha>0$; the volume–weighted execution for a slice is
$\;p^{exec}=p_t-\tfrac{\alpha}{2}\Delta q$ if $b_i=+1$ (sell) and $p^{exec}=p_t+\tfrac{\alpha}{2}\Delta q$ if $b_i=-1$ (buy).
Directional linear impact with a single $\alpha$: selling (closing a long) uses $P^{sell}(x)=p_t-\alpha x$, buying (closing a short) uses $P^{buy}(x)=p_t+\alpha x$. The slice VWAP over $[0,\Delta q]$ is $p^{exec}=p_t\mp\tfrac{\alpha}{2}\Delta q$ (minus for sells, plus for buys). Fix $\alpha=1.0$ and choose $\Delta q$ per case:
\begin{itemize}[leftmargin=12pt]
  \item $\mathfrak p_A$ (Long): $p_t=1.30$, $\Delta q=0.5$ gives $p^{exec}=1.30-0.5\cdot 0.5=1.05$. Here $p^{bk}(\mathfrak p_A)=\max\{p_t-2,0\}=0$, so $p^{exec}>p^{bk}$.
  \item $\mathfrak p_B$ (Long): $p_t=0.95$, $\Delta q=0.2$ gives $p^{exec}=0.95-0.5\cdot 0.2=0.85$. With $p^{bk}(\mathfrak p_B)=\max\{0.95-\tfrac{2}{3},0\}\approx 0.2833$, we have $p^{exec}>p^{bk}$.
  \item $\mathfrak p_C$ (Short): $p_t=1.60$, $\Delta q=2.0$ gives $p^{exec}=1.60+0.5\cdot 2.0=2.60$. Since $p^{bk}(\mathfrak p_C)=p_t+\tfrac{8/3}{4}=p_t+\tfrac{2}{3}=2.2667$, $p^{exec}>p^{bk}$ (adverse for a short).
  \item $\mathfrak p_D$ (Long; targeted): $p_t=0.98$, $\Delta q=0.4$ gives $p^{exec}=0.98-0.5\cdot 0.4=0.78$. With $p^{bk}(\mathfrak p_D)=p_t-\tfrac{2}{19}\approx 0.8747$, we achieve $p^{exec}<p^{bk}$.
  \item $\mathfrak p_E$ (Short): $p_t=1.05$, $\Delta q=0.4$ gives $p^{exec}=1.05+0.5\cdot 0.4=1.25$. With $p^{bk}(\mathfrak p_E)=p_t+\tfrac{10}{99}\approx 1.151$, we have $p^{exec}>p^{bk}$.
\end{itemize}

\subsection{Liquidation Costs Example}
To ground the fee model, let $\tau(\Delta q)=\tau^{fix}+\phi^{mark} p_t\,\Delta q+\phi^{exec} p^{exec}\,\Delta q$ as in practice.
Consider $p_t=1.30$, a slice $\Delta q=0.50$, and a realized $p^{exec}=1.32$.
Two parameterizations:
\begin{itemize}[leftmargin=12pt]
  \item Binance: $(\tau^{fix},\phi^{mark},\phi^{exec})=(0,\,40\,\mathrm{bps},\,0)$ \citep{BinanceFuturesInsuranceFund}. Then $\tau=0.0040\cdot 1.30\cdot 0.50=0.0026$.
  \item Hyperliquid: $(\tau^{fix},\phi^{mark},\phi^{exec})=(0,\,20\,\mathrm{bps},\,10\,\mathrm{bps})$ \citep{HyperliquidDocsLiquidations}. Then $\tau=0.0020\cdot 1.30\cdot 0.50+0.0010\cdot1.32\cdot0.50\approx0.00130+0.00066=0.00196$.
\end{itemize}
\noindent Rates and formulas vary by venue and contract; the above are illustrative parameterizations consistent with public documentation that liquidation fees are charged and, on centralized venues like Binance, credited to the insurance fund.

\subsection{Liquidation Strategy Example}
Consider short $\mathfrak p_E$ when the mark jumps to $p_t=5.5$ (ignore funding for this step).
Equity before liquidation is $e\approx c_E - q_E(p_t-p_0)=\tfrac{10}{99}-1\cdot 4.5\approx -4.399$.
Let $\mu=0.10$ and a linear fee $\tau(\Delta q)=\phi\,p_t\,\Delta q$ with $\phi=30$bps. Suppose execution is $p^{exec}=5.55$.
Using \eqref{eq:liquidation-strategy} and $b=-1$, the minimal slice that restores maintenance solves
\[
\Delta q\;=\; \frac{\mu p_t q - e}{\,b(p^{exec}-p_t) - \phi p_t + \mu p_t\,}
\;=\; \frac{0.1\cdot 5.5\cdot 1 - (-4.399)}{-0.05 - 0.003\cdot 5.5 + 0.1\cdot 5.5}\ \approx\ \frac{4.949}{0.4835}\ \approx\ 10.24.
\]
Since $\Delta q>q_E$, a greedy policy would fully close E (cap at $\Delta q=q_E=1$).

\subsection{Bad Debt Example}
Consider the high‑leverage long $\mathfrak p_D$ and a slice of size $\Delta q=0.4$ at $p_t=0.98$.
Suppose the realized execution is $p^{exec}_D=0.78$ while the bankruptcy level is $p^{bk}_D\approx 0.8747$.
Since $p^{exec}_D<p^{bk}_D$, the realized shortfall from this slice is
\[
 (p^{bk}_D-p^{exec}_D)\,\Delta q\ \approx\ (0.8747-0.78)\cdot 0.4\ \approx\ 0.0379,
\]
which contributes this amount to the period bad debt $D_t$ (cf. Eq.~\eqref{eq:total-bad-debt}).
Coverage follows the solvency waterfall: the insurance fund pays $\min\{\mathsf{IF}_t, D_t\}$ and any residual shortfall is socialized via ADL (see \S\ref{subsec:exchange-solvency}, \S\ref{subsec:adl}).
The short case is symmetric: buying to close above bankruptcy ($p^{exec}>p^{bk}$) realizes a positive contribution to $D_t$.

\subsection{Anatomy of a Liquidation}
Given the bankruptcy, liquidation, and execution prices, we can now describe the high-level algorithm that liquidations follow.
We note that many live liquidation systems will have much more complex liquidation algorithms.
These complexities deal with the coordination costs of coordinating many parties (\eg~oracle provider, liquidators, spot order book liquidity) and precise models that exchanges use for their liquidation strategy.
However, we effectively lump all of these complexities into the definition of the liquidation strategy.
The following liquidation loop is run on every oracle update received by a perpetuals exchange: 
\begin{itemize}
\item For $\mathfrak{p}_{i,t} \in \mathcal{P}_n$
\begin{itemize}
\item If the maintenance margin condition~\eqref{eq:maintenance-margin} is violated for $\mathfrak{p}_{i,t}$
\begin{enumerate}
    \item Remove the position $\mathcal{P}_n \leftarrow \mathcal{P}_n - \{\mathfrak{p}_{i, t}\}$
    \item Estimate quantity to liquidate $\Delta q_i \leftarrow L(\mathfrak{p}_{i,t}, p_{1:T}, \hat{p}_{1:T})$
    \item Liquidator executes $\Delta q_i$-sized liquidation and returns their execution price $p^{exec}(\Delta q_i)$
    \item Update position: $\mathfrak{p}'_{i,t} = (q_i - \Delta q_i, c_i + p^{exec} \Delta q_i - \tau_t(\Delta q_i), t_i, b_i)$
    \item Re-add the position position: $\mathcal{P}_n \leftarrow \mathcal{P}_n \cup \{\mathfrak{p}'_{i, t}\}$
    \item Update equity using~\eqref{eq:adjusted-equity}
\end{enumerate}
\end{itemize}
\item If $\mathfrak{p}_{i,t}$ has bad debt, $\tilde{e}(\mathfrak{p}_{i,t}, p_{1:T}, \hat{p}_{1:T}, \Delta q_i) < 0$, then
\begin{itemize}
    \item Attempt to use the insurance fund, if it exists, to cover the bad debt (\S\ref{subsec:exchange-solvency})
    \item If the insurance fund is insufficiently sized, utilize an ADL mechanism (\S\ref{subsec:adl})
\end{itemize}
\end{itemize}

\iparagraph{Example.}
We illustrate a five–step path using the running set $\mathcal{P}_5$ from above. Take $T=5$, $p_0=1$ and
\[
 p_{0:5}=(1.00,\ 0.96,\ 0.94,\ 0.97,\ 1.05,\ 1.12),\qquad \hat p_t=p_t\ \ (t=0,\dots,5),\qquad \mu=m_I=0.10.
\]
Executions follow the directional linear impact rule introduced in the execution example: for a slice of size $\Delta q$ at time $t$, the volume‑weighted execution is
\[
\begin{aligned}
  p^{exec}\;=\;p_t-\tfrac{\alpha}{2}\,\Delta q&\quad\text{(sell to close a long, $b=+1$)},\\
  p^{exec}\;=\;p_t+\tfrac{\alpha}{2}\,\Delta q&\quad\text{(buy to close a short, $b=-1$)}.
\end{aligned}
\]
with $\alpha>0$. We take $\alpha=1.0$ and choose $\Delta q$ via the loop's liquidation size $\Delta q_i=L(\mathfrak{p}_{i,t}, p_{1:T}, \hat p_{1:T})$.

\iparagraph{D liquidates at $t=2$ (no bad debt).}
At $t=2$ we have $p_2=0.94$ and the maintenance condition \eqref{eq:maintenance-margin} is violated for $\mathfrak p_{D,2}$, so the loop attempts a partial liquidation. Take $\Delta q_D=L(\mathfrak{p}_{D,2},\cdot)=0.20$ for illustration. By \eqref{eq:bankruptcy-price-no-funding},
\[
 p^{bk}(\mathfrak p_{D,2})\;=\;p_2-\tfrac{2}{19}\;\approx\;0.8347,\qquad
 p^{exec}_D\;=\;p_2-\tfrac{\alpha}{2}\Delta q_D\;=\;0.94-0.10\;=\;0.84.
\]
Since $p^{exec}_D>p^{bk}(\mathfrak p_{D,2})$, this slice executes without bad debt; the position is updated to $\mathfrak{p}'_{D,2}=(q_D-\Delta q_D,\ c_D+p^{exec}_D\Delta q_D-\tau_2(\Delta q_D),\ t_D,\ b_D)$ and equity is updated per \eqref{eq:adjusted-equity} before reinserting $\mathfrak{p}'_{D,2}$ into $\mathcal{P}_n$.

\iparagraph{E becomes bad debt at $t=4$ (short; liquidation fails).}
At $t=4$ we have $p_4=1.05$ and \eqref{eq:maintenance-margin} is violated for $\mathfrak p_{E,4}$ with $b_E=-1$. The loop selects a liquidation size; take a full close $\Delta q_E=L(\mathfrak{p}_{E,4},\cdot)=1$. By \eqref{eq:bankruptcy-price-no-funding},
\[
 p^{bk}(\mathfrak p_{E,4})\;=\;p_4+\tfrac{10}{99}\;\approx\;1.1510,\qquad
 p^{exec}_E\;=\;p_4+\tfrac{\alpha}{2}\Delta q_E\;=\;1.05+0.50\;=\;1.55.
\]
For a short, $p^{exec}>p^{bk}$ realizes bad debt. The loop records the shortfall
\[
 D_4\;=\;\big(p^{exec}_E-p^{bk}(\mathfrak p_{E,4})\big)\,\Delta q_E\;\approx\;0.399,\qquad \tilde e\big(\mathfrak p_{E,4}, p_{1:5}, \hat p_{1:5}, \Delta q_E\big)\;=\;-D_4<0,
\]
and then attempts coverage via the insurance fund (up to $\min\{\mathsf{IF}_4,D_4\}$); any residual shortfall $R_4$ defined by Eq.~\eqref{eq:adl-residual} is socialized by ADL (see \S\ref{subsec:adl}).

\subsection{Optimal Capital Structure Derivation}\label{app:optimal-capital}
In this section, we compute the optimal static insurance fund size $IF^\star$ that trades off the opportunity cost of capital and expected uncovered losses beyond $IF$.

\paragraph{Setup.} Let $D_T$ denote the round deficit with pdf $f_D$ and tail $\bar F_D(x)=\Pr[D_T>x]$. Let $r>0$ be the per‑unit capital cost and $\kappa>0$ the per‑unit social loss weight of uncovered deficits.
The objective is
\[
\min_{IF \ge 0}\ \mathcal{J}(IF)\;=\; r\,IF\;+\;\kappa\,\mathbb{E}\big[(D_T-IF)_+\big]
\]
which equals $r\,IF+\kappa\int_{IF}^{\infty}(x-IF)f_D(x)\,dx$ when $D_T$ is continuous.

\paragraph{Optimality condition.} Differentiating yields,
\[
\mathcal{J}'(IF)= r - \kappa\,\bar F_D(IF).
\]
Hence any interior minimizer satisfies $\bar F_D(IF^\star)=r/\kappa$, i.e.,
\[
IF^\star \;=\; \bar F_D^{-1}\!\Big(\tfrac{r}{\kappa}\Big)\;=\;\mathrm{VaR}_{\,1-r/\kappa}(D_T).
\]

\paragraph{Assumptions.}
In order for this argument to hold, we assume that $\mathcal J$ is convex and differentiable.
Moreover, if $r\ge \kappa$ then we define $IF^\star=0$, whereas if $r\to 0$, then $IF^\star \rightarrow \sup D_T$.
\section{Moral Hazard and Extreme Value Analysis}\label{app:proofs}\label{app:moral-hazard-2}

In this appendix, we formalize the moral hazard properties of ADL mechanisms.
We analyze the optimal control of the Profitability-to-Total-Solvency Ratio (PTSR) and the Profitability-to-Maximum Solvency Ratio (PMR) defined in~\S\ref{sec:risk-metrics}, and derive their asymptotic behavior under distributional assumptions.

\subsection{Setup and Assumptions}

We work in the \emph{large-market limit} ($n\to\infty$) under the heavy-tailed assumptions characteristic of crypto markets.
Recall that $D^\pi_T = \theta_\pi D_T$ is the total socialized loss and $\Delta^\pi_T = \theta_\pi \Delta_T$ is the maximum socialized shortfall under policy $\pi$.
The survivor of the top winner's endowment is denoted $\upsilon^\pi_T = \max_i (w_{i,T} - x_{\pi,i})_+$ (under PNL-only, this is the maximum post-ADL positive PNL).

\iparagraph{Assumption A (regular variation).}
The right tail of the winner endowment distribution $\bar F_+(x)$ (under PNL-only, the distribution of positive PNL) and the right tail of the loser shortfall distribution $\bar F_-(x)$ are regularly varying with indices $\alpha_+ > 0$ and $\alpha_- > 0$, respectively.
That is, $\bar F_\pm(x) = L_\pm(x) x^{-\alpha_\pm}$ where $L_\pm$ are slowly varying functions.

\iparagraph{Assumption B (LLN and EVT).}
We assume the standard Law of Large Numbers (LLN) and Extreme Value Theory (EVT) scaling limits apply:
\begin{itemize}
    \item \emph{Aggregates:} $U_T/n \xrightarrow{p} \mu_+$ (total endowment capacity) and $D_T/n \xrightarrow{p} \mu_-$ (total deficit), for constants $\mu_\pm \in (0, \infty)$.
    \item \emph{Extremes:} The maximum winner endowment $\upsilon_T$ scales as $b_{k_n}^+ = F_+^\leftarrow(1-1/n)$, and the maximum loser shortfall $\Delta_T$ scales as $b_{m_n}^- = F_-^\leftarrow(1-1/n)$.
\end{itemize}

\subsection{Optimal Control of Moral Hazard}

\subsection{Queue maximizes top-winner damage}\label{app:queue-vs-pr-top}

We first establish that the \emph{Queue} (or Top-First) rule minimizes the moral hazard metrics defined in the main text for any fixed budget.

\begin{proposition}[Queue Minimizes Top Survivor]\label{prop:queue-min-top}
    Fix a budget $H = D^\pi_T$. Let $\upsilon_T$ be the endowment of the largest winner (under PNL-only, the maximum positive PNL).
    For any feasible haircut vector $h$ satisfying $\sum_i h_i w_{i,T} = H$ (budget constraint in endowment space), the top-winner endowment survivor $\upsilon^\pi_T$ satisfies
    \[
        \upsilon^\pi_T \;\ge\; \max\{\upsilon_T - H, 0\}.
    \]
    Equality is attained by the Queue rule, which sets the haircut on the top winner to $h_{(1)} = \min(H/\upsilon_T, 1)$ and others to $0$ (until $h_{(1)}$ saturates).
\end{proposition}

\begin{proof}
    Let $h_{(1)}$ be the haircut applied to the top winner's endowment.
    Since $h_i w_{i,T} \ge 0$ for all $i$, we have $h_{(1)} \upsilon_T \le \sum_{i} h_i w_{i,T} = H$.
    The survivor is $\upsilon^\pi_T = \upsilon_T - h_{(1)} \upsilon_T \ge \upsilon_T - H$.
    Since endowment cannot be negative, $\upsilon^\pi_T \ge \max\{\upsilon_T - H, 0\}$.
    The Queue rule greedily allocates the budget to the largest endowment position, achieving $h_{(1)} \upsilon_T = \min(H, \upsilon_T)$, thus attaining the lower bound.
\end{proof}

\begin{corollary}[Minimality of PTSR/PMR]
    Since $D^\pi_T$ and $\Delta^\pi_T$ are fixed for a given policy severity, the Queue rule minimizes both $\mathsf{PTSR}_T$ and $\mathsf{PMR}_T$ among all budget-balanced policies.
\end{corollary}

\iparagraph{Gap versus pro-rata.}
The Queue rule minimizes moral hazard but concentrates the entire loss on the top winner (extreme inequality).
In contrast, the Pro-Rata rule spreads the loss proportionally across all winners, prioritizing \emph{smoothness} (treating similar positions similarly) over minimizing the top survivor's burden.
For $H \le \upsilon_T$, the survivor gap is
\[
    \upsilon^{\mathrm{PR}}_T - \upsilon^{\mathrm{Queue}}_T \;=\; H \left( 1 - \frac{\upsilon_T}{U_T} \right),
\]
where $U_T = \sum_i w_{i,T}$ is the total endowment capacity.
This gap scales linearly with the budget $H$, quantifying the ``cost of fairness'': by choosing the smoother Pro-Rata allocation, the system allows the top winner to retain more profit capacity than is strictly necessary to cover the deficit.

\subsection{Asymptotic Scaling Results}

We now characterize the asymptotic behavior of PTSR and PMR under ``gentle'' policies (like Pro-Rata) where the top winner is not specifically targeted.

\begin{theorem}[PTSR scaling]\label{thm:master-ptsr}\label{thm:ptsr-scaling}\label{thm:ev-main}
    Under Assumptions A and B, for any policy $\pi$ with severity $\theta_n$ where $\upsilon^\pi_T \sim \upsilon_T$ (e.g., Pro-Rata with $H \ll U_T$), the PTSR scales as
    \[
        \mathsf{PTSR}_T(\pi) \;\asympp\; \frac{b_{k_n}^+}{\theta_n n},
    \]
    where $\upsilon^\pi_T$ is the maximum post-ADL endowment.
\end{theorem}

\begin{proof}
    By definition, $\mathsf{PTSR}_T(\pi) = \Expect[\upsilon^\pi_T/D^\pi_T]$.
    Under the hypothesis, the numerator scales as $\upsilon_T \sim b_{k_n}^+$ (the maximum endowment, which under PNL-only is the maximum positive PNL).
    The denominator is $D^\pi_T = \theta_n D_T$. By the LLN, $D_T \sim \mu_- n$, so $D^\pi_T \sim \theta_n \mu_- n$.
    Thus, the ratio scales as $b_{k_n}^+/(\theta_n n)$.
    Using bounded convergence for the expectation yields the result.
\end{proof}

\iparagraph{Implication.}
The behavior of PTSR depends critically on the tail class of winner endowments (under PNL-only, positive PNL):
\begin{itemize}
    \item \emph{Pareto (Heavy) Tails:} $b_{k_n}^+ \asymp n^{1/\alpha_+}$. Here $\mathsf{PTSR}_T \asymp n^{1/\alpha_+ - 1}/\theta_n$. Moral hazard vanishes ($\mathsf{PTSR} \to 0$) if and only if winners have finite mean ($\alpha_+ > 1$). If $\alpha_+ < 1$, the top survivor grows faster than the aggregate deficit, making the moral hazard wedge permanent.
    \item \emph{Exponential/Gaussian (Light) Tails:} $b_{k_n}^+ \asymp (\log n)^\gamma$. Here $\mathsf{PTSR}_T \asymp (\log n)^\gamma / (n \theta_n)$. Since polylog growth is slower than linear, moral hazard vanishes rapidly for any non-vanishing severity $\theta_n$, as the aggregate deficit overwhelms the largest individual winner's endowment.
\end{itemize}

\begin{theorem}[PMR Scaling]\label{thm:pmr-scaling}
    Assume winner endowments scale with leverage mass $\ell_n^+$ and loser deficits scale with leverage mass $\ell_n^-$ (representing total leverage), and that the underlying normalized distributions satisfy Assumption A.
    The PMR scales as:
    \[
        \mathsf{PMR}_T(\pi) \;\asympp\; \frac{1}{\theta_n} \cdot \frac{\ell_n^+}{\ell_n^-} \cdot \frac{b_{k_n}^+}{b_{m_n}^-} \;\asymp\; \frac{1}{\theta_n} \cdot \frac{\ell_n^+}{\ell_n^-} \cdot n^{\frac{1}{\alpha_+} - \frac{1}{\alpha_-}},
    \]
    where $\upsilon^\pi_T$ is the maximum post-ADL endowment.
\end{theorem}

\begin{proof}
    We have $\mathsf{PMR}_T(\pi) = \Expect[\upsilon^\pi_T / \Delta^\pi_T]$.
    The top winner's endowment scales with total winner leverage mass: $\upsilon_T \sim \ell_n^+ b_{k_n}^+$ (under PNL-only, this is the maximum positive PNL).
    The maximum loser shortfall scales with total loser leverage mass: $\Delta_T \sim \ell_n^- b_{m_n}^-$.
    The budget balance condition implies $\Delta^\pi_T = \theta_n \Delta_T$.
    Thus, the ratio scales as
    \[
        \frac{\ell_n^+ b_{k_n}^+}{\theta_n \ell_n^- b_{m_n}^-} \;=\; \frac{1}{\theta_n} \frac{\ell_n^+}{\ell_n^-} \frac{b_{k_n}^+}{b_{m_n}^-}.
    \]
    Substituting the regular variation scalings $b_{k_n}^+ \sim n^{1/\alpha_+}$ and $b_{m_n}^- \sim n^{1/\alpha_-}$ yields the result.
\end{proof}

Theorem~\ref{thm:pmr-scaling} decomposes moral hazard into three components:
(1) \emph{Policy Severity} ($1/\theta_n$): Lower severity amplifies PMR.
(2) \emph{Leverage Imbalance} ($\ell_n^+/\ell_n^-$): If the winning side holds more leverage mass, PMR increases.
(3) \emph{Tail Risk} ($n^{1/\alpha_+ - 1/\alpha_-}$): Heavier winner tails relative to losers drive PMR divergence.
This decomposition highlights that even with fair tails ($\alpha_+ = \alpha_-$), a systemic leverage imbalance ($\ell_n^+ \gg \ell_n^-$) can sustain a high PMR.
Specifically, if the exchange allows winners to be significantly more leveraged than losers (a "risk-on" imbalance), the top winner's survival will systematically outstrip the worst-case socialized loss, creating a persistent moral hazard where maximal profits are privatized while maximal losses are capped.

\subsection{Relationship to Classical Risk Measures}\label{app:classical-risk-measures}

These two metrics have natural interpretations in terms of financial risk measures.
The deficit $D_T$ corresponds to the aggregate \emph{Expected Shortfall} (ES) of the losing tail, while $\Delta_T$ corresponds to the \emph{Value-at-Risk} (VaR) at the extreme quantile ($1/n$).
Specifically, PTSR compares the \emph{Maximum Profit} to the \emph{Aggregate Socialized Loss} (ES-like), measuring efficiency in bulk.
PMR compares the \emph{Maximum Profit} to the \emph{Maximum Socialized Loss} (VaR-like), measuring efficiency in the extreme tail.
A high PMR implies that the system permits ``unicorn'' wins that vastly exceed the worst-case individual losses, potentially incentivizing excessive risk-taking if traders perceive a capped downside but unbounded upside.

We further strengthen the connection to classical risk measures by showing that Queue not only minimizes the top survivor in expectation, but also minimizes it in the sense of VaR and ES at \emph{every} tail level.
\begin{proposition}[Queue minimizes VaR/ES of the top survivor]\label{thm:queue-var-es}
  Fix any budget $h\ge 0$ and $\alpha\in(0,1)$.
  For any feasible haircut vector $h$ with $\sum_i h_i e_i = h$,
  \[
    \omega^\pi_T \;\ge\; (\omega_T - h)_+ \quad\text{a.s.}
  \]
  Consequently,
  \[
    \mathrm{VaR}_\alpha(\omega^\pi_T) \;\ge\; \mathrm{VaR}_\alpha\big((\omega_T - h)_+\big),\qquad
    \mathrm{ES}_\alpha(\omega^\pi_T) \;\ge\; \mathrm{ES}_\alpha\big((\omega_T - h)_+\big).
  \]
  The Queue rule attains equality. Moreover, the following identities hold:
  \[
    \mathrm{VaR}_\alpha\big((\omega_T - h)_+\big) \;=\; \max\{\mathrm{VaR}_\alpha(\omega_T) - h,\, 0\},
  \]
  \[
    \mathrm{ES}_\alpha\big((\omega_T - h)_+\big) \;=\; \frac{1}{1-\alpha}\int_\alpha^1 \max\{\mathrm{VaR}_u(\omega_T) - h,\, 0\}\,du.
  \]
\end{proposition}

\begin{proof}
  The pointwise lower bound $\omega^\pi_T \ge (\omega_T - h)_+$ follows from the budget constraint and nonnegativity of haircuts, as in Proposition~\ref{prop:queue-min-top}.
  Monotonicity of risk measures implies that if $X\ge Y$ almost surely, then $\mathrm{VaR}_\alpha(X)\ge \mathrm{VaR}_\alpha(Y)$ and $\mathrm{ES}_\alpha(X)\ge \mathrm{ES}_\alpha(Y)$.
  For the identities, observe that $x\mapsto (x-h)_+$ is nondecreasing; hence quantiles shift: $\mathrm{VaR}_\alpha((X-h)_+) = \max\{\mathrm{VaR}_\alpha(X)-h,0\}$.
  The ES identity follows from the Kusuoka representation~\cite{Kusuoka2001} $\mathrm{ES}_\alpha(Z)=\frac{1}{1-\alpha}\int_\alpha^1 \mathrm{VaR}_u(Z)\,du$ applied to $Z=(X-h)_+$.
\end{proof}

\iparagraph{Implication for severity design.}
For a random budget $H=\theta_n D_T$, apply Theorem~\ref{thm:queue-var-es} conditionally on $H$ to conclude that Queue minimizes the conditional VaR/ES of the top survivor at every tail level.
When $\alpha_+>1$ (finite mean winners), the tail-equivalence property of regularly varying distributions yields
\[
  \frac{\mathrm{ES}_u(\omega_T)}{\mathrm{VaR}_u(\omega_T)} \;\to\; \frac{\alpha_+}{\alpha_+-1}\qquad\text{as }u\uparrow 1,
\]
so VaR- and ES-based moral hazard conclusions coincide asymptotically with those of PTSR and PMR.

\subsection{Randomized constructions for moral-hazard examples}\label{app:mh-example}

\iparagraph{Extreme-value moral hazard (principal–agent).}
Fix $\rho\in(0,1)$ and $k_n=\lfloor \rho n\rfloor$. Draw winner equities $Y_i^{(n)}$ i.i.d.\ Pareto$(\alpha_+)$ and loser equities $X_i^{(n)}$ with mean $\mu$.
Then $M_n^+=\max Y_i^{(n)} \asymp n^{1/\alpha_+}$ while non-max winners sum to $o(M_n^+)$.
Losers sum to $D_T \approx \mu n$.
For fixed severity $\theta_n \equiv \bar\theta$, the haircut $H_n \approx \bar\theta \mu n$ exceeds the capacity of non-max winners, forcing the top winner to cover the bulk.
Post-ADL equity is $(M_n^+ - \bar\theta \mu n)_+ \to 0$ since $n^{1/\alpha_+} \ll n$ for $\alpha_+ < 1$.

\iparagraph{Leverage-imbalance construction.}\label{app:mh-leverage}
Fix leverage masses $\ell_n^- \gg \ell_n^+$. Draw loser equities $X_i^{(n)}$ i.i.d.\ Pareto$(\alpha_-)$, so $D_T \approx M_n^- \asymp n^{1/\alpha_-}$.
Assign winner leverage $c\,\ell_n^+$ to a random index $I_n$ and distribute the rest evenly.
Then $\omega_n \asymp (\ell_n^+/\ell_n^-) n$.
This satisfies Proposition~\ref{prop:excessive-leverage} assumptions, yielding the claimed threshold.

\section{Theoretical Properties of Capped Pro-Rata}\label{app:capped-pro-rata}
We formalize the theoretical properties of the capped pro-rata rule.
We note that the capped pro-rata algorithm in Algorithm~\ref{alg:capped-pro-rata} is a standard water-filling algorithm~\cite{BoydVandenberghe2004}.
The most similar known prior work to this appendix is the study of how such algorithms provide sybil resistance in concave games in decentralized systems~\cite{johnson2023concave}.
One can view our result as a generalization of this result.

\iparagraph{Properties of ADL rules.}
Fix time $T$, state $\mathcal{P}_n$, winners $\mathcal{W}_T$, and endowments $w_i = w_{T,i}$ (under PNL-only, $w_i = (\mathrm{PNL}_{T,i})_+$).
Let $s_{\pi, i} = (1-h_{\pi, i}) w_i$ be the surviving endowment of position $i$ under ADL policy $\pi$.
The post-ADL equity is $e'_{\pi,i} = c_{i,T} + \mathrm{PNL}_{i,T} - h_{\pi,i} w_i = c_{i,T} + (1-h_{\pi,i})w_i + (\mathrm{PNL}_{i,T})_-$ (where $(\mathrm{PNL}_{i,T})_-$ is the negative PNL component, if any).
We say that a feasible ADL rule $\pi$ (\S\ref{subsec:adl}) satisfies:
\begin{enumerate}
\item \emph{Sybil resistance:} Outcomes are invariant to splitting/merging accounts. For any split $w_i=\sum_{a=1}^r z_a$, the sum of survivors equals the original survivor: $\sum_{a} s^{(i\to z)}_{\pi,a} = s_{\pi,i}$.
\item \emph{Scale invariance:} $s_\pi(c w;\theta_\pi)=c\,s_\pi(w;\theta_\pi)$ for $c>0$.
\item \emph{Monotonicity:} If $w_1 \ge \dots \ge w_k$, then $s_{\pi,1}\ge\dots\ge s_{\pi,k}$ (and consequently, if cash components are comparable, post-ADL equities preserve ordering).
\item \emph{Interior regularity:} The map $(w,H)\mapsto s_\pi(w;H)$ is $C^1$ on the interior, \ie~for $w_i>0$ for all $i$ and $0<H<\sum_i w_i$.
\end{enumerate}
Collectively, we refer to these as the \emph{fairness properties} for an ADL rule.

\paragraph{Queue score dependence: Sybil resistance vs monotonicity.}
The queue rule's Sybil behavior depends on the ranking score, while monotonicity can fail even when Sybil resistance holds.

\begin{proposition}[Queue Sybil resistance for score-preserving splits]\label{prop:queue-sybil-score-preserving}
Consider the wealth-space queue rule $\pi_Q$ with budget $B$, winner endowments $w$, and ranking scores $r_i$.
Fix winner $i$ with a unique score level ($r_i\neq r_j$ for all $j\neq i$).
Suppose $i$ is split into subaccounts with endowments $z_1,\dots,z_m>0$, $\sum_{a=1}^m z_a=w_i$, and each child keeps the parent score ($r_{i,a}=r_i$). Write $x_{\pi_Q,i}(B):=w_i-s_{\pi_Q,i}(B)$ for seizure at budget $B$, and $\tilde x_{\pi_Q,a}(B)$ for seizure on split child $a$.
Then aggregate seizure is invariant:
\[
  \sum_{a=1}^m \tilde x_{\pi_Q,a}(B)=x_{\pi_Q,i}(B).
\]
Hence queue is Sybil resistant for such splits.
\end{proposition}

\begin{proof}
Let $H_i=\{j:\ r_j>r_i\}$ and define the residual budget when the queue reaches score level $r_i$:
\[
  R_i=\Big(B-\sum_{j\in H_i} w_j\Big)_+.
\]
Because $r_i$ is unique, no other account is processed between $H_i$ and account $i$.
In the unsplit instance, queue seizure is
\[
  x_{\pi_Q,i}(B)=\min\{w_i,\ R_i\}.
\]
In the split instance, the same higher-score set $H_i$ is processed first, so the residual is still $R_i$ when the split children are reached.
The queue then seizes greedily across the children, whose total endowment is $\sum_a z_a=w_i$, hence
\[
  \sum_{a=1}^m \tilde x_{\pi_Q,a}(B)=\min\Big\{\sum_{a=1}^m z_a,\ R_i\Big\}=\min\{w_i,\ R_i\}=x_{\pi_Q,i}(B).
\]
\end{proof}

\begin{corollary}[Hyperliquid score is Sybil resistant under proportional splits]\label{cor:hyperliquid-sybil}
Hyperliquid's documented ADL score~\cite{HyperliquidDocsADL},
\[
  r=\frac{\mathrm{mark}}{\mathrm{entry}}\cdot\frac{\mathrm{notional}}{\mathrm{account\ value}},
\]
is homogeneous of degree zero in $(\mathrm{notional},\mathrm{account\ value})$ under proportional scaling: for any $\alpha>0$,
\[
  r(\alpha\,\mathrm{notional},\alpha\,\mathrm{account\ value})=r(\mathrm{notional},\mathrm{account\ value}).
\]
Therefore proportional account splitting preserves score and, under Proposition~\ref{prop:queue-sybil-score-preserving}'s no-tie assumption at the split score level, preserves aggregate queue seizure (Sybil resistance). See also the worked example in~\cite{Doug2025SybilTweet}.
\end{corollary}

\begin{proof}
The ratio $\mathrm{notional}/\mathrm{account\ value}$ is unchanged by common scaling, and the mark/entry term is unchanged by splitting, so each split child has the same score as the parent.
Apply Proposition~\ref{prop:queue-sybil-score-preserving}.
\end{proof}

\begin{proposition}[Queue monotonicity failure]\label{prop:queue-monotonicity-failure}
There exist two winners with $w_1>w_2>0$, queue scores $r_1>r_2$, and budget $B\in(w_1-w_2,w_1]$ such that queue violates monotonicity:
\[
  s'_{1}<s'_{2},
\]
where $s'_i=w_i-x_i$ is post-ADL surviving endowment.
\end{proposition}

\begin{proof}
With $r_1>r_2$, queue processes winner $1$ first.
For $B\le w_1$, queue seizes $x_1=B$ and $x_2=0$.
Hence
\[
  s'_1=w_1-B,\qquad s'_2=w_2.
\]
If $B>w_1-w_2$, then $w_1-B<w_2$, so $s'_1<s'_2$ despite $w_1>w_2$ pre-ADL.
\end{proof}

\paragraph{Absolute-score vulnerability.}
If the score is an absolute quantity (\eg~$r_i=\mathrm{PNL}_i$), splitting changes child scores and can move a position deeper in the queue.
This is the standard Sybil-vulnerable queue setting.

\begin{proposition}[Uniqueness of the Pro-Rata Rule]\label{thm:unique-pro-rata}
If a feasible ADL policy $\pi$ satisfies the fairness properties, then $s_\pi(w;H)=s^{\mathrm{PR}}(w;H)$ for all feasible inputs, where $s$ is the surviving endowment.
\end{proposition}

\begin{proof}
Fix a feasible budget $H$ (\eg~maximum value of $\theta_n D_T$) and write $\beta_i\in(0,1]$ for the haircut cap on winner $i$ (as in the capped pro-rata rule~\eqref{eq:capped-pro-rata}), and sort so that $w_1\ge\cdots\ge w_k$.
View the haircuts as a function of the realized budget $b\in[0,H]$ and write $h_i(b)$ for the haircut on winner $i$ when total budget $b$ has been allocated.
For each $b$, define the \emph{active set}
\begin{equation}\label{eq:active-set}
A(b)=\{i:\ h_i(b)<\beta_i\}
\end{equation}
of winners whose caps are not yet binding.
Since there are only finitely many caps $\beta_i$, there exists a partition $0=b_0<b_1<\cdots<b_L=H$ such that on each open interval $I_\ell:=(b_{\ell-1},b_\ell)$ the active set is constant.
Fix one such interval $I=(b_-,b_+)$ and write $A = A(b)$ for any $b\in I$.
On this interval, by scale invariance on $A$ (per–unit budget increases all active haircuts at the same rate) and feasibility (budget constraint in endowment space),
\[
  \frac{d h_i}{d b} \;=\; \frac{1}{\sum_{j\in A} w_j}\quad (i\in A),
  \qquad
  \frac{d h_i}{d b}=0\quad (i\notin A).
\]
Using the interior regularity property, we can integrate these terms on $I$.
Integrating from $b_-$ to any $b\in I$ gives the \emph{unconstrained} evolution
\[
  \tilde h_i(b)
  \;=\;
  \begin{cases}
    h_i(b_-)\;+\;\dfrac{b-b_-}{\sum_{j\in A} w_j} & (i\in A),\\[4pt]
    h_i(b_-) & (i\notin A).
  \end{cases}
\]
All active coordinates in $A$ move in lockstep, so on $I$ there exists a scalar function $\eta_I(b)$ such that
\[
  h_i(b) \;=\; \min\{\eta_I(b),\beta_i\}\quad (i\in A),\qquad
  h_i(b)=h_i(b_-) \quad (i\notin A).
\]
The budget identity can then be written on $I$ as
\[
  b \;=\; \sum_{i\notin A} w_i h_i(b_-)\;+\;\sum_{i\in A} w_i \min\{\eta_I(b),\beta_i\},
\]
which, for fixed $b\in I$, is a continuous strictly increasing function of $\eta_I(b)$ as long as $A\neq\emptyset$.
Thus for each $b\in I$ there is a unique $\eta_I(b)$ solving the budget identity.
In particular, on the \emph{interior} interval where no caps bind ($A=\{1,\dots,k\}$ and $b\in(0,U_T)$ where $U_T=\sum_i w_i$), we have $h_i(b)=\eta_I(b)$ for all $i$ so the budget identity reduces to
\[
  b \;=\; \sum_{i=1}^k w_i h_i(b) \;=\; \eta_I(b)\sum_{i=1}^k w_i \;=\; \eta_I(b)\,U_T,
\]
which implies $\eta_I(b)=b/U_T$.
Thus on this interval
$s_{\pi,i}=(1-\eta_I(b))w_i=(1-b/U_T)w_i$, i.e., proportional to endowment.

Sybil resistance implies allocations depend only on total endowment: splitting $w_i=\sum_a z_a$ leaves
$\sum_a z_a h_a(b)$ and hence the survivor $\sum_a (1-h_a(b)) z_a$ unchanged for each $b$, so the proportional form on interior intervals is preserved under arbitrary splits.
Monotonicity further restricts how indices can exit the active set as $b$ increases.
At an endpoint $b_\ell$ where $\eta_I(b_\ell)$ first hits some cap $\beta_m$, all indices $j\ge m$ with $\beta_j\le \beta_m$ must saturate together; otherwise we would have $w_m\ge w_j$ but $(1-h_m(b_\ell))w_m < (1-h_j(b_\ell))w_j$, violating monotonicity.
Thus, as we pass from $I_\ell$ to $I_{\ell+1}$, a (possibly empty) tail $\{j>m\}$ leaves $A$, contributing a fixed amount $\sum_{j>m} w_j\beta_j$ to the budget, and the same water-filling argument applies on the remaining head with reduced budget.

Concatenating the solution across all intervals $I_\ell$ yields the reverse–waterfilling form
$h_i=\min\{\eta^\star,\beta_i\}$ with $\eta^\star$ chosen so that $\sum_i w_i h_i=H$, which is exactly capped pro–rata.
Uniqueness follows either from the strict convexity of the Euclidean projection onto
$\{h\in[0,1]^k:\ \sum_i w_i h_i=H\}$ or from the monotone one–dimensional search that defines $\eta^\star$.
\end{proof}

\iparagraph{Convex optimality}\label{app:convex-optimality}
We now formalize the convex-welfare interpretation of capped pro-rata from \S\ref{sec:fairness}.
Fix time $T$, winners $\mathcal{W}_T$ with endowments $w_i = w_{T,i}$ (under PNL-only, $w_i = (\mathrm{PNL}_{T,i})_+$), and effective caps $\beta_i = \min\{\overline{h}_i, 1-(\underline{e}_i - c_{i,T})/w_{i,T}\}$ defined by the haircut and equity constraints~\eqref{eq:haircut-constraint}--\eqref{eq:equity-constraint}, where the equity constraint is reinterpreted in terms of minimum post-ADL endowment.
Let $B_T = \theta_{\pi} D_T(\mathcal{P}_n)$ denote the haircut budget from~\eqref{eq:eta}, and let $\phi:[0,1]\to\reals$ be a strictly convex increasing function representing per-unit haircut disutility as in~\eqref{eq:haircut-optimization}.
We consider choosing haircuts $h=(h_i)_{i\in\mathcal{W}_T}$ to minimize the endowment-weighted total disutility $\sum_{i\in\mathcal{W}_T} w_i \phi(h_i)$ subject to the per-account bounds $0\le h_i\le\beta_i$ and the aggregate budget constraint $\sum_{i\in\mathcal{W}_T} w_i h_i = B_T$ (budget constraint in endowment space).

\begin{proposition}[Convex optimality]\label{prop:capped-prorata-optimal}
    For any strictly convex increasing $\phi$, the unique solution to
    \[
      \min_{h}\ \sum_{i\in\mathcal{W}_T} w_i\,\phi(h_i)
      \quad\text{s.t.}\quad
      \sum_{i\in\mathcal{W}_T} w_i h_i = B_T,\quad 0\le h_i\le \beta_i
    \]
    is the capped pro-rata rule~\eqref{eq:capped-pro-rata}, \ie
    \[
      h_{\pi_{CP},i}^\star = \min\{\eta^\star, \beta_i\},
    \]
    where $\eta^\star$ satisfies $\sum_{i\in\mathcal{W}_T} w_i h_{\pi_{CP},i}^\star = B_T$.
\end{proposition}

\begin{proof}
    The optimization problem is convex with a strictly convex objective and linear constraints, so any point satisfying the Karush--Kuhn--Tucker (KKT) conditions is the unique global minimizer~\cite[Ch.~5]{BoydVandenberghe2004}.
    For $B_T$ in the interior of the feasible region ($0<B_T<C(\beta)$, where $C(\beta)=\sum_i w_i \beta_i$ is defined in~\eqref{eq:max-cap}), Slater's condition holds, so KKT conditions are necessary and sufficient.
    The Lagrangian is
    \[
      \mathcal{L}(h,\lambda,\mu,\nu)
      \;=\;
      \sum_{i\in\mathcal{W}_T} w_i \phi(h_i)
      + \lambda\!\left(\sum_{i\in\mathcal{W}_T} w_i h_i - B_T\right)
      + \sum_{i\in\mathcal{W}_T} \mu_i(h_i-\beta_i)
      - \sum_{i\in\mathcal{W}_T} \nu_i h_i,
    \]
    with multipliers $\lambda\in\reals$ and $\mu_i,\nu_i\ge 0$.
    Stationarity with respect to $h_i$ gives
    \[
      w_i \phi'(h_i) + \lambda w_i + \mu_i - \nu_i = 0 \quad (i\in\mathcal{W}_T),
    \]
    together with complementary slackness $\mu_i(h_i-\beta_i)=0$ and $\nu_i h_i=0$.
    For any index $i$ with $0<h_i<\beta_i$, we must have $\mu_i=\nu_i=0$, so $\phi'(h_i)=-\lambda$.
    Because $\phi'$ is strictly increasing, this implies $h_i=c$ for some common scalar $c$ independent of $i$.
    If $h_i=\beta_i$ then $\mu_i\ge 0$ and $\nu_i=0$, and if $h_i=0$ then $\nu_i\ge 0$ and $\mu_i=0$, so in all cases the KKT conditions imply the water-filling form
    \[
      h_i = \min\{c,\beta_i\}\quad (i\in\mathcal{W}_T).
    \]
    The budget constraint $\sum_{i\in\mathcal{W}_T} w_i h_i = B_T$ then reduces to finding $c$ such that
    \[
      \sum_{i\in\mathcal{W}_T} w_i \min\{c,\beta_i\} = B_T.
    \]
    The left-hand side is continuous and strictly increasing in $c$ on $[0,1]$ as long as some $\beta_i>0$, so there is a unique $c=\eta^\star$ solving this equation.
    Thus every KKT point has the capped pro-rata form $h_i=\min\{\eta^\star,\beta_i\}$ with $\eta^\star$ determined by the budget, which is exactly the rule in~\eqref{eq:capped-pro-rata}.
    Strict convexity of the objective implies this KKT point is the unique global minimizer, proving the claim (see also the standard water-filling derivations in~\cite[Ch.~5]{BoydVandenberghe2004}).
\end{proof}

\section{Algorithms for Pro-Rata Haircut Rules}\label{app:pro-rata-algorithms}

This appendix collects the explicit water-filling procedures for both the capped pro-rata rule from~\cref{subsec:fairness} and the Risk-Aware Pro-Rata (RAP) rule from~\cref{sec:glpr}.
Each algorithm takes the winner endowments $\{w_i\}$ (under PNL-only, positive PNL), effective caps $\{\beta_i\}$, and the target budget $B_T = \theta_{\pi} D_T(\mathcal{P}_n)$, returning the optimal haircut vector $h$ or declaring infeasibility if the aggregate capacity $\sum_i w_i \beta_i$ is insufficient.

\subsection{Capped Pro-Rata Water-Filling}

The procedure below enforces the constraints in \cref{eq:capped-pro-rata,eq:eta} by leveling caps until the target budget is met.

\begin{algorithm}
\caption{Capped Pro-Rata Haircut Allocation (Water-Filling)}
\label{alg:capped-pro-rata}
\begin{algorithmic}[1]
\Require Winner endowments $w = \{w_1, \dots, w_n\}$ (under PNL-only, positive PNL), Effective caps $\beta = \{\beta_1, \dots, \beta_n\}$, Target budget $B_T$
\Ensure Haircut vector $h = \{h_1, \dots, h_n\}$ or \textbf{Infeasible}

\State $C \leftarrow \sum_{i=1}^n w_i \beta_i$ \Comment{Compute total maximum capacity}
\If{$B_T > C$}
    \State \Return \textbf{Infeasible}
\EndIf

\State Sort indices $p$ such that $\beta_{p_1} \le \beta_{p_2} \le \dots \le \beta_{p_n}$
\State $\beta_{p_0} \leftarrow 0$
\State $V \leftarrow 0$ \Comment{Cumulative value covered}
\State $R \leftarrow \sum_{i=1}^n w_i$ \Comment{Remaining uncapped endowment mass}

\For{$k = 1$ to $n$}
    \State $\Delta \beta \leftarrow \beta_{p_k} - \beta_{p_{k-1}}$
    \State $V_{\mathrm{step}} \leftarrow \Delta \beta \cdot R$
    
    \If{$V + V_{\mathrm{step}} \ge B_T$}
        \State $\eta \leftarrow \beta_{p_{k-1}} + (B_T - V) / R$ \Comment{Found the water level $\eta$}
        \State \textbf{break}
    \EndIf
    
    \State $V \leftarrow V + V_{\mathrm{step}}$
    \State $R \leftarrow R - w_{p_k}$ \Comment{User $p_k$ becomes fully capped}
\EndFor

\If{$V < B_T$} \Comment{Handling numerical edge cases}
    \State $\eta \leftarrow \beta_{p_n}$
\EndIf

\For{$i = 1$ to $n$}
    \State $h_i \leftarrow \min\{\eta, \beta_i\}$
\EndFor
\State \Return $h$
\end{algorithmic}
\end{algorithm}

\subsection{Risk-Aware Pro-Rata Water-Filling}

The RAP algorithm augments the capped procedure by prioritizing accounts according to their ``cap-to-weight'' ratios $\beta_i / \tilde{w}_i$, where $\tilde{w}_i$ are risk weights (not to be confused with the endowment $w_i$).

\begin{algorithm}
\caption{Risk-Aware Pro-Rata Haircut Allocation (Weighted Water-Filling)}
\label{alg:rap-water-filling}
\begin{algorithmic}[1]
\Require Winner endowments $w = \{w_1, \dots, w_n\}$ (under PNL-only, positive PNL), Effective caps $\beta = \{\beta_1, \dots, \beta_n\}$, Risk weights $\tilde{w} = \{\tilde{w}_1, \dots, \tilde{w}_n\}$, Target budget $B_T$
\Ensure Haircut vector $h = \{h_1, \dots, h_n\}$ or \textbf{Infeasible}

\State $C \leftarrow \sum_{i=1}^n w_i \beta_i$ \Comment{Total capacity (budget constraint in endowment space)}
\If{$B_T > C$}
    \State \Return \textbf{Infeasible}
\EndIf

\State Compute ratios $r_i \leftarrow \beta_i / \tilde{w}_i$ for all $i$ (treat $\tilde{w}_i=0$ as $r_i=\infty$)
\State Sort indices $p$ such that $r_{p_1} \le r_{p_2} \le \dots \le r_{p_n}$
\State $r_{p_0} \leftarrow 0$
\State $V \leftarrow 0$ \Comment{Cumulative value covered}
\State $W_{\mathrm{rem}} \leftarrow \sum_{i=1}^n w_i \tilde{w}_i$ \Comment{Remaining weighted endowment mass}

\For{$k = 1$ to $n$}
    \State $\Delta \tau \leftarrow r_{p_k} - r_{p_{k-1}}$
    \State $V_{\mathrm{step}} \leftarrow \Delta \tau \cdot W_{\mathrm{rem}}$
    
    \If{$V + V_{\mathrm{step}} \ge B_T$}
        \State $\tau \leftarrow r_{p_{k-1}} + (B_T - V) / W_{\mathrm{rem}}$ \Comment{Found the scaling factor $\tau$}
        \State \textbf{break}
    \EndIf
    
    \State $V \leftarrow V + V_{\mathrm{step}}$
    \State $W_{\mathrm{rem}} \leftarrow W_{\mathrm{rem}} - w_{p_k} \tilde{w}_{p_k}$ \Comment{User $p_k$ becomes fully capped}
\EndFor

\If{$V < B_T$} \Comment{Numerical edge case}
    \State $\tau \leftarrow r_{p_n}$
\EndIf

\For{$i = 1$ to $n$}
    \State $h_i \leftarrow \min\{\beta_i, \tau \tilde{w}_i\}$
\EndFor
\State \Return $h$
\end{algorithmic}
\end{algorithm}

\section{Risk-Aware Pro-Rata (RAP)}\label{app:rap}

\subsection{Examples of Risk-Aware Pro-Rata (RAP) and Next Deficit}\label{app:rap-examples}

In this section, we provide detailed numerical examples illustrating the properties of the RAP rule and the impact of post-haircut shocks.

\iparagraph{RAP weighting example.}
We illustrate the three choices of risk models on an example at $T=2$ where the winners are $\mathcal W_2=\{A,C,E\}$ with effective leverages $\lambda^+_{A,2}\approx 1.031$, $\lambda^+_{C,2}\approx 0.925$, and $\lambda^+_{E,2}\approx 1.548$.
Recall that under pro-rata, the normalized shares are $s^{\mathrm{PR}}_i \propto e_{T,i}$, yielding $(s^{\mathrm{PR}}_A, s^{\mathrm{PR}}_C, s^{\mathrm{PR}}_E) \approx (0.163, 0.728, 0.109)$.
For RAP with $w_i=\lambda_i g(\lambda_i)$, the shares allocate proportional to $e_{T,i}w_i$.
The resulting shares (order $A,C,E$) are:
\begin{itemize}
    \item \textbf{Linear} $g(\lambda)=\lambda$: $\approx (0.164, 0.589, 0.246)$.
    \item \textbf{Power} $g(\lambda)=\lambda^2$: $\approx (0.155, 0.498, 0.348)$.
    \item \textbf{CVaR} $g(\lambda)=(\lambda-0.9)_+$: $\approx (0.149, 0.114, 0.737)$.
\end{itemize}
RAP shifts haircut mass toward high-leverage winners; the tilt is mild for linear $g$, stronger for $\lambda^2$, and concentrates almost entirely on the over-threshold tail for the CVaR model.

\iparagraph{Next deficit and leverage sensitivity.}
Consider the setup where $T=2$ with $D_T \approx 0.705$ and $W_T \approx 7.72$.
Under the normal pro-rata rule, the haircut rate is $h^{\mathrm{PR}}_T \approx 0.0913$.
We consider a simple Markovian shock whose direction is uniformly random and whose magnitude grows with the winner leverage mass:
\[
Z_{T,i} = \xi_T \zeta_T, \quad \xi_T \in \{-1,+1\} \text{ equiprobable}, \quad \zeta_T = \alpha \frac{\ell^+_T}{k_T}.
\]
With $\alpha=1.2$, the expected next deficit is:
\[
\Expect[D^{\mathrm{next}}_{T+1} \mid \mathcal F_T] \approx 1.46 > D_T.
\]
This illustrates a failure mode for pro-rata when the shock kernel scales strongly with leverage: pro-rata shrinks all winners uniformly and leaves effective leverages $\lambda_{T,i}$ unchanged, so the shock magnitude $\zeta_T$ is unaffected while residual exposure remains large on high-leverage winners.

\iparagraph{Correlated shocks example.}
Consider two winners with equal equity $e$ and leverage levels $\lambda_{t,1}>\lambda_{t,2}\ge 1$, and budget $b_t=2\varepsilon e$.
Assume price shocks are AR(1): $Z_{t+2}=\rho Z_{t+1}+\varepsilon_{t+2}$, $\rho\in(0,1)$.
RAP with $w_{t,i} \propto \lambda_{t,i}\psi(1/\lambda_{t,i})$ puts more haircut mass on account 1 (the higher leverage account).
Assume account 1's exposure to $Z_{t+1}$ offsets the loss term (a "hedge") in the next step deficit: $D^{\text{next}}_{t+1}=(\alpha-\beta s_{t,1})Z_{t+1}$.
Shrinking $s_{t,1}$ (winner 1's residual equity) weakens the hedge into $t+2$.
Since $s^{\mathrm{RAP}}_{t,1} < s^{\mathrm{PR}}_{t,1}$, the two-step sum of deficits $S_t$ can satisfy $S_t(h^{\mathrm{PR}}_t) < S_t(h^{\mathrm{RAP}}_t)$, even if RAP minimizes the one-step deficit.

\iparagraph{Exchange incentive compatibility example.}
Consider two positions $\mathfrak{p}_A$ (high value, $\theta_A=100$) and $\mathfrak{p}_B$ (low value, $\theta_B=5$) with equal initial leverage $\ell=1$.
Suppose the exchange must reduce total leverage by 1 unit.
\begin{itemize}
    \item \textbf{RAP (Targeted):} Fully liquidates $\mathfrak{p}_A$. Continuation value: $5\beta$.
    \item \textbf{Pro-Rata:} Reduces both by 50\%. Continuation value: $52.5\beta$.
\end{itemize}
Pro-rata yields significantly higher utility by preserving the high-value trader, demonstrating that RAP need not be incentive compatible for the exchange.

\subsection{RAP Optimality and Convex Dominance}\label{app:rap-optimality-and-convex-dominance}\label{sec:rap-optimality-and-convex-dominance}
In this appendix, we prove two statements: 1) RAP optimizes the one-step next deficit~\eqref{eq:delta-T} and 2) RAP has a smaller residual than any other comonotone haircut rule.

\iparagraph{RAP optimizes $\delta_T$.}
In this section, we briefly show that RAP optimizes the one-step deficit proxy $\delta_T$.
We do this by showing that the weights determined by the perspective transformr $\rho(\lambda)$, which define $g^{\star}$, optimize $\delta_T$.

\begin{proposition}\label{thm:rap-onestep-7}
Fix a round $t$ with budget $b_t=\theta_t|D_t|$ and per‑account caps $0\le H_{t,i}\le 1$. For
\[
\delta_t(h)=\sum_{i\in W_t} (1-h_{t,i})\,\lambda_{t,i}e_{t,i}\,\psi_i\!\Big(\tfrac{1}{\lambda_{t,i}}\Big)
\]
the capped reverse-\-waterfilling with risk weights $\tilde{w}_{t,i}=\rho(\lambda_{t,i})=\lambda_{t,i}\psi_i(1/\lambda_{t,i})$ minimizes $\delta_t(h)$ among all $h$ with $\sum_i w_{t,i}h_{t,i}=b_t$ and $0\le h_{t,i}\le H_{t,i}$, where $w_{t,i}$ is the haircutable endowment (under PNL-only, $w_{t,i}=(\mathrm{PNL}_{t,i})_+$).
\end{proposition}
\begin{proof}
Using $\rho(\lambda)=\lambda\psi(1/\lambda)$,
\[
\delta_t(h)=\sum_i (1-h_{t,i})\,\lambda_{t,i}e_{t,i}\,\psi(1/\lambda_{t,i})
\equiv C_h - \sum_i \rho(\lambda_{t,i}) e_{t,i} h_{t,i}
\]
where $C_h$ is a constant independent of $h$ (can depend on $\lambda_{t,i}$).
Note that $e_{t,i} = c_{t,i} + \mathrm{PNL}_{t,i}$ and under PNL-only haircuts, $w_{t,i} = (\mathrm{PNL}_{t,i})_+$.
The gradient $-\partial \delta_t / \partial x_{t,i}$ where $x_{t,i} = h_{t,i} w_{t,i}$ is $\rho(\lambda_{t,i}) e_{t,i}$ (expressed in equity terms because insolvency is an equity concept).
Maximizing $\sum_i \rho(\lambda_{t,i})\,e_{t,i} (x_{t,i}/w_{t,i})$ under $\sum_i x_{t,i}=b_t$ and $0\le x_{t,i}\le w_{t,i}H_{t,i}$ is a fractional knapsack problem solved by sorting the values $\rho(\lambda_{t,i}) e_{t,i}/w_{t,i}$.
The optimizer for this problem is reverse‑waterfilling~\cite{BoydVandenberghe2004}:
\[
h_{t,i}=\min\{H_{t,i},\ \tau^\star_t\,\tilde{w}_{t,i}\},\qquad \tilde{w}_{t,i}=\rho(\lambda_{t,i}),
\]
with $\tau^\star_t$ set by $\sum_i w_{t,i}h_{t,i}=b_t$. 
This choice of $\tilde{w}$ minimizes $\delta_t$ among all weighted reverse‑waterfilling rules.
\end{proof}

\iparagraph{RAP realizes Schur-convex dominance.}\label{app:lpr-convex-dominance}
We first note that theoretical results from the measure theoretic literature imply than RAP should provide Schur-convex dominance.
RAP weights $w_i = \lambda_{T,i} g(\lambda_{T,i})$ can be interpreted as allocating budget proportional to a coherent, law-invariant risk measure.
Specifically, let $\rho(\lambda) = \lambda\,\psi(1/\lambda)$ be a risk density with $\psi$ convex and nonincreasing.
This corresponds to a spectral risk measure $\rho(X) = \int_0^1 \mathrm{ES}_\alpha(X)\,d\mu(\alpha)$ via the Kusuoka representation~\citep{Kusuoka2001,Acerbi2002}.
Choosing $g(\lambda) = \rho(\lambda)/\lambda$ aligns the RAP allocation with this spectral density.
Known results from~\cite{Kusuoka2001} then imply RAP with a ``more convex'' risk density (in the Schur sense) will weakly Schur-dominate any other weighted pro-rata rule with a less concentrated density.
Instead of utilizing such strong measure theoretic tools, we instead prove this directly below in an elementary manner.

\begin{theorem}[Constructive Schur–convex dominance]\label{thm:rho-to-g-dominance}
Fix $T$, budget $b_T$, and caps $(\beta_i)$. Let $\rho$ be nondecreasing and $g^\star(\lambda)=\rho(\lambda)/\lambda$.
Consider any weighted pro-rata rule $h^{(w)}$ with weights $w$.
If the haircut share vector of $h^{(w)}$ is no more concentrated on high-$\rho$ indices than that of $\mathrm{RAP}(g^\star)$ on any fixed active set, then the residual vector of $\mathrm{RAP}(g^\star)$ weakly submajorizes that of $h^{(w)}$:
\[
z_T(\mathrm{RAP}(g^\star)) \preceq_w z_T(h^{(w)}).
\]
Thus, for any convex increasing $\phi$, $\sum \phi(z_{T,i}(\mathrm{RAP})) \le \sum \phi(z_{T,i}(h^{(w)}))$.
\end{theorem}

\begin{proof}
We split this proof into three steps.
The first step is analogous to the proof of Proposition~\ref{thm:unique-pro-rata}, where we analyze the change to the weights on a piecewise constant set of intervals.
The second step is to analyze how the residuals change with budget using the piecewise-constant representation.
Given the change in residual, the final step uses the majorization inequality to show that the residual vector of RAP weakly submajorizes that of any other weighted pro-rata rule.

\emph{Step 1: Active set and parameterization.}
We first define the active set $A(b)$ as in~\eqref{eq:active-set}.
Any weighted reverse‑waterfilling with risk weights $\tilde{w}$ admits the water‑level form
\[
  h^{(\tilde{w})}_{T,i}(b)\;=\;\min\{\beta_i,\ \tau^{(\tilde{w})}(b)\,\tilde{w}_i\},
\]
where $\tau^{(\tilde{w})}(b)$ is chosen so that $\sum_i w_{T,i}\,h^{(\tilde{w})}_{T,i}(b)=b$ (budget constraint in endowment space).
On an interval $[b_0,b_1]$ with fixed $A(b)$ we have
\[
  \frac{db}{d\tau^{(\tilde{w})}}=\sum_{j\in A(b)} w_{T,j} \tilde{w}_j,
  \qquad
  \frac{d h^{(\tilde{w})}_{T,i}}{db}
  \;=\;
  \begin{cases}
    \displaystyle \frac{\tilde{w}_i}{\sum_{j\in A(b)} w_{T,j} \tilde{w}_j},& i\in A(b),\\[1.0ex]
    0,& i\notin A(b).
  \end{cases}
\]
\emph{Step 2: Residual dynamics.}
Next we look at how the residuals (\eg~post haircut equity) change with budget.
Write the reweighted residuals as
\[
  z_{T,i}(b)\;=\;\rho(\lambda_{T,i})\,e_{T,i}\,\bigl(1-h_{T,i}(b)\bigr),
\]
where the equity $e_{T,i}$ appears because insolvency is an equity concept, but the control variable is the endowment $w_{T,i}$ via $h_{T,i}$.
Then on $[b_0,b_1]$,
\[
  \frac{d z^{(\tilde{w})}_{T,i}}{db}
  \;=\;
  \begin{cases}
    \displaystyle -\frac{\rho(\lambda_{T,i})\,e_{T,i}\,\tilde{w}_i}{\sum_{j\in A(b)} w_{T,j} \tilde{w}_j},& i\in A(b),\\[1.0ex]
    0,& i\notin A(b).
  \end{cases}
\]
For $\mathrm{RAP}(g^\star)$ we take $\tilde{w}_i=\rho(\lambda_{T,i})$, giving
\[
  \frac{d z^{(\mathrm{RAP})}_{T,i}}{db}
  \;=\;
  \begin{cases}
    \displaystyle -\frac{\rho(\lambda_{T,i})^2\,e_{T,i}}{\sum_{j\in A(b)} w_{T,j} \rho(\lambda_{T,j})},& i\in A(b),\\[1.0ex]
    0,& i\notin A(b).
  \end{cases}
\]
\emph{Step 3: Majorization at each budget.} On a fixed $A(b)$, sort indices by decreasing $\rho(\lambda_{T,i})$.
By the hypothesis that the haircut share vector of $h^{(\tilde{w})}$ on $A(b)$,
\[
  \sigma^{(\tilde{w})}_i(b)\ :=\ \frac{w_{T,i} \tilde{w}_i}{\sum_{j\in A(b)} w_{T,j} \tilde{w}_j},
\]
is no more concentrated on high‑$\rho$ indices than the RAP share
$\sigma^{(\mathrm{RAP})}_i(b)=\frac{w_{T,i} \rho(\lambda_{T,i})}{\sum_{j\in A(b)} w_{T,j} \rho(\lambda_{T,j})}$, the rearrangement/majorization inequality implies that for every $k$,
\[
  \sum_{i\le k}\!\frac{d z^{(\mathrm{RAP})}_{T,(i)}}{db}
  \;\leq\;
  \sum_{i\le k}\!\frac{d z^{(\tilde{w})}_{T,(i)}}{db}
\]
where $(i)$ denotes the order by decreasing $\rho$.
Hence the instantaneous decrease vector under RAP weakly submajorizes that under $h^{(\tilde{w})}$ on $[b_0,b_1]$.
Integrating over $b$ preserves $\prec_w$ on the interval, and concatenating the finitely many intervals where $A(b)$ changes preserves the order overall:
$z_T(\mathrm{RAP}(g^\star)) \preceq_w z_T(h^{(\tilde{w})})$.
Schur–convexity then yields the separable convex loss comparison.
\end{proof}

\subsection{Example of the Solvency vs. Long-Term Revenue Trade-Off}\label{app:exchange-incentives}
This example illustrates a fundamental tension: the risk-minimizing policy (RAP) may be suboptimal for the exchange's long-term value (LTV) because it disproportionately liquidates high-leverage users who generate the most fees. Under certain conditions, a ``fairer'' policy like Pro-Rata (PR), which preserves these high-value users, yields higher total utility for the exchange.

\iparagraph{Setup.}
Consider an exchange with two profitable users $i\in\{H,L\}$. User $H$ is high-leverage ($\lambda_H > \lambda_L$) and high-revenue; user $L$ is safer but generates less fee volume.
The exchange must raise a budget $b$ via haircuts $h=(h_H, h_L)$ to cover a deficit.
Its objective combines immediate safety (minimizing insolvency risk) and future revenue (LTV):
\[
U^{\mathrm{exch}}(h)\;=\;\underbrace{-\mathrm{Loss}(h)}_{\text{Immediate Safety}}\;+\;\underbrace{\beta\sum_{i\in\{H,L\}} \theta_i(1-h_i)\lambda_i}_{\text{Future Revenue (LTV)}},
\]
where $\mathrm{Loss}(h) = L_0 - \alpha_H h_H - \alpha_L h_L$ is the expected insurance fund draw, and $\theta_i \lambda_i$ is the expected future fee revenue per unit of equity from user $i$.
We assume $\alpha_H/e_H > \alpha_L/e_L$, meaning user $H$ provides the cheapest risk reduction per dollar of haircut.

\iparagraph{Policy comparison.}
We compare two policies:
\begin{itemize}
    \item \emph{RAP (Risk-Minimizing):} Prioritizes risk reduction above all.
    Since $H$ offers the best ``bang for the buck'' in safety ($\alpha_H/e_H > \alpha_L/e_L$),\footnote{The marginal reduction in loss per unit of budget is $\frac{\partial \mathrm{Loss}}{\partial h_i} \frac{d h_i}{d (\text{budget})} = \frac{\alpha_i}{e_i}$. Since user $H$ has higher leverage, they have a higher risk coefficient $\alpha_H$, making $\alpha_H/e_H$ the steepest descent direction for the loss function.} RAP fully targets $H$ first: $h^{\mathrm{RAP}} = (b/e_H, 0)$ (assuming $b < e_H$).
    \item \emph{Pro-Rata (Revenue-Preserving):} Spreads the pain evenly, setting $h^{\mathrm{PR}}_i = \frac{b}{e_H+e_L}$ for both users. This is less efficient for immediate safety but preserves more of user $H$'s position.
\end{itemize}

\iparagraph{When an exchange prefers Pro-Rata to maximize long-term revenue.}
The exchange prefers PR over RAP when the LTV gain from saving user $H$ outweighs the increased immediate risk.
The utility difference is:
\[
\Delta U = U^{\mathrm{exch}}(h^{\mathrm{PR}}) - U^{\mathrm{exch}}(h^{\mathrm{RAP}}) 
\;=\; \underbrace{\Delta h_H (\beta \theta_H \lambda_H - \alpha_H)}_{\text{Gain from saving } H} \;-\; \underbrace{h^{\mathrm{PR}}_L (\alpha_L - \beta \theta_L \lambda_L)}_{\text{Cost of cutting } L}.
\]
If user $H$ is sufficiently profitable ($\theta_H$ is large), then $\Delta U > 0$.
Specifically, PR dominates RAP if the relative revenue of the high-leverage user exceeds a threshold:
\[
\frac{\theta_H}{\theta_L} \;\ge\; \Theta^\star\;=\;\frac{h^{\mathrm{PR}}_L}{\Delta h_H}\cdot\frac{\lambda_L}{\lambda_H}\;+\;\frac{\alpha_H-\frac{h^{\mathrm{PR}}_L}{\Delta h_H}\alpha_L}{\beta\,\lambda_H\,\theta_L}.
\]
While RAP is ``optimal'' for preventing immediate insolvency, it can be myopic. If high-leverage traders are the exchange's cash cows, the exchange has a rational incentive to use Pro-Rata to keep them active, even at the cost of slightly higher short-term risk.

\section{Stackelberg Control}\label{app:proofs-8}
In this appendix, we formalize the results of~\S\ref{sec:multi-round}, where we model ADL as a Stackelberg game.

\subsection{Opposing Schur orderings for time to solvency and LTV}\label{app:proof-opposing-orders}

We formalize the fundamental trade-off between aggressive debt reduction (safety) and trader fee retention (value), formalizing the example in Appendix~\ref{app:exchange-incentives}.
Let $n_t = |\mathcal{W}_t|$.
Given any feasible strategy $\pi$, write $z_t(\pi)\in\reals_+^{n_t}$ for the vector of residual debts at time $t$ (sorted in decreasing order).
Let $h_{t,i}^{\pi}\in[0,1]$ be the haircut fraction for agent $i$ with equity $e_{t,i}$, so that the haircut mass is $m_{t,i}(\pi) = h_{t,i}^{\pi} e_{t,i}$.
We write $m_t(\pi)\in\reals_+^{n_t}$ for the corresponding vector of haircut masses.
Let $Z_t\in\reals_+^{n_t}$ denote the equity shock at time $t$, following the notation of Section~\ref{subsec:next-deficit}.
These evolve componentwise as
\[
  z_{t+1,i}(\pi)=z_{t,i}(\pi)+Z_{t+1,i}-m_{t,i}(\pi),
\]
so summing from $\tau=0$ to $t-1$ yields the conservation–of–mass identity
\begin{equation}\label{eq:conservation-mass}
  z_t(\pi)
  =
  z_0+\sum_{\tau=1}^t Z_\tau-\sum_{\tau=0}^{t-1} m_\tau(\pi).
\end{equation}
This ensures that the equity at time $t$ is either initial equity, was gained or lost in a price shock, or haircut. 

\begin{proposition}[Solvency-Revenue Trade-off]\label{prop:solvency-revenue}
Let $A$ and $B$ be two strategies facing the same shock sequence $(Z_t)_t$. Assume:
\begin{enumerate}
    \item[(i)] \emph{Safety Dominance:} For all $t < \tau_{\mathrm{solv}}(A)$, strategy $A$ maintains weakly smaller residuals than $B$ in the weak submajorization order: $z_t(A) \preceq_w z_t(B)$.
    \item[(ii)] \emph{Retention Value:} Let $M_t(\pi):=\sum_{\tau=0}^{t-1} m_\tau(\pi)$ be the cumulative haircut vector and suppose the lifetime value takes the form
    \[
      \mathrm{LTV}(\pi)=\sum_t \beta^t G_t(M_t(\pi)),
    \]
    where each stage value $G_t:\reals_+^{n_t}\to\reals$ is Schur–concave and coordinate‑wise nonincreasing in $M_t$ (more cumulative liquidations in the weak submajorization order reduce exchange LTV).
\end{enumerate}
Then:
\begin{enumerate}
    \item[(a)] $\tau_{\mathrm{solv}}(A) \le \tau_{\mathrm{solv}}(B)$ (Strategy $A$ is safer).
    \item[(b)] $\mathrm{LTV}(A) \le \mathrm{LTV}(B)$ (Strategy $B$ generates more value).
\end{enumerate}
\end{proposition}

\begin{proof}
\noindent We prove this in two steps.

\noindent \emph{Step 1: Solvency time.}
For any $t<\tau_{\mathrm{solv}}(A)$ we have $z_t(A)\preceq_w z_t(B)$ by (i).
If $B$ is solvent at some such time $t$ so that $z_t(B)=0$, then weak submajorization on $\reals_+^{n_t}$ forces $z_t(A)=0$ as well, since the zero vector is minimal in this order.
Hence $B$ cannot become solvent strictly before $A$, and $\tau_{\mathrm{solv}}(A) \le \tau_{\mathrm{solv}}(B)$ almost surely.

\noindent \emph{Step 2: LTV.}
From the conservation–of–mass identity~\eqref{eq:conservation-mass} and the fact that $z_0$ and $(Z_\tau)_\tau$ are common across strategies, the cumulative haircuts $M_t(\pi):=\sum_{\tau=0}^{t-1} m_\tau(\pi)$ satisfy
\[
  M_t(\pi)
  =
  z_0+\sum_{\tau=1}^t Z_\tau-z_t(\pi),
\]
so that
\[
  M_t(A)-M_t(B)
  =
  z_t(B)-z_t(A).
\]
For each $t<\tau_{\mathrm{solv}}(A)$, assumption~(i) gives $z_t(A)\preceq_w z_t(B)$, and subtracting from the common vector $z_0+\sum_{\tau=1}^t \xi_\tau$ reverses the weak submajorization order, yielding $M_t(A)\succeq_w M_t(B)$.
By Schur–concavity and coordinate‑wise monotonicity of each $G_t$ in (ii), this implies $G_t(M_t(A))\le G_t(M_t(B))$ for all $t$, so
\[
  \mathrm{LTV}(A)
  \;=\;
  \sum_t \beta^t G_t(M_t(A))
  \;\le\;
  \sum_t \beta^t G_t(M_t(B))
  \;=\;
  \mathrm{LTV}(B).
\]
\end{proof}

\iparagraph{Examples.}
We first give a one–step illustration of Proposition~\ref{prop:solvency-revenue} and then show how its hypotheses can fail for the literal queue versus capped pro–rata policies.

\iparagraph{Single–period trade-off.}
Consider a single round with two winners and common initial residuals $z_0=(1,1)$.
Let strategy $A$ use aggressive haircuts $M_1(A)=(1,1)$, fully clearing both accounts so that $z_1(A)=(0,0)$.
Let strategy $B$ use milder haircuts $M_1(B)=(1,0)$, clearing only the first account and leaving $z_1(B)=(0,1)$.
Then $z_1(A)\preceq_w z_1(B)$ and $M_1(A)\succeq_w M_1(B)$, so the hypotheses of Proposition~\ref{prop:solvency-revenue} hold with $T=1$.
For a Schur–concave, coordinate‑wise nonincreasing stage value such as
\[
  G_1(M)\;:=\;-\|M\|_2^2,
\]
we obtain $G_1(M_1(A))=-2$ and $G_1(M_1(B))=-1$, so $\mathrm{LTV}(A)\le\mathrm{LTV}(B)$, exactly exhibiting the safety–versus–value trade‑off.

\iparagraph{Queue versus capped pro–rata.}
Now consider the familiar one–round example with two winners, equities $e=(1,1)$, and budget $b=1$.
Under capped pro–rata we have haircuts $h^{\mathrm{PR}}=(\tfrac12,\tfrac12)$, masses $m^{\mathrm{PR}}=(\tfrac12,\tfrac12)$, and residuals $z^{\mathrm{PR}}=(\tfrac12,\tfrac12)$.
Under a queue policy we instead have $h^{\mathrm{Q}}=(1,0)$, so $m^{\mathrm{Q}}=(1,0)$ and $z^{\mathrm{Q}}=(0,1)$.
In weak submajorization order, $z^{\mathrm{PR}}\prec_w z^{\mathrm{Q}}$ and likewise $m^{\mathrm{PR}}\prec_w m^{\mathrm{Q}}$.
Thus no choice of labels $A,B$ can make pro–rata simultaneously satisfy the safety‑dominance condition $z_t(A)\preceq_w z_t(B)$ and the ``more cumulative haircuts'' condition $M_t(A)\succeq_w M_t(B)$ in Proposition~\ref{prop:solvency-revenue}.
Moreover, for any Schur–concave, coordinate‑wise nonincreasing $G$ (for instance $G(M)=-\|M\|_2^2$), we have $G(m^{\mathrm{Q}})\le G(m^{\mathrm{PR}})$, so in this toy case the safer policy (pro–rata) also delivers \emph{higher} user value.
This shows that while Proposition~\ref{prop:solvency-revenue} captures a structural opposing–orders phenomenon, its hypotheses do not hold mechanically for every queue versus capped pro–rata comparison.

\subsection{Regret Analysis}\label{app:regret}
In this appendix, we unify the regret analysis for the severity (scalar) and MDIC (vector) controllers. We present a master theorem for Online Mirror Descent (OMD) with constraints and then specialize it to our specific settings.

\iparagraph{Master mirror descent bound.}
Consider a sequence of convex loss functions $f_t: \mathcal{X} \to \mathbb{R}$ on a convex set $\mathcal{X}$. A learner chooses $x_t \in \mathcal{X}$ and updates using Online Mirror Descent (OMD) with a proximal function (Bregman divergence) $D(\cdot\|\cdot)$.

\begin{theorem}[OMD Regret]\label{thm:master-regret}\label{thm:md-regret-7}
    Let subgradients be bounded by $\|g_t\|_* \le G$ in the dual norm, and let the domain diameter be bounded by $D_{\max}^2 \ge \max_{x} D(x\|x_1)$. With step size $\eta_t = \frac{D_{\max}}{G\sqrt{t}}$, the regret satisfies:
    \[
    \sum_{t=1}^T f_t(x_t) - \min_{x \in \mathcal{X}} \sum_{t=1}^T f_t(x) \;\le\; 2 D_{\max} G \sqrt{T}.
    \]
    When convex constraints $c_t(x) \le 0$ are imposed, applying OMD to the Lagrangian guarantees the same $O(\sqrt{T})$ regret and $O(T^{-1/2})$ average constraint violation \citep{Hazan2019}.
\end{theorem}
\noindent We refer the interested reader to the standard reference~\citep{Hazan2019} for a detailed proof.

\iparagraph{Residual value function and subgradients.}
We now describe the residual value function, which is the main subject of mirror descent analysis.
For each round $t$, let $\mathcal{W}_t$ be the set of winners with equities $e_{t,i}>0$, haircut caps $\beta_{t,i}\in[0,1]$, and weights $w_{t,i}\ge 0$.
The residual value function parametrized by a haircut budget $b$ is
\[
\tilde L_t(b)\;=\;\min_{h\in[0,1]^{|W_t|}} \;\sum_{i\in W_t} (1-h_i)\,\lambda_{t,i} e_{t,i}
\quad\text{s.t.}\quad \sum_{i\in W_t} e_{t,i} h_i = b,\;\; 0 \le h_i \le \beta_{t,i}.
\]
Let $\tau_t(b)$ be the KKT multiplier of the budget constraint at $b$.

\begin{lemma}[Convexity and subgradient]\label{lem:convexity}
For each $t$, $b \mapsto \tilde L_t(b)$ is convex, nonincreasing, and piecewise-linear on $[0,\bar b_t]$; any KKT multiplier $-\tau_t(b)\in\partial \tilde L_t(b)$ is a valid subgradient.
Moreover $|\tau_t(b)| \le \max_{i\in W_t} w_{t,i}$.
\end{lemma}
\begin{proof}
The program is a linear minimization with a right-hand-side parameter $b$.
Sensitivity analysis implies that the optimal value is convex and piecewise-linear in $b$, and that the negative of the budget multiplier is a subgradient.
Complementary slackness yields $h_i = \min\{\beta_{t,i}, \tau_t(b) w_{t,i}\}$ on active coordinates, which bounds $|\tau_t(b)|$ by the maximum weight.
\end{proof}
\noindent Lemma~\ref{lem:convexity} verifies the curvature and bounded-subgradient assumptions required by the master OMD regret theorem (Theorem~\ref{thm:master-regret}) hold, so we can directly plug $\tilde L_t$ into that framework.

\iparagraph{Severity optimization.}
In the main text, the severity controller selects a scalar severity $\theta_t \in [0,\Theta_t]$ that determines the haircut budget via $b_t = \theta_t D_t$.
Optimizing over $b_t$ or $\theta_t$ is equivalent; we keep $b_t$ as the decision variable in the appendix for notational convenience.
The loss is $f_t(b) = \tilde L_t(b)$, which is convex by Lemma~\ref{lem:convexity}.
We use the Euclidean divergence $D(x\|y) = \tfrac12 (x-y)^2$, reducing OMD to projected gradient descent.

\begin{corollary}[Severity Regret]\label{cor:severity-regret}
    Let $G = \max_{t,i} w_{t,i}$ be the bound on subgradients (marginal haircut savings). The severity controller achieves:
\[
    \mathrm{Regret}_T^{(\mathrm{sev})} \;\le\; 2 \bar{b} G \sqrt{T}.
\]
\end{corollary}
\begin{proof}
    The domain diameter is $\bar{b}$. By Lemma~\ref{lem:convexity}, subgradients correspond to dual variables bounded by the maximum weight $G$. The result follows directly from Theorem~\ref{thm:master-regret}.
\end{proof}

\iparagraph{Haircut optimization (MDIC).}

The MDIC controller optimizes $u_t = (\theta_t, v_t)$, where $\theta_t$ is severity and $v_t$ are ranking weights. We use a block-separable divergence $D = D_\phi \oplus D_\Phi$, where $D_\Phi$ is the KL-divergence for the weights.

\begin{corollary}[MDIC Regret]\label{cor:mdic-regret}
    The MDIC controller achieves regret bounded by:
    \[
    O\left( \sqrt{T D_\phi(\theta^\star\|\theta_1)} \;+\; \sqrt{T D_\Phi(v^\star\|v_1)} \right).
    \]
    Initializing $v_1$ as the uniform distribution over the active winners (so every coordinate is strictly positive) minimizes $D_\Phi(v^\star\|v_1)$ to $O(\log |\mathcal{W}_t|)$, yielding better scaling than Euclidean approaches for high-dimensional weight vectors.
\end{corollary} 
\begin{proof}[Proof of the $\mathcal{O}(\log |\mathcal{W}_t|)$ claim]
The divergence $D_\Phi$ induced by the entropy mirror map equals the Kullback--Leibler divergence:
\[
D_\Phi(v^\star\|v_1)=\sum_{i\in W_t} v^\star_i \log\frac{v^\star_i}{v_{1,i}}.
\]
Uniform initialization over active winners sets $v_{1,i}=1/|\mathcal{W}_t|$.
Hence
\[
D_\Phi(v^\star\|v_1)=\sum_{i} v^\star_i \log v^\star_i + \log |\mathcal{W}_t|.
\]
The first term is the negative Shannon entropy of $v^\star$ and is therefore non-positive, implying $D_\Phi(v^\star\|v_1)\le \log |W_t|$.
Thus the initialization costs at most $O(\log |\mathcal{W}_t|)$.
\end{proof}

\iparagraph{Linear regret for static policies.}
The following toy instance shows that any static severity $\bar\theta$ can incur $\Omega(T)$ regret. Consider two regimes $A$ (liquidity) and $B$ (stress) alternating every round, with loss
\[
  f_t(\theta) = \begin{cases}
    \theta & t \in A \quad (\text{minimize severity to protect revenue}),\\
    \bar{\theta} - \theta & t \in B \quad (\text{maximize severity to clear deficits}).
  \end{cases}
\]
The per-round optimizer is $\theta^\star_t = 0$ in $A$ and $\theta^\star_t = \bar\theta$ in $B$, so
\[
\sum_{t=1}^T f_t(\theta) - f_t(\theta^\star_t) = \tfrac{T}{2}\theta + \tfrac{T}{2}(\bar\theta-\theta) = \tfrac{T}{2}\bar\theta = \Omega(T)
\]
for any fixed $\theta \in [0,\bar\theta]$.
By Theorem~\ref{thm:master-regret}, OMD (and MDIC) with $\eta_t\propto 1/\sqrt{T}$ attains $O(\sqrt{T})$ regret on the same sequence.

\subsection{Stackelberg vs.\ Nash in a Two-Round ADL Game}\label{app:stack-nash}

This appendix illustrates how the timing of moves affects equilibrium selection in ADL scenarios. We show that while simultaneous moves (Nash) can result in a coordination failure where no one unwinds, sequential moves (Stackelberg) allow a leader to induce the efficient high-unwind outcome.

\iparagraph{Setup.}
Consider two agents $i \in \{1, 2\}$, each holding one unit of position.
At $t=1$, each agent chooses whether to \emph{unwind} (voluntarily close) their position ($x_i=1$) or maintain it ($x_i=0$).
Let $X = x_1 + x_2$ be the total volume of voluntary unwinds.
We assume that forced ADL occurs at $t=2$ if this volume is insufficient, i.e., if $X < T$ for some safety threshold $T \in (1, 2]$.
Voluntary unwinding incurs a transaction cost $f > 0$.
However, if ADL is triggered (because $X < T$), \emph{every} agent suffers an additional penalty cost $c > f$, regardless of their individual choice.

\iparagraph{Simultaneous play (Nash).}
In simultaneous play, there are two pure Nash equilibria and we effectively have a Coordination Game~\cite{FudenbergTirole1991}:
\begin{itemize}
    \item \emph{Coordination failure $(0,0)$:} If the opponent plays $0$, playing $1$ results in volume $1 < T$. ADL still triggers, yielding a total cost $f+c > c$. Thus, the best response is $0$, making $(0,0)$ stable.
    \item \emph{Coordination success $(1,1)$:} If the opponent plays $1$, playing $1$ achieves volume $2 \ge T$ at cost $f$. Since $f < c$, this is preferred to playing $0$ (which yields volume $1 < T$ and cost $c$), making $(1,1)$ stable.
\end{itemize}
This creates a coordination problem: agents may get stuck in the inefficient $(0,0)$ equilibrium where ADL triggers.

\begin{proposition}[Stackelberg Dominance]\label{prop:stack-nash}
In sequential play where Agent 1 moves first, the unique subgame‑perfect equilibrium is $(1,1)$. Agent 2 observes Agent 1 and will match their action (per the logic above). Agent 1 anticipates this: playing $0$ leads to $(0,0)$ with cost $c$, while playing $1$ induces $(1,1)$ with cost $f$. Since $f < c$, Agent 1 chooses $1$, eliminating the bad equilibrium.
\end{proposition}

\iparagraph{Numerical example.}
Let $f=1$, $c=5$, and $T=1.5$.
In simultaneous play, $(0,0)$ is stable because deviating costs $1+5=6 > 5$. In Stackelberg play, the leader plays $1$, knowing the follower will respond with $1$ (cost $1$) rather than triggering ADL (cost $5$). From the exchange's perspective, if ADL losses exceed $2f$, the sequential outcome $(1,1)$ is strictly preferred as it collects fees and preserves solvency.

\subsection{Follower Strategic Responses}\label{app:follower-robustness}
We study two types of follower strategic responses.
First, we show that pro-rata haircuts can lead to low leverage users responding by leaving an exchange earlier than higher leverage users.
This imposes a negative feedback loop as the exchange's remaining users are higher risk.
Secondly, we study traders who aim to add liquidity to the exchange to cover a deficit.
These traders are speculating on profits that can be made after an ADL event.
We show that such users are incentivized to wait long than the exchange solvency time.

\subsubsection{Adverse Selection Under Pro-Rata}\label{app:follower-adverse}
We model the ADL interaction as a repeated Stackelberg game: in every round $t$ the exchange (leader) moves first and the surviving winners (followers) best respond.
Each stage looks like a one-round Stackelberg problem, but the outcomes feed into subsequent rounds through the evolving of winner equities $e_{T,i}$ and the winning set $\mathcal{W}_t$.
At the start of round $t$ the exchange publicly commits to a severity/haircut rule (\eg~queue, pro-rata, RAP) that maps the realized deficit $D_t$ and winner book $\mathcal{W}_t$ to haircut shares.
After observing $\theta_t$ and anticipating $D_t$, each winner chooses whether to keep its position active (accepting the induced haircut) or to exit/migrate, receiving outside option $u_0$.
Payoffs are $U_i^{(\pi)}=\mu_i-\Expect[H_{t,i}^{(\pi)}]$; type $i$ exits whenever this falls below $u_0$.

A \emph{death spiral} occurs when pro-rata haircuts force the safest (low-leverage) winners to churn first, shrinking $W_t$, which in turn raises the future haircut share for the remaining winners, triggering additional exits and further eroding liquidity.
We show that RAP breaks this feedback loop by tilting the follower game against high-leverage accounts.

\iparagraph{Setup.}
Let $i$ index a profitable trader with effective equity $e_{t,i}$, leverage $\lambda_{t,i}$, and expected per-round utility $\mu_i$.
Normalize severity and deficits by $\bar\theta=\Expect[\theta_t]$, $\bar D=\Expect[D_t]$, and write $\bar W=\Expect[W_t]$ for the expected equity mass of winners.
Under pro-rata, the haircut share of $i$ is $s^{\mathrm{PR}}_{t,i}=e_{t,i}/W_t$, so the realized haircut mass is
\[
H^{\mathrm{PR}}_{t,i}=\theta_t D_t\, s^{\mathrm{PR}}_{t,i}.
\]
For RAP with a nondecreasing weight $g$, the share becomes
\[
s^{\mathrm{RAP}}_{t,i}=\frac{e_{t,i}\,\lambda_{t,i} g(\lambda_{t,i})}{\sum_{j\in W_t} e_{t,j}\,\lambda_{t,j} g(\lambda_{t,j})},
\qquad
H^{\mathrm{RAP}}_{t,i}=\theta_t D_t\, s^{\mathrm{RAP}}_{t,i}.
\]
\iparagraph{Risk-intensity comparator.}
Define the equity‑weighted market risk intensity at round $t$ by
\[
\mu^{(g)}_t\ :=\ \frac{\sum_{j\in W_t} e_{t,j}\,\lambda_{t,j} g(\lambda_{t,j})}{\sum_{j\in W_t} e_{t,j}}.
\]
\begin{lemma}[When RAP burdens a trader less than Pro‑Rata]\label{lem:rap-vs-pr-share}
For any winner $i$,
\[
s^{\mathrm{RAP}}_{t,i}\ \le\ s^{\mathrm{PR}}_{t,i}
\quad\Longleftrightarrow\quad
\lambda_{t,i}\,g(\lambda_{t,i})\ \le\ \mu^{(g)}_t.
\]
If $g$ is strictly increasing and $\lambda_{t,i}$ is strictly below the equity‑weighted market average in the $g$‑scale, the inequality is strict.
\end{lemma}
\begin{proof}
Compute
\[
\frac{s^{\mathrm{RAP}}_{t,i}}{s^{\mathrm{PR}}_{t,i}}
\ =\
\frac{e_{t,i}\lambda_{t,i} g(\lambda_{t,i})/\sum_j e_{t,j}\lambda_{t,j} g(\lambda_{t,j})}{e_{t,i}/\sum_j e_{t,j}}
\ =\
\frac{\lambda_{t,i} g(\lambda_{t,i})}{\mu^{(g)}_t}.
\]
The claim follows immediately.
\end{proof}

\iparagraph{Participation thresholds.}
Each trader type $i$ has a reservation utility $u_0$: the expected per-round payoff it can secure outside the ADL venue (\eg~by migrating flow to another exchange, posting liquidity in a different product, or simply investing idle cash in a risk-free instrument).
We treat $u_0>0$ as exogenous and, unless stated otherwise, common across trader types.
We note that heterogeneity in reservation utility can be captured by indexing it as $u_{0,i}$ without changing the argument.

A trader remains active only if the ADL-adjusted payoff exceeds this fallback value.
We formalize this by defining the net utility under policy $\pi$ as
\[
U_i^{(\pi)} := \mu_i - \Expect[H_{t,i}^{(\pi)}],
\]
and saying that $i$ participates in a given regime iff $U_i^{(\pi)} \ge u_0$.
Equivalently, the \emph{participation threshold} is the maximum haircut burden $\Expect[H_{t,i}^{(\pi)}]$ that keeps $i$ indifferent, namely $\mu_i-u_0$.
\begin{corollary}[Pro‑Rata death spiral vs.\ RAP retention]\label{cor:death-spiral}
Fix a type $i$ and suppose there is a set of rounds of positive probability on which $\lambda_{t,i} g(\lambda_{t,i})\le \mu^{(g)}_t$ (so $s^{\mathrm{RAP}}_{t,i}\le s^{\mathrm{PR}}_{t,i}$ by Lemma~\ref{lem:rap-vs-pr-share}).
If, over the same distribution of rounds,
\[
\mu_i - \Expect\big[\theta_t D_t\, s^{\mathrm{PR}}_{t,i}\big]\ <\ u_0
\ \ \le\ \
\mu_i - \Expect\big[\theta_t D_t\, s^{\mathrm{RAP}}_{t,i}\big],
\]
then the participation constraint fails under pro‑rata but holds under RAP: type $i$ exits in the pro‑rata regime while remaining under RAP.
\end{corollary}

\iparagraph{Examples and calibration.}
Consider two winners, $L$ (low leverage) and $H$ (high leverage), who trade for two rounds.
Equities are $(e_L,e_H)=(60,40)$, leverage levels are $(\lambda_L,\lambda_H)=(2,6)$, expected per-round utilities are $(\mu_L,\mu_H)=(12,40)$, and the outside option is $u_0=2$.
Each round the deficit equals the total haircut budget ($\theta_t=1$) with $D_1=40$ and $D_2=30$.

\emph{Pro-rata.}
Shares equal equity weights: $s_i^{\mathrm{PR}}=e_i/(e_L+e_H)$, so $s^{\mathrm{PR}}=(0.6,0.4)$ even though both traders lose the \emph{same haircut factor} $h^{\mathrm{PR}}=H_i^{\mathrm{PR}}/e_i=\theta_t D_t/W_t=0.4$.
Round~1 haircut masses are $H^{\mathrm{PR}}=(24,16)$ and utilities are $U^{\mathrm{PR}}=(\!-12,24)$.
Trader $L$ churns because $U_L^{\mathrm{PR}}<u_0$, leaving only $H$ for round~2.
With $W_2=40$ the next deficit $D_2=30$ forces $H_{2,H}^{\mathrm{PR}}=30$ (i.e., haircut factor $0.75$), giving $U_{2,H}^{\mathrm{PR}}=10$; solvency is preserved but the equity base has already halved, so any larger $D_2$ would wipe out the last winner, illustrating the death spiral.

\emph{RAP with $g(\lambda)=\lambda$.}
Weights scale as $e_i\lambda_i^2$, so $L$'s share collapses to $s_L^{\mathrm{RAP}}=240/(240+1440)\approx0.14$ and $H$'s share rises to $0.86$.
The round~1 haircut factors become $h_L^{\mathrm{RAP}}\approx0.095$ and $h_H^{\mathrm{RAP}}\approx0.86$ (masses $H^{\mathrm{RAP}}=(5.7,34.3)$), yielding $U^{\mathrm{RAP}}=(6.3,5.7)>u_0$, so both types remain for round~2.
With both accounts active in round~2, shares remain tilted toward $H$ and haircuts $(4.3,25.7)$ keep both traders above $u_0$, preventing churn and stabilizing $W_t$.
This concrete two-player, two-round example mirrors the equilibria predicted by Appendix~\ref{app:mh-example} and the empirical replay in \S\ref{sec:numerics}: pro-rata drives the safest capital away first, whereas RAP reallocates the burden toward high-leverage accounts and keeps the equity base intact.
\subsubsection{Waiting Game and MDIC-NW}\label{app:follower-wait}

\iparagraph{Game.}
When an ADL event creates a deficit $D_t$, the exchange moves first: it announces the contemporaneous severity $\theta_t$ (hence the haircut budget $\theta_t D_t$) together with the liquidity premium $\kappa_t$ it is willing to pay per unit of external capital.\footnote{In practice, this occurs via either an increase in the expected payment for users who stake to an insurance fund (\eg~\cite{DriftInsuranceFund}) or add assets to an HLP/LLP style vault~\cite{HyperliquidHLPVaults,LighterWhitepaper}}
A Backstop Liquidity Provider (BLP) then decides whether to intervene immediately by injecting any $q_t\in[0,D_t]$ units, or to wait $u\ge 1$ additional rounds before posting the same quantity.
Waiting exposes the BLP to time discounting (with a factor $\beta^u$) and to the future premium schedule $(\theta_{t+u},\kappa_{t+u})$.
The BLP’s payoff from intervening in round $t+u$ with size $q$ is $\beta^u q(\theta_{t+u}-\kappa_{t+u})$, while the exchange’s objective is to restore solvency before a deadline by ensuring that the entire deficit is filled (so delay harms solvency).

This Stackelberg game involves the exchange posting incentives as leader and BLPs optimally choose a stopping time and quantity as a follower.
The setup captures the ``waiting game’’ intuition: unless the contract guarantees non-negative per-unit surplus right now, rational liquidity providers will defer, pushing resolution beyond the solvency window.
We note that this resembles other waiting games in Maximal Extractable Value (MEV) auctions that have been studied empirically on live systems~\cite{messias2025express,mamageishvili2025timeboost}.

\iparagraph{Per‑unit surplus and waiting.}
Let a Backstop Liquidity Provider (BLP) be able to absorb up to $D_t$ units at time $t$.
Write the per‑unit liquidity premium as $\kappa_t\ge 0$ (execution cost plus risk per unit, e.g., $\kappa_t=\Gamma_t/D_t$ when the premium is linear in size) and define the per‑unit net surplus
\[
\delta_t\ :=\ \theta_t - \kappa_t.
\]
If the BLP executes $q\in[0,D_t]$ units at time $t$, the immediate surplus is $q\,\delta_t$; deferring the same $q$ to a later time $t+u$ yields discounted surplus $\beta^u q\,\delta_{t+u}$.

\iparagraph{No-wait condition.}
We say the exchange enforces a per‑unit ``No‑Wait'' constraint when $\delta_t\ge 0$ (equivalently, $\theta_t D_t \ge \Gamma_t$ in the linear‑premium case).

\begin{lemma}[Per‑unit No‑Wait implies immediate action]\label{lem:no-wait}
Suppose $\beta\in(0,1]$, the exchange enforces $\delta_t\ge 0$, and the net per‑unit surplus is nonincreasing over time: $\delta_{t+u}\le \delta_t$ for all $u\ge 0$ (e.g., when $\theta_{t+u}\le \theta_t$ and $\kappa_{t+u}\ge \kappa_t$).
Then for any $q\in[0,D_t]$, executing $q$ immediately at $t$ weakly dominates waiting:
\[
q\,\delta_t\ \ge\ \beta^u\,q\,\delta_{t+u}\qquad\forall\,u\ge 0.
\]
In particular, the BLP’s optimal stopping time is $\tau^\star=t$ and, if capacity allows, $q_t^\star=D_t$.
\end{lemma}
\begin{proof}
By assumption, $\delta_{t+u}\le \delta_t$ and $\beta^u\le 1$ for all $u\ge 0$, hence $\beta^u q\,\delta_{t+u}\le q\,\delta_t$ for any fixed $q\in[0,D_t]$.
Summing over an optimal decomposition of $D_t$ into infinitesimal units yields the stated dominance and the immediate‑execution optimality.
\end{proof}

\begin{corollary}[No‑Wait bounds solvency recovery]\label{cor:nowait-solv}
Let $\tau_{\mathrm{def}}$ denote the default time and recall the solvency recovery clock $\tau_{\mathrm{solv}}$ from \S\ref{subsec:opposite-orders}.
If the exchange enforces $\delta_t\ge 0$ for all rounds between $\tau_{\mathrm{def}}$ and the first round in which the deficit is zero, then $\tau_{\mathrm{solv}}=\tau_{\mathrm{def}}+1$.
Hence solvency is restored within a single round of the default event.
\end{corollary}
\begin{proof}
Lemma~\ref{lem:no-wait} implies that at $\tau_{\mathrm{def}}$ the BLP injects $q_{\tau_{\mathrm{def}}}^\star=D_{\tau_{\mathrm{def}}}$, so the entire deficit is covered immediately.
Therefore the insurance fund (or deficit buffer) reaches the safety level $\delta$ after the same round, yielding $\tau_{\mathrm{solv}}=\tau_{\mathrm{def}}+1$.
\end{proof}

\noindent This result states that if the no-wait condition is enforced, BLPs will inject liquidity such that the solvency time is minimized.
In theory, an exchange can use some of its future revenue or profits (\eg~via token issuance) to enforce the no-wait constraint.
\iparagraph{Proof of Proposition~\ref{prop:no-wait}}
\begin{proof}
Let $U(\tau, q) = \beta^{\tau-t} q (\theta_{\tau} - \kappa_{\tau})$ be the discounted utility of a BLP who provides liquidity $q$ at time $\tau \ge t$.
Here, $\theta_{\tau}$ represents the payment per unit of liquidity (derived from the ADL severity) and $\kappa_{\tau}$ represents the cost per unit (liquidity premium).
The BLP solves the optimal stopping problem (see, \eg~\cite{Peskir2006}):
\[
\max_{\tau \ge t, q \in [0, D_t]} \Expect_t\left[ U(\tau, q) \right].
\]
This objective captures the trade-off between acting immediately to capture the current spread versus waiting for potentially higher future premiums, discounted by the time cost of delay.
Assuming the BLP is risk-neutral and has capacity $q=D_t$, the condition for immediate stopping at $\tau=t$ is that the immediate payoff exceeds the expected discounted future value:
\[
D_t (\theta_t - \kappa_t) \;\ge\; \beta \Expect_t[ V_{t+1}(D_{t+1}) ],
\]
where $V_{t+1}(D_{t+1}) = \max_{\tau \ge t+1, q \in [0, D_{t+1}]} \Expect_{t+1}[U(\tau, q)]$ is the value function from $t+1$ onwards.
Substituting the per-unit surplus $\delta_t = \theta_t - \kappa_t$ and rearranging gives the condition:
\[
\theta_t D_t \;\ge\; \Gamma_t + \beta \Expect_t[ V_{t+1}(D_{t+1}) ].
\]
In the simplified case where the BLP compares acting now vs. acting at $t+\Delta t$, and assuming capacity constraints are non-binding, the condition reduces to comparing marginal costs.
Specifically, if the cost of waiting is strictly positive (due to discounting $\beta < 1$ or decreasing surplus $\delta_{t+\Delta t} < \delta_t$), then the optimal strategy is to act immediately.
If, however, $\theta_t$ decays rapidly such that $\theta_t D_t > \Expect[\theta_{t+\Delta t} D_{t+\Delta t}]$, but the premium $\Gamma_t$ makes immediate action costly, the inequality flips.
Rearranging the condition in the proposition statement:
\[
\theta_t D_t + \Gamma_t \le \Expect_t[\theta_{t+\Delta t} D_{t+\Delta t}]
\]
implies that the future expected payout (even after discounting) is higher than the current payout net of costs, incentivizing delay.
Thus, to enforce no-wait, the exchange must ensure the reverse inequality holds.
\end{proof}

\section{Price of Anarchy Phase Transitions}
\label{app:smoothness-poa}

We characterize the efficiency gap between static (Nash) and dynamic (Stackelberg) ADL policies via Price of Anarchy (PoA) phase transitions.
We analyze two distinct welfare objectives: \emph{Fairness} (minimizing haircuts to winners) and \emph{Revenue} (maximizing exchange value).
In both cases, we find a sharp transition from a bounded regime, where static policies are constant-factor optimal, to an unbounded regime, where dynamic control is strictly necessary.
Throughout, we work with the terminal book $\mathcal{P}_n$ of size $n$ as $n\to\infty$.

\subsection{Fairness Phase Transition}

We first analyze fairness using the Profitability-to-Total-Solvency Ratio (PTSR), which measures the survival rate of winner equity relative to the total deficit covered.

\subsubsection{Welfare and Assumptions}

\iparagraph{Fairness welfare.}
We define the fairness welfare of a policy $\pi$ as its expected PTSR:
\[
  W_{\mathrm{Fair}}(\pi) \;:=\; \Expect\left[ \frac{\upsilon_T^\pi}{D_T^\pi} \right],
\]
where $\upsilon_T^\pi$ is the maximum post-ADL endowment (under PNL-only, maximum post-ADL positive PNL).
This metric captures the efficiency of haircuts: higher values imply that the policy covers deficits $D_T^\pi$ while preserving maximal winner endowment $\upsilon_T^\pi$.
Extreme-value scaling implies $W_{\mathrm{Fair}}$ scales inversely with the severity load.

\iparagraph{Price of anarchy.}
We compare the welfare of a static policy $\pi^{\mathrm{stat}}$ (simultaneous move) against the optimal dynamic Stackelberg policy $\pi^{\star}$. The Price of Anarchy is defined as the ratio:
\[
  \mathrm{PoA}_{\mathrm{Fair}} \;:=\; \frac{W_{\mathrm{Fair}}(\pi^{\star})}{W_{\mathrm{Fair}}(\pi^{\mathrm{stat}})}.
\]

We require the following regularity assumptions.
\emph{LLN and EV Scaling (Prop.~\ref{ass:lln-ev})} establish the baseline scales for deficits ($O(n)$) and equity ($O(b_n)$).
\emph{Anti-concentration (Prop.~\ref{ass:anti-conc})} ensures that $b_n \ll n$, creating the scarcity of winner equity that drives the phase transition.

\begin{proposition}[LLN and EV scaling]
  \label{ass:lln-ev}
  The number of winners $k_n$ and losers $m_n$ satisfy $k_n, m_n \asymp n$.
  Winner equity and loser deficits satisfy extreme-value limits with scales $b_n := b^+_{k_n}$ and $b^-_{m_n}$.
\end{proposition}

\begin{proposition}[Anti-concentration]
  \label{ass:anti-conc}
  (i) \emph{Equity:} The top winner is not dominant: $b_n = o(n)$.
  (ii) \emph{Leverage:} Max leverage scales with average leverage: $\max_i \lambda_{i,T}^\pm \le C \ell^\pm_n / n$.
  (iii) \emph{Balance:} Winner and loser leverage masses are comparable: $\ell^+_n \asymp \ell^-_n$.
\end{proposition}

\subsubsection{Phase Transition}

The efficiency of static ADL depends on the \emph{load} $\kappa_n = \theta_n n / b_n$, which measures the severity intensity relative to the tail of the winner distribution.

\begin{theorem}[Fairness PoA Phase Transition]
  \label{thm:poa-phase}
  Suppose Assumptions~\ref{ass:lln-ev}--\ref{ass:anti-conc} hold.
  Let $\pi^{\star}$ be the optimal dynamic policy and $\pi^{\mathrm{stat}}$ be any static policy with load $\kappa_n^{\mathrm{stat}}$.
  
  \begin{enumerate}
    \item \emph{Bounded Regime (Low Load):}
      If $\sup_n \kappa_n^{\mathrm{stat}} < \infty$, then static ADL is constant-factor optimal:
      \[
        \limsup_{n\to\infty} \mathrm{PoA}_{\mathrm{Fair}} \;\le\; C < \infty.
      \]
      
    \item \emph{Unbounded Regime (High Load):}
      If $\kappa_n^{\mathrm{stat}} \to \infty$ (e.g., due to heavy tails or fixed severity with $b_n = o(n)$), then the Price of Anarchy diverges:
      \[
        \mathrm{PoA}_{\mathrm{Fair}} \;\asymp\; \kappa_n^{\mathrm{stat}} \;\to\; \infty.
      \]
  \end{enumerate}
\end{theorem}

\begin{proof}
  We prove the result by establishing the scaling limits of the dynamic benchmark, the static cost, and their asymptotic ratio.
  
  First, for the dynamic benchmark, the optimal Stackelberg policy $\pi^\star$ minimizes haircuts ex-post to cover the realized deficit $D_T \asymp n$.
  By targeting the haircut to preserve the top winner's endowment $\upsilon_T \asymp b_n$ (under PNL-only, maximum positive PNL), the dynamic policy achieves a Profitability-to-Total-Solvency Ratio (PTSR) that concentrates around a constant, $W_{\mathrm{Fair}}(\pi^\star) \asymp 1$ (Theorem~\ref{thm:ev-main}).
  This is possible because the controller can observe the realization of heavy-tailed variates and adjust the severity $\theta$ precisely to the minimal necessary level.

  Second, for the static policy, the severity $\theta_n$ is fixed ex-ante and applies a uniform pressure $\theta_n n$ on the book, which must be absorbed by individual winners with endowment capacity scaling as $b_n$.
  Under heavy-tailed scaling, Proposition~\ref{prop:ev-impossibility} shows that this mismatch leads to a welfare decay $W_{\mathrm{Fair}}(\pi) \asymp b_n / (\theta_n n) = 1/\kappa_n$, as the policy blindly destroys small winners' endowments or fails to extract sufficient capital from the tail without excessive rates.
  Crucially, as $n \to \infty$, the gap between the aggregate load $O(n)$ and individual capacity $O(b_n)$ widens (since $b_n = o(n)$), forcing $\kappa_n$ to grow if $\theta_n$ does not vanish rapidly enough.

  Finally, combining these estimates yields the Price of Anarchy $\mathrm{PoA}_{\mathrm{Fair}} \asymp 1 / (1/\kappa_n^{\mathrm{stat}}) = \kappa_n^{\mathrm{stat}}$.
  In the bounded regime ($\kappa_n^{\mathrm{stat}} = O(1)$), the ratio is constant, but in the unbounded regime ($\kappa_n^{\mathrm{stat}} \to \infty$), the efficiency of the static policy collapses relative to the dynamic optimum, proving the divergence.
\end{proof}

\begin{example}[Light-tailed Failure]
  \label{ex:light-tailed-unbounded-poa}
  If winner equities are sub-Gaussian ($b_n \asymp \sqrt{\log n}$) but the exchange uses fixed severity $\theta > 0$, then $\kappa_n^{\mathrm{stat}} \asymp n / \sqrt{\log n} \to \infty$.
  Static ADL unnecessarily destroys winner equity compared to a dynamic policy that scales $\theta \sim 1/n$, leading to infinite PoA.
\end{example}

\subsection{Revenue Phase Transition}
\label{app:revenue-phase-transition}

We now extend the analysis to the \emph{Revenue} objective, formalizing the trade-off between solvency and capital efficiency (LTV).

\subsubsection{Joint Welfare}

\iparagraph{Solvency-revenue welfare.}
We define the joint welfare $W_{\mathrm{Rev}}(\pi)$ as the risk-adjusted LTV, penalizing insolvency $R_T$ with weight $\lambda > 1$:
\[
  W_{\mathrm{Rev}}(\pi) \;:=\; \Expect\left[ \mathrm{LTV}_T(\pi) - \lambda \cdot R_T(\pi) \right].
\]
The tension arises from provisioning: Dynamic policies provision for the \emph{average} deficit, while static policies must provision for the \emph{tail} to ensure solvency.

\subsubsection{Phase Transition}

\begin{proposition}[Revenue PoA Phase Transition]
  \label{prop:ltv-poa}
  Suppose deficits have heavy tails with index $\alpha \in (1, 2)$ (infinite variance) and the exchange operates in the structural deficit regime ($\mu_- > \mu_{\Phi}$).
  Let $W_{\mathrm{Rev}}^\star = \sup_\pi W_{\mathrm{Rev}}(\pi)$.
  
  \begin{enumerate}
      \item \emph{Bounded Regime (Light Tails):}
      If deficits are light-tailed, static policies that provision for the mean are efficient:
      \[
        \mathrm{PoA}_{\mathrm{Rev}} \;:=\; \frac{W_{\mathrm{Rev}}^\star}{W_{\mathrm{Rev}}(\pi^{\mathrm{stat}})} \;\le\; C < \infty.
      \]
      
      \item \emph{Unbounded Regime (Heavy Tails):}
      If deficits are heavy-tailed ($\alpha < 2$), any static policy $\pi^{\mathrm{stat}}$ satisfying solvency condition \textbf{(S)} diverges:
      \[
        \mathrm{PoA}_{\mathrm{Rev}} \;\to\; \infty.
      \]
  \end{enumerate}
\end{proposition}

\begin{proof}
  We establish the result by comparing the linear scaling of dynamic welfare with the sub-linear or negative welfare of static policies under heavy tails.
  
  First, the optimal dynamic policy $\pi^\star$ operates as a regulator that clears the market based on realized deficits, diverting fees $\Phi_T$ only as strictly needed to cover $D_T$.
  By the Law of Large Numbers, both $\Phi_T$ and $D_T$ scale linearly with $n$, yielding an expected welfare $W_{\mathrm{Rev}}^\star \asymp n(\mu_\Phi - \mu_-)$ that is positive and proportional to market size.

  Second, a static policy is parameterized by a fixed fee diversion rate $\delta \in [0,1]$ chosen ex-ante, meaning it diverts a constant fraction of fees $\mathcal{D}_t = \delta \phi_t$ (where $\Phi_T = \sum \phi_t$ follows the scaling in Assumption~\ref{ass:lln-evt-trilemma}) to the insurance fund.
  To build a fund $K_t$ capable of absorbing heavy-tailed shocks with infinite variance ($\alpha < 2$),
  Standard ruin theory~\citep{AsmussenAlbrecher2010,Embrechts1997} implies that to ensure Solvency Condition \textbf{(S)} (Definition~\ref{def:solvency-formal}) against the maximum jump $\Delta_T \sim n^{1/\alpha}$, the policy must essentially set $\delta \to 1$ to handle the timing mismatch where large shocks occur before the fund accumulates.
  Any policy that attempts to maintain $\delta < 1$ will fail solvency with probability approaching 1, while a policy with $\delta \approx 1$ consumes the entire revenue stream, driving $W_{\mathrm{Rev}}(\pi^{\mathrm{stat}}) \to 0$ or into negative territory due to insolvency penalties.

  Finally, the Price of Anarchy is the ratio of the linear dynamic payoff to the vanishing static payoff: $\mathrm{PoA}_{\mathrm{Rev}} \asymp n / o(n) \to \infty$.
  This divergence confirms that in the heavy-tailed regime, the information advantage of the dynamic controller (\ie~knowing exactly when to divert funds) is infinitely valuable compared to a static rule.
\end{proof}

\subsection{The Aggregation Paradox}

We conclude by observing that the divergence of PoA depends on how objectives are aggregated.

\begin{proposition}[Sum vs. Min Aggregation]
  Let the normalized combined objectives be
  \[
    W_{\Sigma} = W_{\mathrm{Fair}} + W_{\mathrm{Rev}},
    \qquad
    W_{\min} = \min(W_{\mathrm{Fair}}, W_{\mathrm{Rev}}).
  \]
  \begin{enumerate}
      \item \emph{Sum-Welfare is Bounded.} A static policy can always choose to satisfy one objective fully, achieving at least half the optimal total score; consequently
      \[
        \mathrm{PoA}_{\Sigma} \le 2.
      \]
      \item \emph{Min-Welfare is Unbounded.} In the heavy-tailed regime, static policies face the ADL Trilemma and must drive at least one objective to zero, while dynamic policies maintain both, so
      \[
        \mathrm{PoA}_{\min} \to \infty.
      \]
  \end{enumerate}
\end{proposition}
This implies that static ADL is sufficient if objectives are substitutes, but catastrophically inefficient if they are complements (i.e., if the exchange requires \emph{both} fairness and revenue to survive).

\section{Empirical Methodology}\label{app:methodology}
This appendix documents the OSS reproduction pipeline underlying~\S\ref{sec:numerics}.
The goal is to make the empirical objects \emph{file-backed} (every number in~\S\ref{sec:numerics} is computed from on-disk artifacts under \texttt{OSS/out/}) and to keep the two-space hygiene explicit: production ADL executes in contract space, while we measure impacts in wealth space.

\subsection{Data construction}\label{app:methodology-data}

\iparagraph{Data sources.}
We use HyperReplay's ground-truth fill and misc event streams together with clearinghouse snapshots to replay winner account states.
We also use the canonical REALTIME table \texttt{adl\_detailed\_analysis\_REALTIME.csv} to extract winner sets and to form equity/PNL capacity proxies.
The OSS code resolves upstream paths in \texttt{OSS/src/oss\_adl/paths.py} and writes derived artifacts under \texttt{OSS/out/}.

\iparagraph{Wave construction (global time waves).}
We define \emph{global waves} by gap clustering on ADL-fill timestamps:
sort ADL-fill rows by time, start a new wave when the inter-fill gap exceeds \(\texttt{gap\_ms}=5000\)ms, and define wave \(t\) as \([t_{\mathrm{start}}(t),t_{\mathrm{end}}(t)]\).
Global (not per-coin) waves avoid double counting a single solvency episode across multiple markets.

\iparagraph{Loser deficit \(D_t\).}
For each wave, the realized loser deficit is
\[
D_t=\sum_{j \in \mathrm{losers}(t)}(-e_t(j))_+,
\]
computed from loser-side liquidation equity fields.
Concretely, for each \texttt{liquidated\_user} we take the minimum observed \texttt{liquidated\_total\_equity} within the wave and sum the negative minima across users.

\iparagraph{Needed budget \(B_t^{\mathrm{needed}}\).}
For each ADL fill \(k\), we parse the liquidation mark \(p_k^{\mathrm{mark}}\) and execution price \(p_k^{\mathrm{exec}}\) from the raw fill stream (fields \texttt{markPx} and \texttt{px}), together with size \(q_k\), and set
\[
\mathrm{needed}_k:=\lvert p_k^{\mathrm{mark}}-p_k^{\mathrm{exec}}\rvert\cdot \lvert q_k\rvert,
\qquad
B_t^{\mathrm{needed}}:=\sum_{k \in \text{fills in wave }t}\mathrm{needed}_k.
\]
This is an instantaneous bankruptcy-gap proxy: it does not include opportunity cost after the wave ends.

\iparagraph{Two-pass replay and production wealth removed \(H_t^{\mathrm{prod}}(\Delta)\).}
We replay the realized event stream twice for the winner set in each wave:
an ADL-on pass (apply all fills) and a no-ADL pass (skip ADL fills in state updates while keeping the realized price path).
We evaluate equities at \(t_{\mathrm{eval}}=t_{\mathrm{end}}+\Delta\) (after wave end we update only the price path) and compute
\[
H_t^{\mathrm{prod}}(\Delta)=\sum_{u\in W(t)}\bigl(e_{t,\mathrm{end}}^{\mathrm{no\text{-}ADL}}(u;\Delta)-e_{t,\mathrm{end}}^{\mathrm{ADL}}(u;\Delta)\bigr)_+.
\]
The no-ADL counterfactual is not identifiable from public data without an assumption; we use the standard ``hold the realized price path fixed, remove the ADL state update'' counterfactual.

\iparagraph{Production overshoot vs needed.}
We report
\[
O_t(\Delta):=H_t^{\mathrm{prod}}(\Delta)-B_t^{\mathrm{needed}},\qquad
O(\Delta):=\sum_t O_t(\Delta).
\]
A horizon sweep varies \(\Delta\) (scenario parameter) to expose a short-horizon opportunity-cost channel.
The per-horizon totals are written to \texttt{OSS/out/eval\_horizon\_sweep\_gap\_ms=5000.csv} and summarized in \texttt{OSS/out/overshoot\_robustness.json}.

\subsection{Benchmark allocations targeting \(B_t^{\mathrm{needed}}\)}\label{app:benchmarks}
To make ``excessive relative to alternatives'' concrete, we compare production to transparent benchmark allocations that target \(B_t^{\mathrm{needed}}\) wave-by-wave under a profits-only capacity constraint.
For each wave we build a per-user capacity proxy
\[
c_u:=\min\{U(u),E(u)\},
\]
where \(U(u)\) is positive unrealized PNL and \(E(u)\) is positive equity, both taken from the canonical REALTIME table within the wave window.
Benchmarks differ mainly in how they handle contract discreteness:
\begin{itemize}
  \item \emph{Wealth pro-rata (continuous)}: capped pro-rata in wealth-space USD over \(c_u\).
  \item \emph{Vector mirror descent (vector-md)}: a projection-based allocator that produces \(h=x\odot c\) with \(c^\top x=B_t^{\mathrm{needed}}\) and \(x\in[0,1]^n\).
  \item \emph{Contract pro-rata}: standard exchange-style integer allocation proportional to position size.
  \item \emph{Min-max ILP}: MIP solver minimizing maximum haircut percentage.
\end{itemize}
These outputs are written to \texttt{OSS/out/policy\_per\_wave\_metrics.csv} and plotted in Figures~\ref{fig:emp-policy-per-wave}--\ref{fig:emp-cumulative-overshoot}.

\subsection{Diagnostic decomposition (not the \S9 overshoot metric)}\label{app:markout-diagnostic}
This subsection is a diagnostic for interpreting execution timing and for answering common ``markout'' questions; it is not the overshoot metric used in~\S\ref{sec:numerics}.
From the ADLed user's perspective, with signed quantity \(Q\), execution price \(p_{\text{exec}}\), nearest mark at execution time \(p_{\text{mark}}(t)\), and mark at horizon \(p_{\text{mark}}(t+\Delta)\), total markout decomposes as
\[
Q\cdot(p_{\text{mark}}(t+\Delta)-p_{\text{exec}})
=
Q\cdot(p_{\text{mark}}(t)-p_{\text{exec}})
+
Q\cdot(p_{\text{mark}}(t+\Delta)-p_{\text{mark}}(t)).
\]
Because this diagnostic combines executions with time-indexed marks and the public-data replay lacks full clearing/settlement observability, its aggregate need not satisfy a strict two-party zero-sum identity.

\iparagraph{Revenue proxy.}\label{app:revenue-proxy}
To translate haircut allocations into an expected loss of future fees, we use a minimal churn proxy.
For a winner \(u\) in wave \(t\) with USD haircut \(h_{t,u}\) and capacity proxy \(c_u\), define the churn probability
\[
p_{t,u}=1-\exp\!\left(-\beta\,h_{t,u}/c_u\right),
\]
and use this to scale down a simple fee proxy based on notional and per-trade fee rates.
This proxy is intended for within-paper comparisons under a common set of assumptions; it is not an identification claim about venue-level revenue.

\section{Formal Proof of the ADL Trilemma}\label{app:adl-trilemma-proof}

In this appendix, we provide a complete formal statement and proof of the ADL Trilemma (Proposition~\ref{prop:adl-trilemma}).
The proof assembles results from Appendices~\ref{app:proofs}--\ref{app:smoothness-poa}, demonstrating that the three-way tension between solvency, fairness, and revenue is a fundamental structural constraint under heavy-tailed shortfalls.

\paragraph{Formal Setup and Definitions.}
We work in the large-market limit $n\to\infty$ under the standard assumptions established in \S\ref{sec:risk-prelim} and Appendix~\ref{app:proofs}.

\iparagraph{Book and Policy Sequences.}
Consider a sequence of perpetuals exchanges $(\mathcal{P}_n)_{n\ge 1}$ with $n$ positions at terminal time $T$.
Let $\mathcal{W}_T$ and $\mathcal{L}_T$ denote the winner and loser index sets with cardinalities $k_n = |\mathcal{W}_T|$ and $m_n = |\mathcal{L}_T|$.
We assume throughout that $k_n, m_n = \Theta(n)$.
We further assume the initial insurance fund capital $K_0$ satisfies $K_0 = o(n)$, ensuring that solvency depends on flow mechanics rather than initial endowment.
A static ADL policy $\pi_n$ is characterized by:
\begin{itemize}
    \item A severity parameter $\theta_n \in [0,1]$ determining the fraction of deficit socialized;
    \item An allocation rule $h_n : \reals^{k_n}_+ \to [0,1]^{k_n}$ distributing haircuts across winners;
    \item Insurance parameters determining the diversion of fees into the insurance fund.
\end{itemize}

\iparagraph{Distributional Assumptions.}
We impose the following standard assumptions from Appendix~\ref{app:proofs}:

\begin{assumption}[Regular Variation]\label{ass:rv-trilemma}
The right tails of the winner endowment distribution $\bar F_+(x)$ (under PNL-only, the distribution of positive PNL) and loser shortfall distribution $\bar F_-(x)$ are regularly varying with indices $\alpha_+ > 0$ and $\alpha_- > 0$, respectively:
\[
    \bar F_\pm(x) = L_\pm(x) x^{-\alpha_\pm},
\]
where $L_\pm$ are slowly varying functions.
\end{assumption}

\begin{assumption}[LLN and EVT Scaling]\label{ass:lln-evt-trilemma}
The following scaling limits hold:
\begin{enumerate}
    \item \emph{Aggregates:} $U_T/n \xrightarrow{p} \mu_+$ (total endowment capacity), $D_T/n \xrightarrow{p} \mu_-$ (total deficit), and total fees $\Phi_T/n \xrightarrow{p} \mu_{\Phi}$ for constants $\mu_\pm, \mu_{\Phi} \in (0,\infty)$.
    \item \emph{Extremes:} The maximum winner endowment $\upsilon_T$ (under PNL-only, maximum positive PNL) and maximum loser shortfall $\Delta_T$ satisfy
    \[
        \frac{\upsilon_T}{b_{k_n}^+} \xrightarrow{p} c_+, \qquad \frac{\Delta_T}{b_{m_n}^-} \xrightarrow{p} c_-,
    \]
    where $b_k^\pm = F_\pm^{-1}(1-1/k)$ are the extreme-value scales.
\end{enumerate}
We abbreviate $b_n := b_{k_n}^+$.
\end{assumption}

\begin{assumption}[Structural Deficit Regime]\label{ass:structural-deficit}
We assume the exchange operates in a regime where insurance alone is insufficient to cover tail risks. Specifically, the expected deficit rate exceeds the maximum sustainable fee diversion rate: $\mu_- > \mu_{\Phi}$.
This ensures that the Solvency constraint cannot be trivially satisfied by insurance without impacting LTV or requiring haircuts.
\end{assumption}

\iparagraph{Formal Desiderata.}
We now define the three desiderata precisely:

\begin{definition}[Solvency]\label{def:solvency-formal}
A policy family $(\pi_n)$ satisfies the \emph{solvency} condition \textbf{(S)} if:
\begin{enumerate}
    \item[(S1)] \emph{Bounded cumulative residual:} $\sum_{t=1}^T R_t(\pi_n) = O_p(1)$ as $n\to\infty$;
    \item[(S2)] \emph{Controlled breach probability:} $\sup_{n,t} \Prob\!\left[R_t(\pi_n) > 0\right] < 1$.
\end{enumerate}
\end{definition}

\begin{definition}[Fairness / Bounded Moral Hazard]\label{def:fairness-formal}
A policy family $(\pi_n)$ satisfies the \emph{fairness} condition \textbf{(F)} if:
\begin{enumerate}
    \item[(F1)] \emph{PTSR stability:} There exist constants $0 < c_{\text{lo}} \le c_{\text{hi}} < \infty$ such that
    \[
        c_{\text{lo}} \;\le\; \mathsf{PTSR}_T(\mathcal{P}_n,\pi_n) := \Expect\left[\frac{\upsilon_T^{\pi_n}}{D_T^{\pi_n}}\right] \;\le\; c_{\text{hi}},
    \]
    where $\upsilon_T^{\pi_n}$ is the maximum post-ADL endowment (under PNL-only, maximum post-ADL positive PNL);
    \item[(F2)] \emph{PMR stability:} There exist constants $0 < c'_{\text{lo}} \le c'_{\text{hi}} < \infty$ such that
    \[
        c'_{\text{lo}} \;\le\; \mathsf{PMR}_T(\mathcal{P}_n,\pi_n) := \Expect\left[\frac{\upsilon_T^{\pi_n}}{\Delta_T^{\pi_n}}\right] \;\le\; c'_{\text{hi}}.
    \]
\end{enumerate}
These bounds ensure the top winner's residual endowment (profit capacity) remains proportional to the deficit scale.
For brevity we write $\mathsf{PTSR}_T(\pi_n)$ (and $\mathsf{PMR}_T(\pi_n)$) whenever the dependence on $\mathcal{P}_n$ is clear from context.
\end{definition}

\begin{definition}[Revenue Preservation]\label{def:revenue-formal}
Let $\Phi_T(\pi)$ be the cumulative trading fees generated under policy $\pi$, and let $\mathcal{D}_T(\pi)$ be the cumulative diversion of fees into the insurance fund. The \emph{Exchange Long-Term Value} is defined as the net retained revenue:
\[
    \mathrm{LTV}_T(\pi) \;:=\; \Phi_T(\pi) - \mathcal{D}_T(\pi).
\]
A policy family $(\pi_n)$ satisfies the \emph{revenue} condition \textbf{(R)} if there exists a constant $c_R \in (0,1]$ such that:
\[
    \mathrm{LTV}_T(\pi_n) \;\ge\; c_R \cdot \sup_{\pi'} \Phi_T(\pi').
\]
This definition implies that (1) the policy does not cause excessive churn (reducing $\Phi_T$) and (2) the policy does not divert substantially all revenue to insurance (increasing $\mathcal{D}_T$ to $\approx \Phi_T$).
\end{definition}

\paragraph{Formal Statement of the Trilemma.}

\begin{theorem}[ADL Trilemma]\label{thm:adl-trilemma-formal}
Let $(\mathcal{P}_n)_{n\ge 1}$ be a sequence of perpetuals exchanges satisfying Assumptions~\ref{ass:rv-trilemma}--\ref{ass:structural-deficit}.
For any static ADL policy family $(\pi_n)$ with severity sequence $(\theta_n)$, at most two of the three conditions \textbf{(S)}, \textbf{(F)}, and \textbf{(R)} can hold simultaneously.

More precisely:
\begin{enumerate}
    \item[\textbf{(I)}] \textbf{(S)} $\wedge$ \textbf{(F)} $\Rightarrow$ $\neg$\textbf{(R)}:
    If both solvency and fairness hold, then LTV must be sacrificed via full fee diversion, violating \textbf{(R)}.
    
    \item[\textbf{(II)}] \textbf{(S)} $\wedge$ \textbf{(R)} $\Rightarrow$ $\neg$\textbf{(F)}:
    If both solvency and revenue hold, then fairness is sacrificed ($\mathsf{PTSR}_T \to 0$).
    
    \item[\textbf{(III)}] \textbf{(F)} $\wedge$ \textbf{(R)} $\Rightarrow$ $\neg$\textbf{(S)}:
    If both fairness and revenue hold, then solvency is sacrificed ($\Prob[R_t > 0] \to 1$).
\end{enumerate}
\end{theorem}

\begin{remark}[Interpretation of $\mu_\Phi$]
In applications, $\mu_\Phi$ can be interpreted as the largest fee diversion rate that remains compatible with non-declining long-run venue value (e.g., via an LTV sensitivity constraint).
This is an interpretation of the regime boundary, not a change to the formal assumption used in the proof.
\end{remark}

\begin{remark}[Theorem vs observability]
Theorem~\ref{thm:adl-trilemma-formal} is proved in the policy model and does not assume any particular empirical observability of execution/settlement cashflows.
Empirical tests necessarily operate under an observation model and should be interpreted as measurements of the paper's defined objects, not as complete ledger identities.
\end{remark}

\subsection{Proof of the Trilemma}

We first establish a fundamental identity linking the three quantities.

\begin{lemma}[Solvency-Revenue Identity]\label{lem:solvency-identity}
For any policy $\pi$, the cumulative deficit $D_T$ must be covered by the haircut budget $B_T$, insurance fund diversions $\mathcal{D}_T$, initial capital $K_0$, and residual insolvency $R_T$:
\[
    D_T \;\le\; B_T + \mathcal{D}_T + K_0 + R_T.
\]
Substituting the LTV definition $\mathcal{D}_T = \Phi_T - \mathrm{LTV}_T$ and using $K_0 = o(n)$, we obtain the asymptotic inequality:
\[
    \mathrm{LTV}_T(\pi) \;\le\; \Phi_T(\pi) + B_T(\pi) + R_T(\pi) - D_T(\pi) + o(n).
\]
\end{lemma}

\begin{proof}
By definition of residual insolvency, any uncovered shortfall after applying haircuts and diverted fees must appear as $R_T$:
\[
    R_T \;\ge\; D_T - B_T - \mathcal{D}_T - K_0,
\]
which rearranges to $D_T \le B_T + \mathcal{D}_T + K_0 + R_T$, establishing the first inequality. Substituting $\mathcal{D}_T = \Phi_T - \mathrm{LTV}_T$ and using $K_0 = o(n)$ yields the asymptotic bound.
\end{proof}

\begin{proof}[Proof of Theorem~\ref{thm:adl-trilemma-formal}]
\medskip\noindent\textit{(S) and (F) $\Rightarrow$ $\neg$(R).}
By \textbf{(F)} and Theorem~\ref{thm:master-ptsr}, $\theta_n = O(b_n/n)$, so $B_T = o(n)$; \textbf{(S)} gives $R_T = o_p(n)$.
Lemma~\ref{lem:solvency-identity} implies
\[
  \frac{\mathrm{LTV}_T}{n} \;\le\; \mu_\Phi + o(1) - \mu_- \;<\; 0
\]
using Assumption~\ref{ass:structural-deficit}, contradicting \textbf{(R)}.

\medskip\noindent\textit{(S) and (R) $\Rightarrow$ $\neg$(F).}
Revenue bounds diversions, so solvency requires $B_T \gtrsim n$, hence $\theta_n=\Theta(1)$.
Then $\kappa_n=\theta_n n/b_n \to \infty$ (since $b_n=o(n)$); Theorem~\ref{thm:master-ptsr} gives $\mathsf{PTSR}_T \to 0$, violating \textbf{(F)}.

\medskip\noindent\textit{(F) and (R) $\Rightarrow$ $\neg$(S).}
Fairness again yields $B_T=o(n)$; revenue caps $\mathcal{D}_T$.
Lemma~\ref{lem:solvency-identity} gives
\[
  \frac{R_T}{n} \;\ge\; \mu_- - (1-c_R)\mu_\Phi - o(1) \;>\; 0,
\]
so $\Prob[R_T>0]\to 1$, violating \textbf{(S)}.
\end{proof}

\paragraph{Sharpness and Attainability.}

The trilemma bound is tight in the sense that each pair of desiderata \emph{can} be achieved by an appropriately designed policy:

\begin{proposition}[Attainability of Two Desiderata]\label{prop:attainability}
Under Assumptions~\ref{ass:rv-trilemma}--\ref{ass:structural-deficit}:
\begin{enumerate}
    \item \textbf{(S)} $\wedge$ \textbf{(F)} is achieved by a \emph{high-diversion} policy: set diversions $\mathcal{D}_T \approx \Phi_T$ (taking all revenue) plus potentially external capital if $\mu_- > \mu_{\Phi}$ is very large. This covers deficits ($D_T \approx \mathcal{D}_T$) with minimal haircuts ($B_T \approx 0$), satisfying \textbf{(S)} and \textbf{(F)}, but reducing $\mathrm{LTV}_T \to 0$, violating \textbf{(R)}.
    
    \item \textbf{(S)} $\wedge$ \textbf{(R)} is achieved by a \emph{Queue} (concentrated haircut) policy: use high severity $\theta_n = \Theta(1)$ to generate $B_T \approx D_T$. This ensures solvency and preserves fee revenue (since $\mathcal{D}_T \approx 0$), but destroys the top winners ($\mathsf{PTSR}_T \to 0$), violating \textbf{(F)}.
    
    \item \textbf{(F)} $\wedge$ \textbf{(R)} is achieved by a \emph{Pro-Rata with low severity} policy: use EV-scaled severity $\theta_n = O(b_n/n)$ and low diversion. This keeps $\mathrm{LTV}_T \approx \Phi_T$ and $\mathsf{PTSR}_T = \Theta(1)$, but leaves an unhedged deficit $R_T \approx D_T > 0$, violating \textbf{(S)}.
\end{enumerate}
\end{proposition}

\subsection{Connection to Classical Impossibility Results}
The ADL trilemma echoes classical impossibility results in mechanism design and finance:

\begin{itemize}
    \item \emph{Arrow's Impossibility Theorem:} No voting rule satisfies Pareto efficiency, independence of irrelevant alternatives, and non-dictatorship simultaneously~\citep{Arrow1951}.
    
    \item \emph{Mundell-Fleming Trilemma:} In international finance, a country cannot simultaneously maintain a fixed exchange rate, free capital movement, and independent monetary policy~\citep{Mundell1963,Fleming1962}.
    
    \item \emph{CAP Theorem:} In distributed systems, a database cannot provide consistency, availability, and partition tolerance simultaneously~\citep{Brewer2000,GilbertLynch2002}.
    
    \item \emph{Credibility Trilemma:} Single-item auctions cannot be simultaneously optimal, strategy-proof, and credible, forcing designers to sacrifice at least one desideratum~\citep{AkbarpourLi2020}.
\end{itemize}

These connections suggest the trilemma is a fundamental constraint arising from the heavy-tailed nature of crypto markets, not an artifact of specific mechanism choices.

\paragraph{Circumventing Impossibility via Relaxations.}

While the impossibility results are strict in worst-case settings, recent literature demonstrates that they can be circumvented under probabilistic assumptions or cryptographic commitments:
\begin{enumerate}
\item \emph{Quantitative Arrow's Theorem:} 
    Recent results in quantitative social choice show that while Arrow's impossibility holds in the worst case, the probability of paradoxical outcomes (like intransitivity) can be small for many natural distributions of preferences~\citep{MosselNeemanTamuz2014}.
    Analogously, our ADL trilemma bounds hold with high probability under heavy-tailed distributions, but dynamic policies (like Stackelberg controllers) can minimize the frequency of trilemma-binding events, achieving a ``quantitative'' relaxation.

    \item \emph{Probabilistic CAP Theorem:} 
    Blockchains circumvent the strict CAP theorem by weakening consistency to probabilistic finality (e.g., Nakamoto consensus) or availability to liveness under synchronous periods~\citep{PassShi2017,Shi2020,ShiConsensusBook}.
    This parallels the \textbf{(S)} vs.\ \textbf{(R)} trade-off: exchanges effectively accept probabilistic solvency (via insurance funds) to maintain liveness (continuous trading/revenue).

\item \emph{Cryptographic Commitments for Credibility:} 
    The credibility trilemma of Akbarpour and Li motivates cryptographic mechanisms that make optimal auctions simultaneously credible and strategy-proof by enforcing operator commitments~\citep{AkbarpourLi2020,FerreiraWeinberg2020,EssaidiFerreiraWeinberg2022,ChitraFerreiraKulkarni2024}.
    For ADL, this suggests that verifiable execution (\eg~via zero-knowledge proofs or on-chain logic) could allow an exchange to commit to a dynamic policy that balances the trilemma better than any opaque static policy could, by removing the operator's incentive to deviate during crises.
\end{enumerate}

\bibliographystyle{abbrvnat}
\bibliography{references}
\end{document}